\documentclass{article}

\usepackage{microtype}
\usepackage{graphicx}
\usepackage{subfigure}
\usepackage{booktabs} % for professional tables
\usepackage{float}
\usepackage{multirow}

\usepackage{hyperref}

\usepackage[accepted]{icml2024}

\usepackage{amsmath}
\usepackage{amssymb}
\usepackage{mathtools}
\usepackage{amsthm}

\usepackage[capitalize,noabbrev]{cleveref}

\usepackage{array}
\usepackage{graphicx} % Required for inserting images
\usepackage{hyperref}
\usepackage{natbib}
\usepackage{verbatim}
\usepackage{tikz-cd}
\usepackage{thm-restate}

\newcommand{\boundary}{\partial}

\newcommand{\dmax}{d_{\max}}
\newcommand{\downlap}{L^{\mathrm{down}}}
\newcommand{\proj}{\Pi}
\newcommand{\R}{\mathbb{R}}
\newcommand{\Rtot}{R_{\tot}}

\newcommand{\wpartial}{\widetilde\partial}

\DeclareMathOperator{\im}{im}
\DeclareMathOperator{\tot}{tot}

\DeclareMathOperator{\poly}{poly}

\definecolor{BetterGreen}{rgb}{0.0, 0.8, 0.4}
\definecolor{BetterBlue}{rgb}{0.38, 0.44, 1.0}

\theoremstyle{plain}
\newtheorem{theorem}{Theorem}[section]

\newtheorem{lemma}[theorem]{Lemma}
\newtheorem{corollary}[theorem]{Corollary}
\theoremstyle{definition}

\theoremstyle{remark}

\newcommand{\kharmonic}[2]{H_{#1}^{(#2)}}

\definecolor{ibm-blue}{HTML}{648FFF}
\definecolor{ibm-purple}{HTML}{785EF0}
\definecolor{ibm-magenta}{HTML}{DC267F}
\definecolor{ibm-orange}{HTML}{FE6100}
\newcommand{\first}[1]{{\textcolor{ibm-orange}{\textbf{#1}}}}
\newcommand{\second}[1]{{\textcolor{ibm-magenta}{\textbf{#1}}}}
\newcommand{\third}[1]{{\textcolor{ibm-purple}{\textbf{#1}}}}

\usepackage[textsize=tiny]{todonotes}

\icmltitlerunning{Biharmonic Distance of Graphs and its Higher-Order Variants}

\begin{document}

\twocolumn[
\icmltitle{Biharmonic Distance of Graphs and its Higher-Order Variants: Theoretical Properties with Applications to Centrality and Clustering}



\icmlsetsymbol{equal}{*}

\begin{icmlauthorlist}
\icmlauthor{Mitchell Black}{osu}
\icmlauthor{Lucy Lin}{osu}
\icmlauthor{Weng-Keen Wong}{osu}
\icmlauthor{Amir Nayyeri}{osu}
\end{icmlauthorlist}

\icmlaffiliation{osu}{School of Electrical Engineering and Computer Science, Oregon State University, Corvallis, Oregon, USA}

\icmlcorrespondingauthor{Mitchell Black}{blackmit@oregonstate.edu}

\icmlkeywords{Machine Learning, ICML}

\vskip 0.3in
]



\printAffiliationsAndNotice{} % otherwise use the standard text.

\begin{abstract}
Effective resistance is a distance between the vertices of a graph that is both theoretically interesting and useful in applications. We study a variant of effective resistance called the biharmonic distance~\cite{lipman2010biharmonic}. While the effective resistance measures how well-connected two vertices are, we prove several theoretical results suggesting that the biharmonic distance measures how important an edge is to the global topology of the graph. Our theoretical results connect the biharmonic distance to well-known measures of connectivity of a graph like its total resistance and sparsity. Based on these results, we introduce two clustering algorithms using the biharmonic distance. Finally, we introduce a further generalization of the biharmonic distance that we call the $k$-harmonic distance. We empirically study the utility of biharmonic and $k$-harmonic distance for edge centrality and graph clustering.
\end{abstract}

\section{Introduction}
\label{sec:introduction}

Many areas of machine learning, such as clustering \cite{SheMalik2000SpecClust}, semi-supervised learning \cite{ZhuEtal2003SemiSupHarmonic} and graph mining \cite{Chakrabarti2006survey}, require measuring the distance between two vertices in a graph. Simple distance metrics like the shortest-path distance only consider a single path between vertices but do not contain any information about the global topology of the graph. For many tasks, a better choice of a distance metric is one that accounts for all paths connecting a pair of vertices, thereby capturing aspects of the connectivity of the graph. One commonly used distance metric with this property is effective resistance \cite{kirchhoff1847resistance}.

\begin{figure}[t]
    \centering
    \includegraphics[height=1.75in]{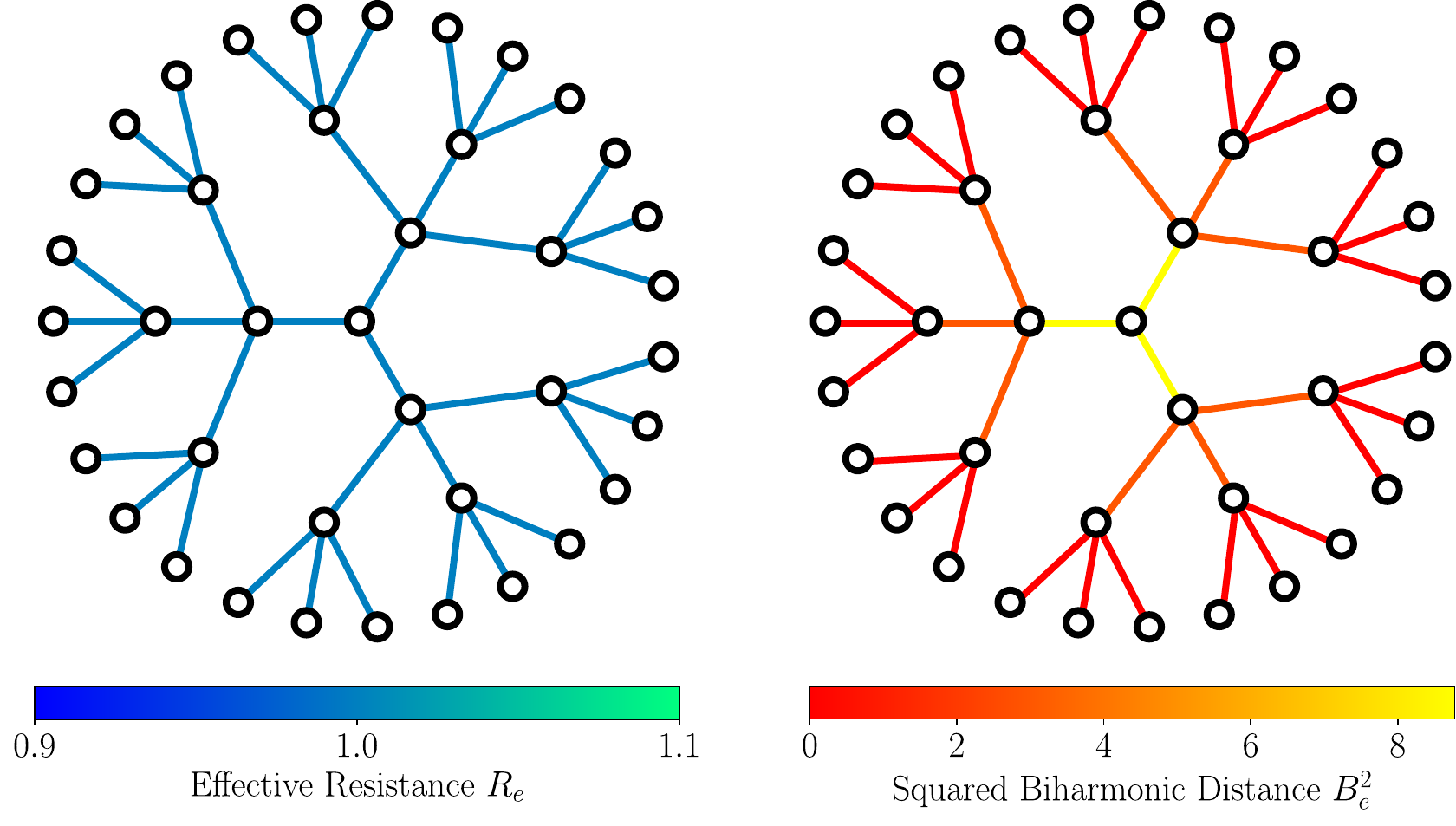}
    \caption{The effective resistance vs.~squared biharmonic distance for edges in a tree. While the effective resistance is 1 on all edges, the biharmonic distance is higher on edges closer to the root, demonstrating that biharmonic distance is aware of the global topology of a graph. See~\Cref{apx:examples} for more examples.}
    \label{fig:trees}
\end{figure}

The \textit{\textbf{effective resistance}} between vertices $s$ and $t$ in a graph is a distance measure defined
$$
    R_{st} = (1_s-1_t)^{T}L^{+}(1_s-1_t)
$$
where $L^{+}$ is the pseudoinverse of the graph Laplacian and $1_v$ is the indicator vector of a vertex $v$.
\par 
While effective resistance was originally introduced by \citet{kirchhoff1847resistance} in the context of electrical networks, it has since been discovered that by treating the mathematics abstractly, effective resistance has many interesting theoretical properties and applications in graph theory. 
An intuitive explanation of effective resistance is that 

\begin{center}
\textit{Effective resistance measure how well-connected a pair of vertices $s$ and $t$ are.}
\end{center}

because the more, and shorter, paths there are connecting $s$ and $t$, the lower the effective resistance between $s$ and $t$. Theoretical works on the effective resistance have reinforced this intuition. For example, it has been shown that effective resistance is proportional to the commute time of random walks between two vertices~\cite{chandra1996resistance} and is proportional to the probability an edge is in a random spanning tree~\cite{kirchhoff1847resistance}. Beyond its theoretical appeal, effective resistance is used in various applications, notably in contexts requiring a measure of connectivity, such as clustering~\cite{yen2005clustering, khoa2012large, AlevAnaryOveis2017GraphClustEffRes}, graph sparsification~\cite{SpielmanSrivastava2011SpecSpars}, edge centrality~\cite{teixeira2013spanning,mavroforakis2015spanning}, semi-supervised learning on graphs~\cite{ZhuEtal2003SemiSupHarmonic}, link prediction~\cite{lu2011link}, and as a positional encoding for graph neural networks~\cite{zhang2023rethinking, velingker2024affinity}. 

There are other graph distances similar to effective resistance. For example, the \textit{\textbf{diffusion distance}} (at time $t$)~\cite{coifman2006diffusion} is $D^{(t)}_{st}=\sqrt{(1_s-1_t)^{T}e^{-tL}(1_s-1_t)}$. Generally, we can consider spectral distances $\sqrt{(1_s-1_t)^{T} f(L)(1_s-1_t)}$ for any function $f$. However, it is an open question which function $f$ defines the best distance for different applications. Unfortunately, besides the effective resistance and diffusion distance, spectral distances are not well-studied and little is known about their theoretical properties or applicability.
\par 
\textit{\textbf{Biharmonic distance}} is one such spectral distance, defined 
$$
    B_{st} = \sqrt{(1_s-1_t)^{T} (L^{+})^{2} (1_s-1_t)}.
$$
The biharmonic distance was introduced by \citet{lipman2010biharmonic} in the context of geometry processing, where it was noted that the biharmonic distance was seemingly more aware of the global structure of a graph than effective resistance. Subsequent works have begun to use it in other applications, including consensus networks~\citep{yi2018consensus,yi2021consensus} and as a centrality measure~\citep{li2018kirchhoff, yi2018biharmonic}. While researchers have begun to study the theoretical properties of the biharmonic distance~\citep{lin2022biharmonic, wei2021biharmonic}, biharmonic distance is not nearly as well-understood theoretically as effective resistance.

\subsection{Contributions}

\begin{table*}[th]
    \centering
    \begin{tabular}{|p{0.2\linewidth}|p{0.25\linewidth}|p{0.45\linewidth}|}
        \hline
        & \textbf{Effective Resistance} & \textbf{Biharmonic Distance}  \\
        \hline
        Definition & $R_{st} = (1_s-1_t)^{T}L^{+}(1_s-1_t)$ &  $B_{st}=\sqrt{(1_s-1_t)^{T}L^{2+}(1_s-1_t)}$ \\
        \hline 
        Intuitive Explanation & \centering\textit{Measures how well-connected the vertices $s$ and $t$ are.}  & \centering \textit{Measures how important the edge $\{s,t\}$ is to the global topology of the graph.} \cr
        \hline 
        Bounds & $\frac{1}{n} \leq R_{st} \leq n-1$ & $\frac{1}{n^{2}}\leq B^{2}_{st}\leq n^{3}$ \hfill (\Cref{thm:lower_bound_biharmonic,thm:upper_bound_biharmonic})
        \\
        \hline 
        Bounds on Edges & $R_{e} \leq 1$ & $B^{2}_{e}\leq n$ \hfill (\Cref{thm:upper_bound_biharmonic_edge})
        \\
        \hline 
        Sum over Edges & $\sum_{e\in E} w_eR_e = n-1$ &  $n\sum_{e\in E} w_eB^{2}_{e} = \Rtot$ \hfill (\Cref{cor:biharmonic_foster})\\
        \hline 
        Electrical Flows & $R_{st} = \sum_{e\in E} f_{st}^{2}(e)/w_e$ & $nw_eB_{e}^{2} = \sum_{s,t\in V} f_{st}(e)^{2}/w_e$ \hfill (\Cref{thm:biharmonic_distance_is_squared_electrical_flow}) \\
        \hline 
        Cut Edges & $R_e = 1$ & $B^{2}_e=(\text{Cut Sparsity})^{-1}$ \hfill (\Cref{thm:biharmonic_of_cutedge_and_sparsity})\\
        \hline 
        Edges \& Sparsity & - & $B^{2}_e\sim (\text{Cut Sparsity})^{-1}$ \hfill (\Cref{thm:sparse_cut_implies_large_biharmonic,thm:large_biharmonic_edge_implies_sparse_cut})\\
        \hline
    \end{tabular}
    \caption{A comparison of the effective resistance and biharmonic distance.}
\end{table*}

We present several new theoretical properties of the biharmonic distance\footnotemark. Our results mainly concern the biharmonic distance of \textit{edges} in a graph. We propose the following intuitive explanation: 
\footnotetext{While the biharmonic distance $B_{st}$ has the property of being a metric, all of our theoretical results are more naturally expressed using the \textbf{squared} biharmonic distance $B^{2}_{st}$, suggesting this is the more interesting measure, even though it is not a metric.}

\begin{center}
\textit{Biharmonic distance measures how important an edge is to the global topology of a graph.}
\end{center}

We present several theoretical results that reinforce this intuition. Accordingly, we believe that the biharmonic distance may be a superior choice over other distances like effective resistance for applications where we need to capture the global topology of a graph with respect to the edges. 

\paragraph{Electrical Flow and Edge Centrality.} We prove that the biharmonic distance of an edge is proportional to its total usage in electrical flows between all pairs of vertices in the graph (\Cref{thm:biharmonic_distance_is_squared_electrical_flow}).
This theorem suggests that edges with high biharmonic distance are important to the connection between many pairs of vertices. This also partially explains the success of biharmonic distance as a measure of centrality. Our theorem also implies a generalization of the well-known Foster's Theorem to biharmonic distance.

\paragraph{Clustering and Sparse Cuts.} We establish a connection between the biharmonic distance of edges and the sparsity of cuts in the graph. In particular, we prove that sparse cuts necessarily contain edges with high biharmonic values, and conversely, edges with high biharmonic values are contained in sparse cuts.
(\Cref{thm:sparse_cut_implies_large_biharmonic,thm:large_biharmonic_edge_implies_sparse_cut}).
This result suggests the use of biharmonic distance in graph clustering algorithms, an idea we explore in~\Cref{sec:application_clustering}

\paragraph{Higher-Order Harmonic Distances.} We introduce a generalization of the biharmonic distance called the \textit{\textbf{k-harmonic distance}}, defined 
$$
    \kharmonic{st}{k} = \sqrt{(1_s-1_t)^{T} L^{k+} (1_s-1_t)}.
$$
Some of our theoretical results for biharmonic distance generalize to $k$-harmonic distance. 

\paragraph{Experiments.} We compare biharmonic and $k$-harmonic distance to other centrality measures in terms of correlation and resilience to change. We also compare our clustering algorithms with other clustering algorithms. We conducted experiments on both synthetic and real world data and observed that algorithms using biharmonic distance consistently outperform effective resistance methods. Furthermore, we observed that $k$-harmonic clustering algorithms achieve optimal results at larger values of $k$. 

\subsection{Related work} 

\paragraph{Biharmonic Distance.}
Previous works have also hinted at our proposed interpretation of biharmonic distance. These works establish a connection between the \textit{\textbf{total resistance}} $\Rtot$ (aka the \textit{\textbf{Kirchhoff index}}) of a graph and the biharmonic distance of an edge. The total resistance of a graph is the sum of effective resistance between all pairs of vertices in a graph, i.e.~$\Rtot = \sum_{s,t\in V} R_{st}$. Because effective resistance measures how well-connected two vertices are, the total resistance is often used as a measure of the connectivity or robustness of the entire graph~\cite{Klein1993resistance, GhoshBoydSaber08MinEffRes, ellens2011effective, summers2015topology, black2023understanding}. These works have shown that the amount adding an edge or changing an edge's weight changes total resistance is proportional to the (squared) biharmonic distance of the edge, meaning the higher the biharmonic distance, the more important the edge is to the connectivity of a graph. For example,~\citet{GhoshBoydSaber08MinEffRes} show:

\begin{restatable}[\citet{GhoshBoydSaber08MinEffRes}]{theorem}{rtotpartialderivative}
\label{thm:rtot_partial_derivative}
    Let $G=(V,E,w)$ be a weighted graph with $n$ vertices. Let $e\in E$ be an edge with weight $w_e$. Then
    $$
        \frac{\partial\Rtot}{\partial w_e} = -nB_e^{2}.
    $$ 
\end{restatable}

Additionally, it has been proved by various authors~\citep{summers2015topology, li2018kirchhoff, black2023understanding} that if an edge $e$ is added to a connected graph, the change in the total resistance is proportional to the biharmonic distance.

\begin{theorem}[Various Authors]
\label{thm:change_total_resistance}
Let $G=(V,E)$ be a connected unweighted graph with $n$ vertices. Let $e\in E$ be an edge such that $G\setminus\{e\}$ is connected. Then
$$
 \Rtot(G) -\Rtot(G\setminus\{e\}) = -n\frac{B_e^{2}}{1+R_e}.
$$
\end{theorem}

\paragraph{p-Resistance.}
Effective resistance has also been generalized to another distance called the \textit{\textbf{p-resistance}}~\citep{herbster2009predicting}, where the $p$-resistance between $s$ and $t$ is the minimum $p$-norm of any $st$-flow. (Effective resistance is then the 2-resistance.) The $p$-resistance is unrelated to our $k$-harmonic distance, even when $p=k$.

\subsection{Overview of the Paper}
In \Cref{sec:background}, we introduce the necessary background for this paper. In \Cref{sec:new_formula}, we give a new formula for biharmonic distance that connects biharmonic distance with an operator from algebraic topology called the down Laplacian. In~\Cref{sec:biharmonic_and_electrical_flows}, we prove a connection between biharmonic distance and the electrical flows used to define effective resistance. We then argue how this connection explains the success of biharmonic distance as a centrality measure. In~\Cref{sec:high_biharm_and_sparse_cuts}, we prove a connection between biharmonic distance and sparsity in a graph. Based off these theoretical results, in~\Cref{sec:application_clustering}, we introduce two graph clustering algorithms that use the biharmonic distance. In~\Cref{sec:karmonic}, we introduce a generalization of the biharmonic distance called the $k$-harmonic distance. In~\Cref{sec:experiments}, we test the clustering algorithms introduced in~\Cref{sec:application_clustering} on synthetic and real-world graph clustering datasets.

\section{Background} 
\label{sec:background}

\begin{figure*}[ht]
    \centering
    \includegraphics[width=0.9\linewidth]{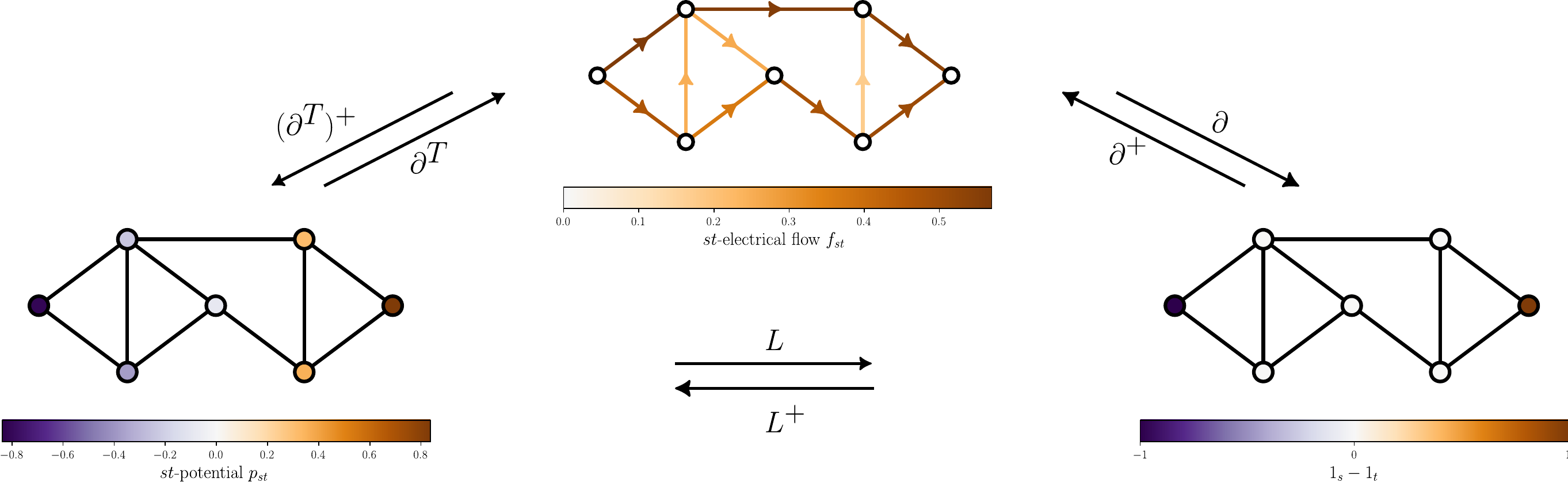}
    \caption{From right to left: the vector $1_s-1_t$, the electrical flow $f_{st} = \partial^+(1_s - 1_t)$, the potentials $p_{st}=L^{+}(1_s-1_t)$, and the maps that connect them.}
    \label{fig:potentials_flows_net_currents}
\end{figure*}

Let $G=(V,E,w)$ be a undirected, weighted graph with $n=|V|$ and $m=|E|$. The \textit{\textbf{adjacency matrix}} is the matrix $A\in\R^{n\times n}$ where $A_{uv}=w(\{u,v\})$ if $\{u,v\}\in E$ and 0 otherwise. The \textit{\textbf{degree matrix}} is the diagonal matrix $D\in\R^{n\times n}$ such that $D_{vv} = \sum_{u\in V} w(\{u,v\}) =\deg(v)$, the degree of $v$. The \textit{\textbf{Laplacian}} is the matrix $L=D-A\in\R^{n\times n}$. 
We index elements of the vectors in $\R^n$ and $\R^{m}$ by vertices and edges  of $G$, where we assume a bijection between $V$ and $E$ and orthonormal bases of $\R^{n}$ and $\R^{m}$.
\par 
There is another common, equivalent way of defining the graph Laplacian. For each edge $\{s,t\}\in E$, fix an arbitrary order of its vertices $e=(s,t)$. The \textit{\textbf{boundary matrix}} (aka \textit{\textbf{signed incidence matrix}}) is the matrix $\boundary:\R^{m}\to\R^{n}$ such that for an edge $e=(s,t)\in E$, the column $\boundary 1_e = 1_s-1_t$, where $1_{x}$ is the indicator column vector of a vertex or edge $x$. The \textit{\textbf{graph Laplacian}} is $L = \boundary W \boundary^{T}$, where $W\in\R^{m\times m}$ is the diagonal weight matrix. To simplify notation, in this paper, we denote $\wpartial = \partial W^{1/2}$, hence, $L = \wpartial\wpartial^T$.
\par 
For an edge $e=\{s,t\}\in E$, we use the notation $R_e$ and $B_e$ as shorthand for the effective resistance or biharmonic distance between its endpoints, i.e.~$B_e=B_{st}$.

\paragraph{Electrical Interpretation of Effective Resistance}
As the name suggests, the effective resistance has an interpretation in terms of electrical networks. For a graph $G$ with weights $w$, assume that each edge is a wire with resistance $w^{-1}_{e}$ (i.e.~conductance $w_e$). If we insert a unit of current into $s$ and remove a unit of current from $t$, then $R_{st}$ is the resistance between $s$ and $t$. 
\par 
There are other quantities associated with the electrical interpretation of effective resistance.
The \textit{\textbf{st-potential}} is the function $p_{st}:V\to\R$ defined $p_{st} = L^{+}(1_s-1_t)$; in the interpretation, $p_{st}(v)$ is the voltage at the vertex $v$. The \textit{\textbf{st-electrical flow}} between $s$ and $t$ is the function $f_{st}:E\to\R$ defined $f_{st} = W\partial^{T} L^{+}(1_{s}-1_{t})$; $f_{st}(e)$ is the amount of current flowing through the edge $e$. For unweighted graphs, we can simply this to $f_{st} = \partial^{+}(1_s-1_t)$. See~\Cref{fig:potentials_flows_net_currents}.
\par 
The following are well-known properties of $st$-potentials and $st$-electrical flows. For completeness, proofs of these lemmas can be found in~\Cref{apx:background}.

\begin{restatable}{lemma}{propertiesofpotentials}
\label{lem:properties_of_potentials}
\emph{(Properties of $st$-potentials)} % \emph to remove italics
\vspace{-0.25cm}
    \begin{enumerate} 
        \item $B_{st}^{2} = \| p_{st} \|^{2}$
        \item $R_{st} = p_{st}(s) - p_{st}(t)$
        \item $s = \arg\max_{v\in V} p_{st}(v)$ and $t = \arg\min_{v\in V} p_{st}(v)$
        \item $\sum_{v\in V} p_{st}(v) = 0$
    \end{enumerate}
\end{restatable}

\begin{restatable}{lemma}{propertiesofelectricalflows}
\label{lem:properties_of_electrical_flows}
\emph{(Properties of $st$-electrical flows)} % \emph to remove italics
\vspace{-0.25cm}
    \begin{enumerate}
        \item $R_{st} = f_{st}^{T}W^{-1} f_{st} = \sum_{e\in E} f_{st}^{2}(e)/w_e$
        \item  $f_{st} = \arg\min\{ f^{T}W^{-1}f : \partial f = (1_{s}-1_{t})\}$
    \end{enumerate}
\end{restatable}

\section{Formula for Biharmonic Distance of Edges}
\label{sec:new_formula}

In this section, we prove a new formula for the biharmonic distance on edges of a graph in terms of a different Laplacian associated with the graph called the \textit{\textbf{down Laplacian}}\footnotemark that is closely related to the graph Laplacian. Recall that the graph Laplacian is defined $L= \partial W \partial^{T}:\R^{n}\to\R^{n}$. The down Laplacian is defined $L^{down}=W^{-1/2}\partial^T\partial W^{-1/2}=\wpartial^T\wpartial:\R^{m}\to\R^{m}$. \Cref{thm:biharmonic_down_laplacian_diagonal} shows that the biharmonic distance of an edge is proportional to the diagonal entry of the down Laplacian. The proof of this theorem can be found in~\Cref{apx:new_formula}.

\footnotetext{The down Laplacian is a part of a family of linear operators on a topological space called the \textit{\textbf{combinatorial}} or \textbf{\textit{Hodge Laplacians}}~\citep{horak2013spectra} that are important to algebraic topology as they can be used to define the homology groups.}

\begin{restatable}{theorem}{biharmonicdownlaplaciandiagonal}
\label{thm:biharmonic_down_laplacian_diagonal}
    Let $G=(V, E, w)$ be a graph, $e\in E$, and $\downlap$ the down Laplacian of $G$. Then
    $$
        w_{e}\cdot B_{e}^{2} = (\downlap)^{+}_{ee} = \|1_e^T \wpartial^+\|,
    $$
    where $1_e^T\wpartial^+$ is the row of $\wpartial^+$ that corresponds to $e$.
\end{restatable}

One reason this formula is interesting is because it parallels the fact that the diagonals of the pseudoinverse of the Laplacian $L_{vv}^{+}$ have useful properties and have been used as a centrality measure on vertices called \textit{\textbf{current-flow closeness centrality}} or \textit{\textbf{information centrality}}~\citep{brandes2005centrality, stephenson1989rethinking}. We explore biharmonic distance as a centrality measure on edges in~\Cref{sec:centrality}.

\section{Biharmonic Distance and Electrical Flows}
\label{sec:biharmonic_and_electrical_flows}

Recall (\Cref{lem:properties_of_electrical_flows}) that the $st$-electrical flow $f_{st}$ is related to the effective resistance between $s$ and $t$ as follows:
$$
    R_{st} = f_{st}^{T}W^{-1}f_{st} = \sum_{e\in E} \frac{f_{st}(e)^{2}}{w_{e}}.
$$
As a corollary, the total resistance can also be defined using the electrical flows between all pairs of vertices $s$ and $t$.
$$
    \Rtot = \sum_{s,t\in V}\sum_{e\in E} \frac{f_{st}(e)^{2}}{w_e}
$$
Reordering the sums gives the formula
$$
    \Rtot = \sum_{e\in E}\sum_{s,t\in V} \frac{f_{st}(e)^{2}}{w_e}. 
$$
Consider the quantity $\sum_{s,t\in V} f_{st}(e)^{2}/w_{e}$ for an edge $e$. At first glance, this seems like an interesting quantity and a good measure of how important an edge is to the global connectivity of a graph. If it is high, then many pairs of vertices use the edge $e$ to send current in their electrical flow, so it is important to the connection between these vertices. We call this quantity the \textit{\textbf{squared electrical-flow centrality}} because of its connection to electrical-flow centrality; see~\Cref{sec:centrality} for details. Amazingly, the squared electrical-flow centrality is proportional to the squared biharmonic distance. 

\begin{restatable}{theorem}{biharmonicelectricalflow}
\label{thm:biharmonic_distance_is_squared_electrical_flow}
    Let $G=(V,E, w)$ be a connected weighted graph with $n$ vertices. Let $e\in E$. Then
    $$
    n\cdot w_e\cdot B_{e}^{2} = \sum_{s,t\in V} \frac{f_{st}(e)^{2}}{w_e}.
    $$
\end{restatable}

A proof of this theorem can be found in~\Cref{apx:biharmonic_and_electrical_flows}.

\subsection{Application: Edge Centrality}
\label{sec:centrality}

Edge centrality measures are ways of assigning values to the edges in a graph to determine which edges are most important to the global connectivity of a graph. \citet{yi2018biharmonic} proposed to use the weighted squared biharmonic distance $nw_e^2B_e^2$ as a measure of edge centrality, while~\citet{li2018kirchhoff} implicitly use the biharmonic distance in their Kirchhoff Edge Centrality. Both of these papers proposed the biharmonic distance as a centrality measure because of its connection to the total resistance. 
\par 
\Cref{thm:biharmonic_distance_is_squared_electrical_flow} provides further evidence why the squared biharmonic distance is a good measure of edge centrality. Intuitively, the biharmonic distance measures how important an edge is to the connection between all pairs of vertices in the graph. This is analogous to the \textit{\textbf{edge-betweenness centrality}}, which measures what proportion of shortest paths between pairs of vertices an edge appears in. Furthermore, the squared current-flow centrality is closely related to the \textit{\textbf{current-flow centrality}} $C_{e}$ introduced by \citet{brandes2005centrality} (which is equivalent to \textit{\textbf{random-walk betweenness}} introduced by \citet{newman2005measure}) defined
$$
C_{e} = \sum_{s,t\in V}|f_{st}(e)|.
$$
As the squared biharmonic distance of an edge in an unweighted graph is $nB_{e}^{2}=\sum_{s,t\in V} f^{2}_{st}(e)$, then $B_{e}^{2}$ and $C_{e}$ are comparable as they both measure the amount of current that flows through the edge $e$ in all-pairs electrical flows. This may explain some of the success of using the biharmonic distance as a centrality measure.
\par 
However, one advantage of the biharmonic distance over the current-flow centrality is computation time. \citet{brandes2005centrality} gave an algorithm for computing the current-flow centrality of all edges in $O(n^{3}+mn\log n)$ time, where $O(n^{3})$ is the time needed to compute the pseudoinverse of an $n\times n$ matrix. In contrast, because $B_{st} = \sqrt{L_{ss}^{2+}+L_{tt}^{2+}-2L_{st}^{2+}}$, computing the biharmonic distance for all edges takes $O(n^{3}+m)$ time, as computing a single biharmonic distance takes $O(1)$ time after computing the pseudoinverse of the $L^{2}$ in $O(n^{3})$. Thus, excluding the time to invert the matrices, biharmonic distance is faster to compute than current-flow centrality. See \Cref{apx:time_comparision} for an empirical comparison of these algorithms. 
\par 
Furthermore, \citet{yi2018biharmonic} observe that biharmonic distance can be efficiently approximated using random projections and fast solvers for linear systems in the Laplacian. This algorithm uses the fact that the biharmonic distance is a Euclidean distance, meaning it can be approximated by another Euclidean distance in a lower-dimensional space via random projection. As current-flow centrality is not a Euclidean distance, random projection techniques cannot be used. While computing the pseudoinverse of the Laplacian requires solving $n$ linear system in the Laplacian, approximating the biharmonic distances only requires solving $O(\log n)$ linear systems in the Laplacian. Combined with fast Laplacian solvers~\citep{spielman2004nearly}, this algorithm takes $O(m\poly\log n)$ time to approximate the biharmonic distance of all edges.

\subsection{Corollary: A Biharmonic Foster's Theorem}

A corollary to \Cref{thm:biharmonic_distance_is_squared_electrical_flow} is an analog of Foster's theorem for biharmonic distance. Foster's Theorem connects the effective resistances on the edges of the graph to the size of the graph, while the Biharmonic Foster's Theorem connects the biharmonic distance on edges to the total resistance.

\begin{theorem}[Foster's Theorem~\citep{foster1949average}]
\label{thm:foster}
    Let $G=(V,E,w)$ be a connected graph with $n$ vertices. Then 
    $$
        \sum_{e\in E} w_e R_{e} = n-1
    $$
\end{theorem}

\begin{corollary}[Biharmonic Foster's Theorem]
\label{cor:biharmonic_foster}
    Let $G=(V,E,w)$ be a connected graph with $n$ vertices. Then 
    $$
        n\sum_{e\in E} w_e B_{e}^{2} = \Rtot
    $$
\end{corollary}

\section{Biharmonic Distance and Cuts}
\label{sec:high_biharm_and_sparse_cuts}

\begin{figure}
    \centering
    \includegraphics[height=1.5in]{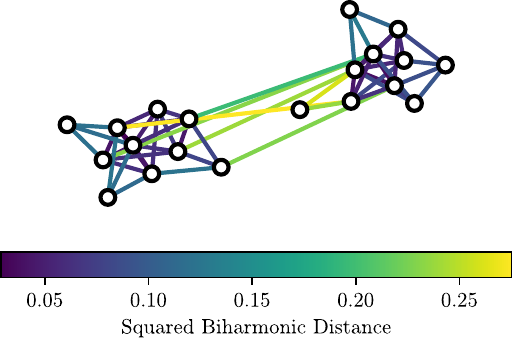}
    \caption{\Cref{thm:sparse_cut_implies_large_biharmonic,thm:large_biharmonic_edge_implies_sparse_cut} suggest that edges crossing sparse cuts will have high biharmonic distance, which we can see in this example of a block stochastic graph.}
    \label{fig:biharmonic_on_bsg}
\end{figure}

In this section, we prove several result connecting the biharmonic distance on edges to the existence of sparse cuts in unweighted graphs. Our results are in terms of the isoperimetric ratio of a cut. The isoperimetric ratio is a way of measuring the sparsity of a cut, because it is lower the fewer the edges in the cut and the closer the cut is to dividing the graph in half. The isoperimetric ratio is commonly used to quantify the quality of a clustering, from the classic Cheeger inequality~\cite{chung1997spectral} to more recent analyses of spectral clustering algorihtms~\cite{KannaVempalaVetta2004ClustGoodBadSpec, LeeGharanTrevisan2014MultiwaySpectral}.

Let $S\subset V$. Let $E(S, V\setminus S) = \{\{s,t\}\in E:s\in V,\,t\notin V\}$ be the edges leaving $S$. The \textit{\textbf{isoperimetric ratio}}\footnotemark of $S$ is 
$$
\Theta(S) = \frac{n|E(S, V\setminus S)|}{|S||V\setminus S|}.
$$
 
\footnotetext{It is common to alternatively define the isoperimetric ratio as $\Theta'(S) = \frac{|E(S, V\setminus S)|}{\max\{|S|,|V\setminus S|\}}$. However, these quantites are the same up to a constant, as $2\Theta(S) \leq \Theta'(S) \leq \Theta(S)$.}

\subsection{Biharmonic Distance on Cut Edges}

Our first result is about cut edges of a graph. For a connected graph $G$, a \textit{\textbf{cut edge}} is an edge $e$ such that $G\setminus\{e\}$ is disconnected.~\Cref{thm:biharmonic_of_cutedge_and_sparsity} shows that the squared biharmonic distance of a cut edge is the inverse isoperimetric ratio of its cut.

\begin{restatable}{theorem}{biharmoniccutedgesparsity}
\label{thm:biharmonic_of_cutedge_and_sparsity}
    Let $G=(V,E)$ be a connected, unweighted graph. Let $e$ be a cut edge of $G$. Let $S\cup T = V$ be the connected components of $G\setminus\{e\}$. Then 
    $$
        B_{e}^{2} = \frac{|S||T|}{n} = \Theta(S)^{-1}.
    $$
\end{restatable}

Proofs for all theorems in this section are in~\Cref{apx:high_biharm_and_sparse_cuts}.

\paragraph{Comparision with Effective Resistance} \Cref{thm:biharmonic_of_cutedge_and_sparsity} gives a good way to compare effective resistance to biharmonic distance. The effective resistance of a cut edge is always 1, so while the effective resistance reveals when an edge is a cut edge, it does not reveal anything about the cut and is the same whether the cut separates a single vertex or half the graph. In contrast, the biharmonic distance of a cut edge equals the inverse isoperimetric ratio of the cut, so in some sense, biharmonic distance reveals which cut edges are more important for connecting the graph. See~\Cref{fig:trees}.

\begin{theorem}[Folklore]
\label{thm:restance_of_cutedge}
    Let $G$ be an unweighted graph. Let $e$ be a cut edge. Then $R_e=1$
\end{theorem}

\subsection{Biharmonic Distance on Edges and Sparse Cuts}

\Cref{thm:biharmonic_of_cutedge_and_sparsity} shows that the squared biharmonic distance of a cut edge equals the sparsity of the cut. In this section, we generalize this theorem to all cuts, not just cuts containing a single edge. First, \Cref{thm:sparse_cut_implies_large_biharmonic} shows that the existence of a sparse cut implies the existence of an edge with high biharmonic distance. Conversely,~\Cref{thm:large_biharmonic_edge_implies_sparse_cut} shows that the existence of an edge with high biharmonic distance implies the existence of a sparse cut. See~\Cref{fig:biharmonic_on_bsg}.

\begin{restatable}{theorem}{sparsecutimplieslargebiharmonic}
\label{thm:sparse_cut_implies_large_biharmonic}
    Let $G=(V,E)$ be a connected, unweighted graph. Let $S\subset V$. Then
    $$
        \sum_{e\in E(S, V\setminus S)} B_{e}^{2} \geq \frac{|S||V\setminus S|}{|E(S, V\setminus S)|n} = \Theta(S)^{-1}.
    $$
    In particular, there is an edge $e\in E(S, V\setminus S)$ such that $B_e^{2}\geq\frac{\Theta(S)^{-1}}{|E(S, V\setminus S)|}$.
\end{restatable}

\begin{restatable}{theorem}{largebiharmonicedgeimpliessparsecut}
\label{thm:large_biharmonic_edge_implies_sparse_cut}
    Let $G=(V,E)$ be an unweighted graph. Let $\{s,t\}\in E$. Then there is a subset $S\subset V$ such that $s\in S$, $t\in V\setminus S$, and 
    $$
        B_{st}^{2}\in O(d_{\max}\Theta(S)^{-2}),
    $$
    where $d_{\max}$ is the maximum degree of any vertex in $G$.
\end{restatable}

\subsection{Application: Clustering Algorithms}
\label{sec:application_clustering}

The results in the previous section suggest two natural algorithms for graph clustering.

\paragraph{Biharmonic $k$-means}

\Cref{thm:sparse_cut_implies_large_biharmonic} intuitively suggests that vertices on opposite sides of a sparse cut likely have large biharmonic distance. This motivates a clustering algorithm that aims to separate vertices with large biharmonic distance. The $k$-means clustering algorithm minimizes the distance between points in the same cluster; however, $k$-means only can be applied to points in Euclidean space, not graphs. Fortunately, the biharmonic distance \textit{is} a Euclidean metric. Namely, if we consider the points $p_v:=L^{+}1_v\in\R^{n}$, then $B_{st} = \|p_u-p_v\|$. Therefore, we can cluster a graph by performing $k$-means on the points $\{p_v:v\in V\}$. This idea has previously been proposed for effective resistance by~\citet{yen2005clustering} and is analogous to the spectral clustering algorithm~\citep{SheMalik2000SpecClust, ng2001spectral} that similarly embeds vertices into Euclidean space then clusters them with $k$-means. The proposed clustering algorithm is summarized in Algorithm \ref{alg:biharmonic_kmeans_algorithm}.
\par 
We explore an alternative interpretation of this algorithm and compare it to spectral clustering in~\Cref{apx:interpretation_spectral_clustering}.

\begin{algorithm}[h]
\caption{Biharmonic $k$-Means (Graph $G$, int $k$)}
\label{alg:biharmonic_kmeans_algorithm}
\begin{algorithmic}[1]
    \STATE \# Cluster the graph $G$ into $k$ clusters
    \STATE Compute $L^{+}$ \# The columns of $L^{+}$ are the points $p_v$ 
    \STATE \textbf{return} $k$-means($L^{+}$)
\end{algorithmic}
\end{algorithm}

\paragraph{Biharmonic Girvan-Newman Algorithm}

\Cref{thm:sparse_cut_implies_large_biharmonic} implies that the existence of a sparse cut implies an edge with large biharmonic distance. We now propose a graph clustering algorithm inspired by this intuition.
\par 
Our algorithm is an instance of the generic \textit{\textbf{Girvan-Newman Algorithm}}~\cite{girvan2002community} for graph clustering, given below. The Girvan-Newman algorithm repeatedly removes the edge in a graph that maximizes some measure on the edges, typically a centrality measure. Our variant (\Cref{alg:biharmonic_girvan_newman_algorithm}) will use the (squared) biharmonic distance as the measure on the edges.

\begin{algorithm}[h]
\caption{Biharmonic Girvan-Newman (Graph $G$, int $k$)}
\label{alg:biharmonic_girvan_newman_algorithm}
\begin{algorithmic}
    \STATE \# Cluster the graph $G$ into $k$ clusters
    \WHILE{$G$ has fewer than $k$ connected components}
        \STATE $e_{\max}\gets \underset{e\in E}{\arg\max} B_{e}^{2}$ 
        \STATE $G\gets G\setminus\{e_{\max}\}$
    \ENDWHILE
    \STATE \textbf{return} connected components of $G$
\end{algorithmic}
\end{algorithm}

\section{\textit{k}-Harmonic Distance}
\label{sec:karmonic}

So far, we have been concerned with the biharmonic distance as a variant of effective resistance. The biharmonic distance alters the definition of effective resistance by using the squared pseudoinverse $L^{2+}$ instead of the pseudoinverse $L^{+}$ and adding a square root. However, we can generalize the biharmonic resistance even further by considering arbitrary powers of the pseudoinverse $L^{k+}$. We define the \textit{\textbf{k-harmonic distance}} between vertices $s$ and $t$ as
$$
    \kharmonic{st}{k} = \sqrt{(1_s-1_t)^{T} L^{k+} (1_s-1_t)}.
$$
The effective resistance between $s$ and $t$ is $R_{st} = (\kharmonic{st}{1})^{2}$, and the biharmonic distance between $s$ and $t$ is $B_{st} = \kharmonic{st}{2}$.
\par 
For the experimental part of this paper, we consider integer values of $k\geq 1$,  but generally, $k$ can take any real value. In particular, $(L^+)^{k}=L^{-k}$ for negative value of $k$, and
 $(L^+)^0$ is the orthogonal projection onto $\im(L^+)=\im(L)$.

\subsection{Low-Rank Approximation}
\label{sec:low_rank}

Another way to generalize effective resistance is to consider an approximation that only uses a subset of the eigenvectors of $L^{+}$. Since the Laplacian (and therefore all of it powers) are symmetric, we can spectrally decompose the pseudoinverse of $L^{k+} = \sum_{i=2}^{n}\frac{1}{\lambda_{i}^{k}}x_ix_{i}^{T}$. Using this decomposition, we can approximate $L^{k+}$ using only eigenvectors correspond to the the smallest $r$ eigenvalues, as these are the eigenvectors with the largest coefficients $1/\lambda_{i}^{k}$ in the spectral decomposition. We can then approximate the $k$-harmonic distance with this approximation of $L^{k+}$. The \textit{\textbf{rank-r k-harmonic distance}} is defined
$$
    H^{k,r}_{s,t} = \sqrt{\sum_{i=2}^{r+1} \frac{1}{\lambda_{i}^{k}} (1_s-1_t)^{T}x_{i}x_{i}^{T}(1_{s}-1_{t})}
$$
This sort of low rank approximation is common practice in applied spectral graph theory, e.g.~\citep{lipman2010biharmonic}.
\par 
The rank-$r$ $k$-harmonic distance allows us to place $k$-harmonic $k$-means clustering on a continuous spectrum between existing graph clustering algorithms. At one extreme of $k\to 0$, the coefficients of the eigenvalues $\lambda_{i}^{-k}\to 1$; therefore, we are clustering based on an unweighted embedding of the first $r$ eigenvectors. This is exactly the \textit{spectral clustering} algorithm \citep{SheMalik2000SpecClust,ng2001spectral}. At the other extreme of $k\to\infty$, the contribution of the smallest eigenvalue $\lambda_2$ dominates, so this is exactly the partitioning algorithm used in the proof of Cheeger's Inequality~\citep{cheeger1969, chung1997spectral}.

\subsection{A Foster's Theorem for \textit{k}-Harmonic Distance}

Unfortunately, the $k$-harmonic distance is currently less interpretable than the effective resistance and biharmonic distance (although it is still useful in practice; see~\Cref{sec:experiments}). However, some of our theorems for the biharmonic distance generalize to $k$-harmonics, even if they lose some of their interpretatiblity in the generalization.
\par 
For example, we prove a generalization of \Cref{thm:biharmonic_distance_is_squared_electrical_flow} (\Cref{thm:generalized_biharmonic_st_pairs}), noting that $f_{st}(e) = (\wpartial 1_e)L^+(1_s - 1_t)$, which equals $1_e^T\wpartial(\wpartial^+)^T\wpartial^+(1_s - 1_t) = 1_e^T\wpartial^+(1_s - 1_t)$. Likewise, \Cref{thm:generalized_biharmonic_edges} is a generalization of the fact that $R_{st}$ is the 2-norm of the $st$-electrical flow.

\begin{restatable}{theorem}{generalizedbiharmonicstpairs}
\label{thm:generalized_biharmonic_st_pairs}
    Let $G=(V,E, w)$ be a connected weighted graph with $n$ vertices, let $e\in E$, and let $k\in \R$. Then
    $$
    \sum_{s,t\in V}{\left((\wpartial 1_{e})^T (L^+)^k (1_s - 1_t)\right)^2} = n\cdot w_{e}\cdot (\kharmonic{e}{2k})^2.
    $$
\end{restatable}

\begin{restatable}{theorem}{generalizedbiharmonicedges}
\label{thm:generalized_biharmonic_edges}
    Let $G=(V,E, w)$ be a connected weighted graph with $n$ vertices, let $s,t\in V$, and let $k\in\R$. Then
    $$
    \sum_{e\in E}{\left((\wpartial 1_{e})^T (L^+)^k (1_s - 1_t)\right)^2} = (\kharmonic{st}{2k-1})^2.
    $$
\end{restatable}

From the theorems above, we prove the major result of this section, \Cref{thm:general_foster}, that relates the $2k$-harmonic distance on the edges of a graph to all-pairs $(2k-1)$-harmonic distance. This generalizes the well-known Foster theorem (\Cref{thm:foster}) and its biharmonic variant (\Cref{cor:biharmonic_foster}) for $k=1/2$ and $k=1$, respectively.

\begin{restatable}[$k$-harmonic Foster's Theorem]{theorem}{generalfoster}
\label{thm:general_foster}
Let $G = (V, E, w)$ be a weighted graph, and let $k\in \R$. Then
\[
\sum_{s,t\in V} (\kharmonic{st}{2k-1})^2 = n\cdot \sum_{e\in E} w_{e}\cdot(\kharmonic{e}{2k})^2.
\]
\end{restatable}

A proof of these theorems can be found in~\Cref{apx:karmonic}.

\section{Experiments}
\label{sec:experiments}

\subsection{Centrality}
\label{sec:experiments_centrality}

\begin{figure}[h]
    \centering
    \includegraphics[width=1.6in]{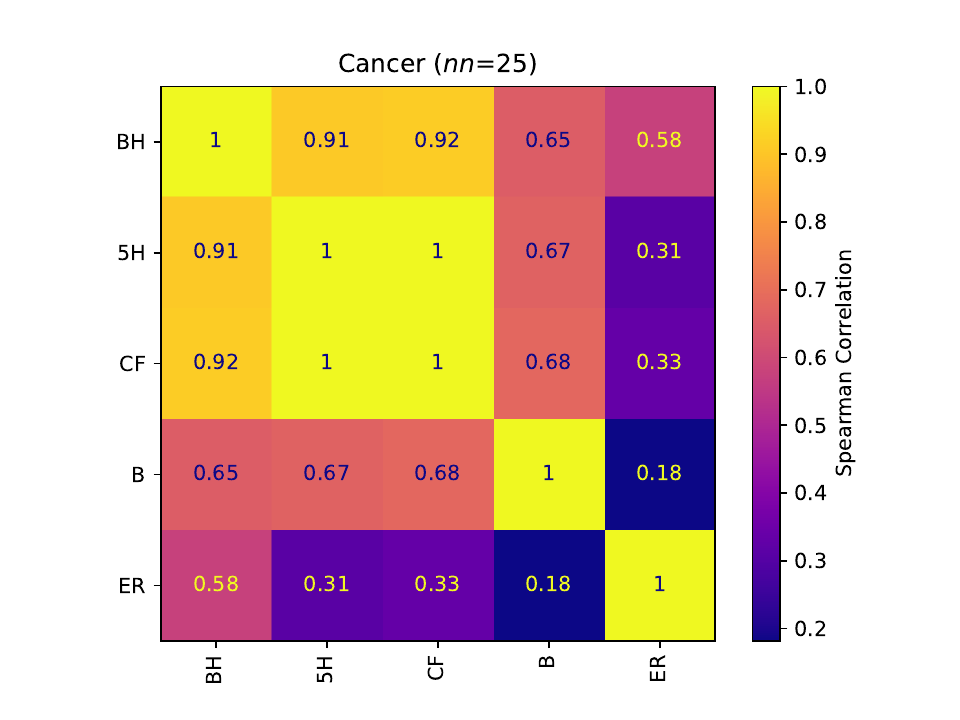}
    \includegraphics[width=1.6in]{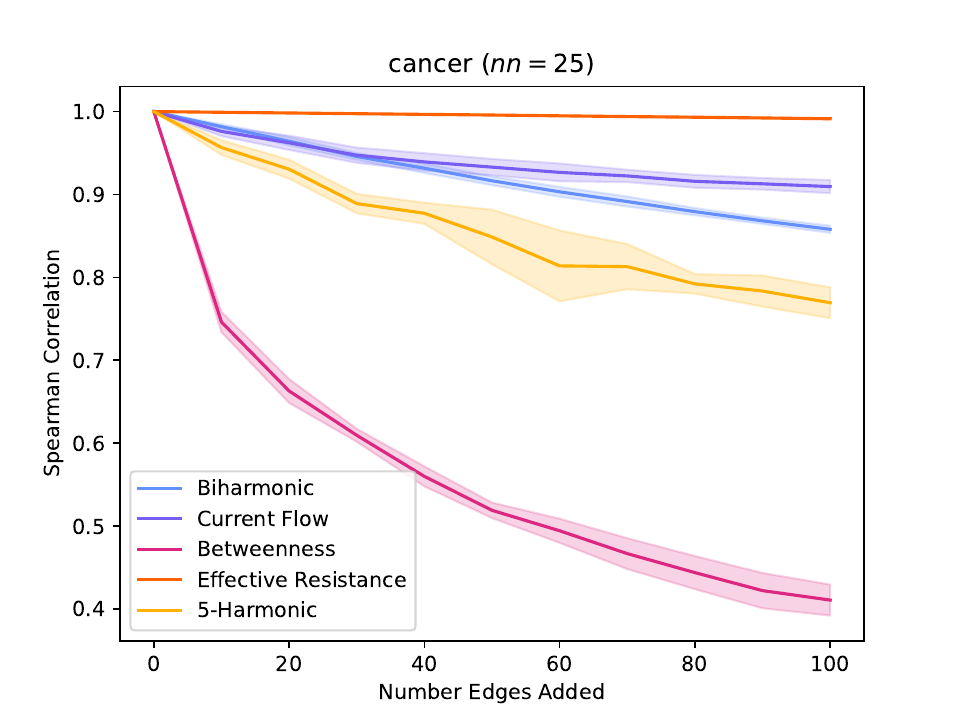}
    \caption{Left: The Spearman Rank Correlation Coefficient between different edge centrality measures on the 25-nearest neighbor graph of the Cancer dataset. Here BH is biharmonic distance, CF is current-flow centrality, B is betweenness centrality, ER is effective resistance, and 5H is 5-harmonic distance. 
    \newline Right: The Spearman Rank Correlation Correlation between an edge centrality measure and the same edge centrality measure after a number of random edges were added. Experiments were repeated 5 times for each centrality measure. Results for more graphs can be found in~\Cref{sec:centrality_experiments}.}
    \label{fig:cancer_confusion_and_resilience}
\end{figure}

This section compares the biharmonic distance and $k$-harmonic distance against the other centrality measures current-flow centrality, betweenness centrality, and effective resistance. We compare them in two ways: correlation and resilience to changes in the graph. For both experiments, we use the $k$-nearest neighbor graphs for $k\in\{25,50,75,100\}$ of the Cancer~\citep{misc_breast_cancer_14}, Wine~\citep{misc_wine_109}, and Iris~\citep{misc_iris_53} datasets from the UCI Machine Learning repository~\citep{ucirepository}. We use $k=5$ for the $k$-harmonic distance.
\par 
For the first experiment, we compare the ranking of the edges by the different centrality measures (e.g. the list of edges order from highest to lowest centrality) using the Spearman Rank Correlation Coefficient~\cite{spearman1904proof}.
\par 
For sparser graphs (fewer nearest neighbors $k$), the biharmonic distance is highly correlated with the current-flow centrality, which we might expect given their mutual connection to electrical flows. However, as $k$ increases, biharmonic distance tends to become less correlated with current-flow centrality and more correlated with effective resistance, while current-flow centrality tends to be more correlated with betweenness centrality. What is most surprising is that $5$-harmonic distance is often almost perfectly correlated with current-flow centrality. Our theoretical results do not posit any explanation for this behavior.
\par 
For the second experiment, we compare the original order of edges given by a centrality measure to the order given by the same centrality measure after adding a number of random edges. This experiment aims to test how resilient a centrality measure is to perturbations in the graph. The less the ordering of the edges changes, the more resilient the centrality measure is. This is a desirable property of a centrality measure as we would like the measure to be roughly the same given small changes in the input graph.
\par
We found the effective resistance had the highest resilience, as the ordering was often nearly unchanged by adding edges. (This extreme resilience to noise was previously noted by \citet{mavroforakis2015spanning}.) We also found that the biharmonic distance and current-flow centrality were the next most resilient to noise, followed by the 5-harmonic distance and then betweenness centrality.

Results for the 25-nearest neighbor graph of the Cancer dataset are given in~\Cref{fig:cancer_confusion_and_resilience}. Results for all graphs can be found in \Cref{sec:centrality_experiments}.

\subsection{Clustering}
\label{sec:experiments_clustering}

\begin{figure*}[ht]
    \centering
    \includegraphics[width=0.24\linewidth]{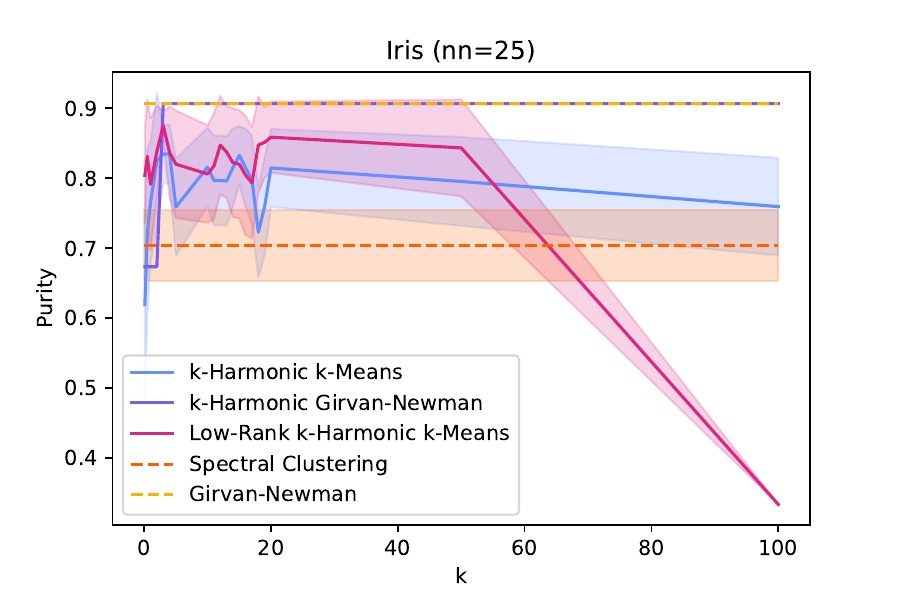}
    \includegraphics[width=0.24\linewidth]{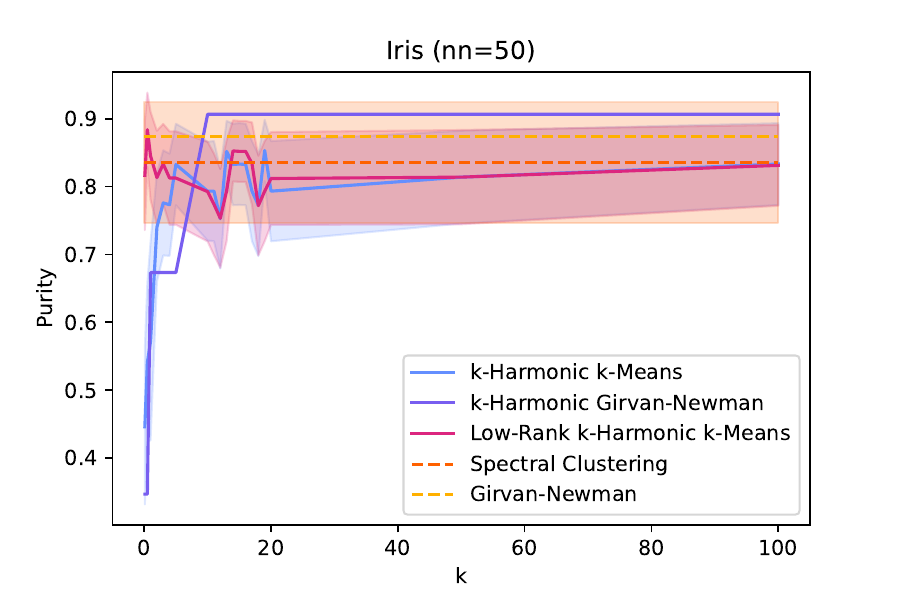}
    \includegraphics[width=0.24\linewidth]{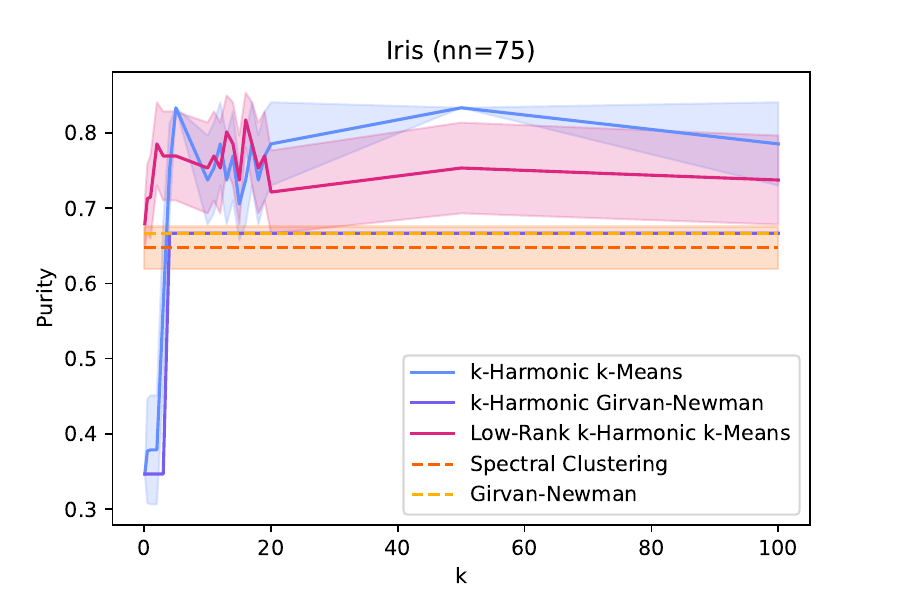}
    \includegraphics[width=0.24\linewidth]
    {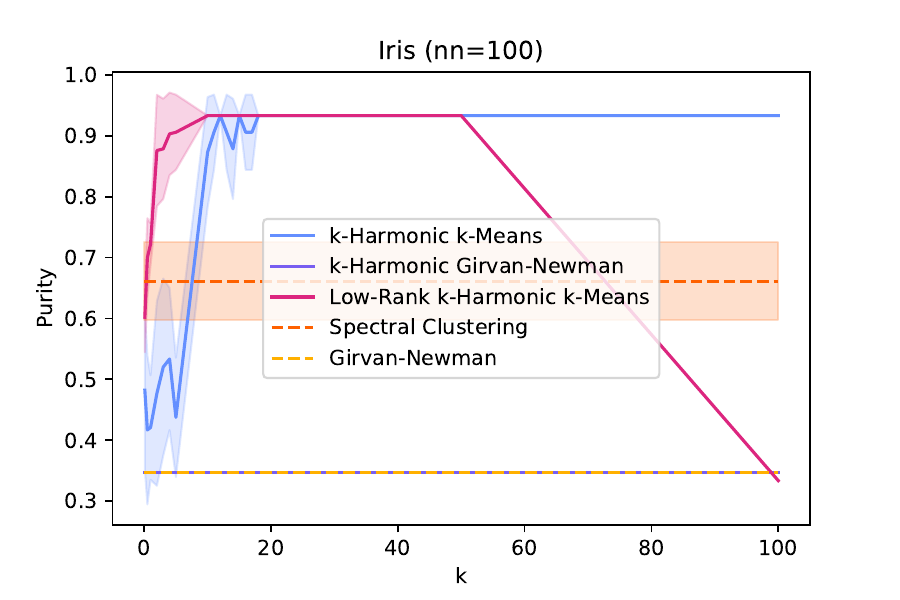}
    
    \includegraphics[width=0.24\linewidth]{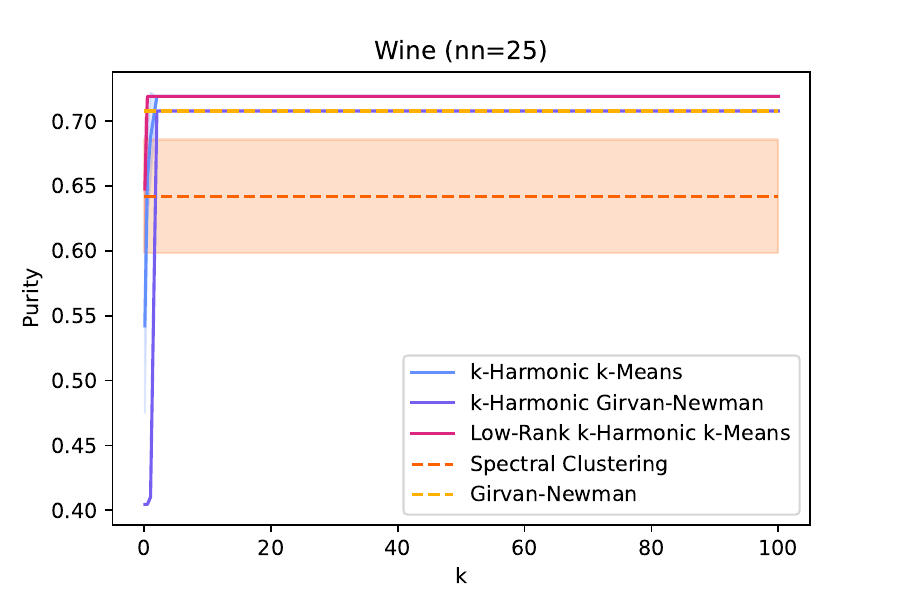}
    \includegraphics[width=0.24\linewidth]{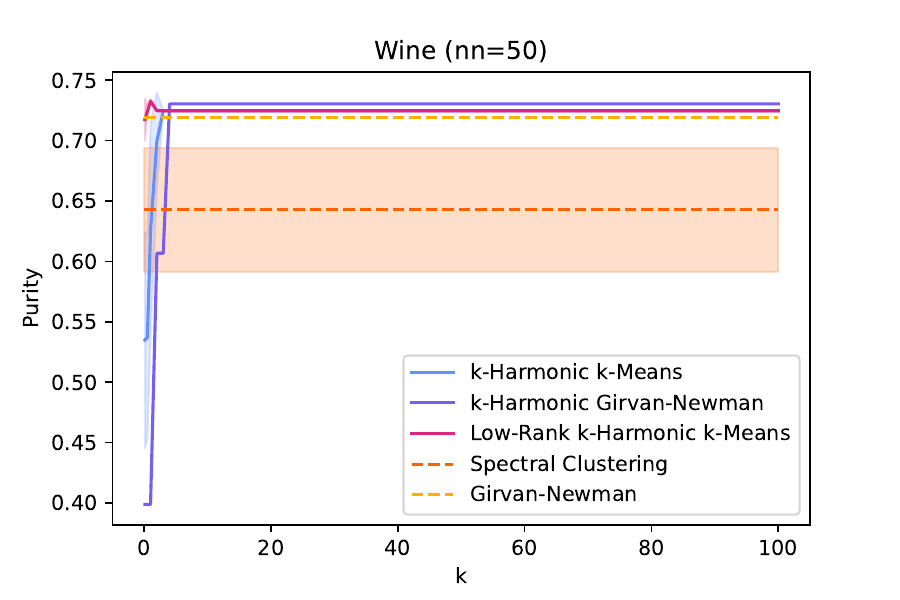}
    \includegraphics[width=0.24\linewidth]{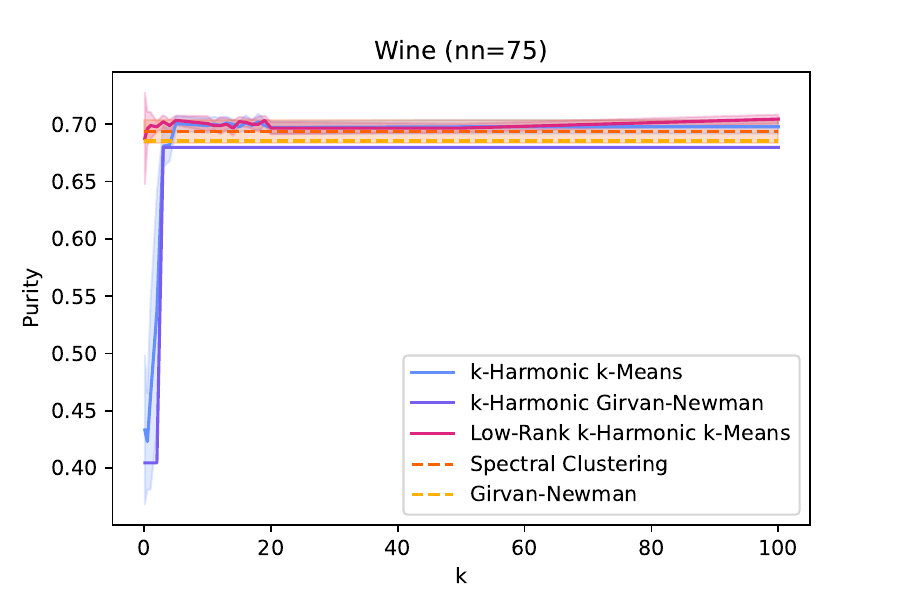}
    \includegraphics[width=0.24\linewidth]{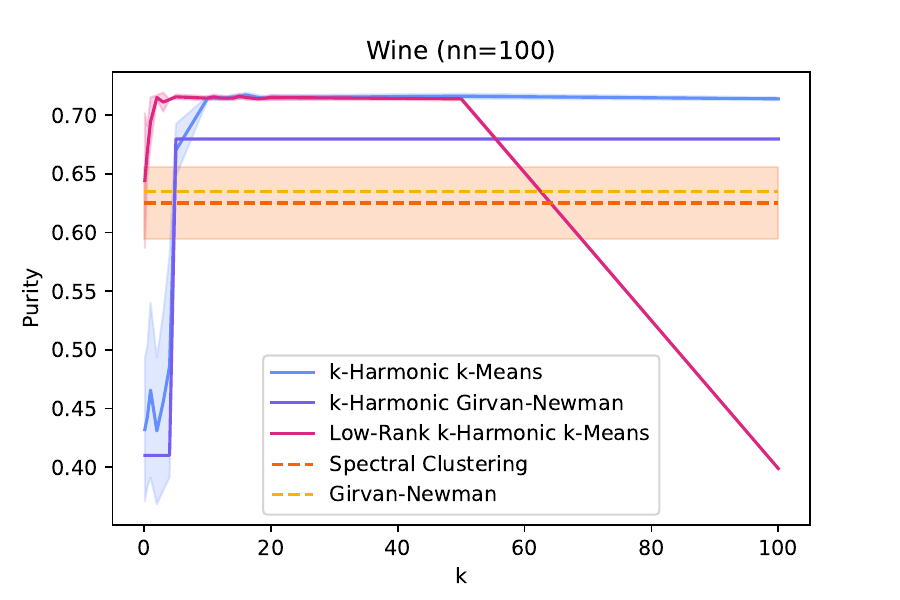}

    \includegraphics[width=0.24\linewidth]{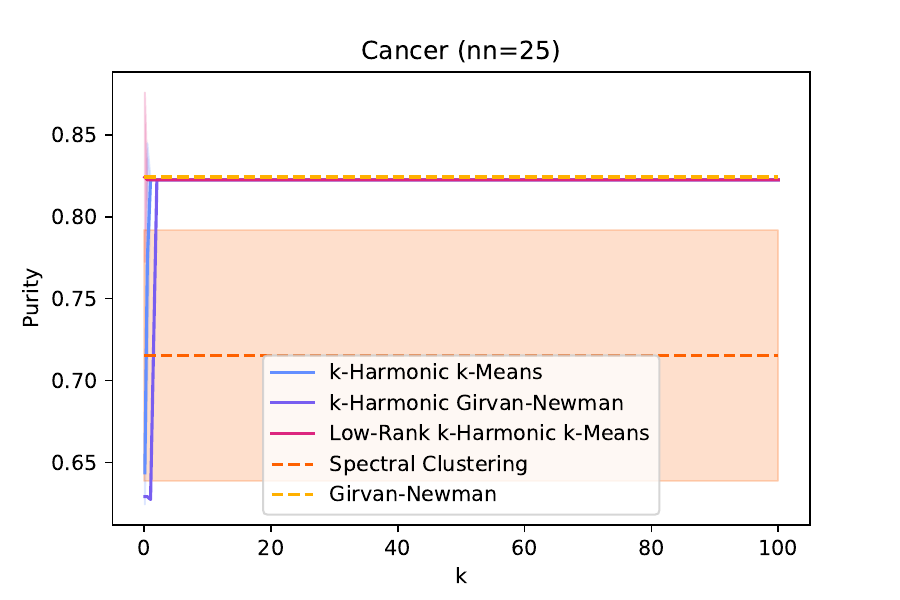}
    \includegraphics[width=0.24\linewidth]{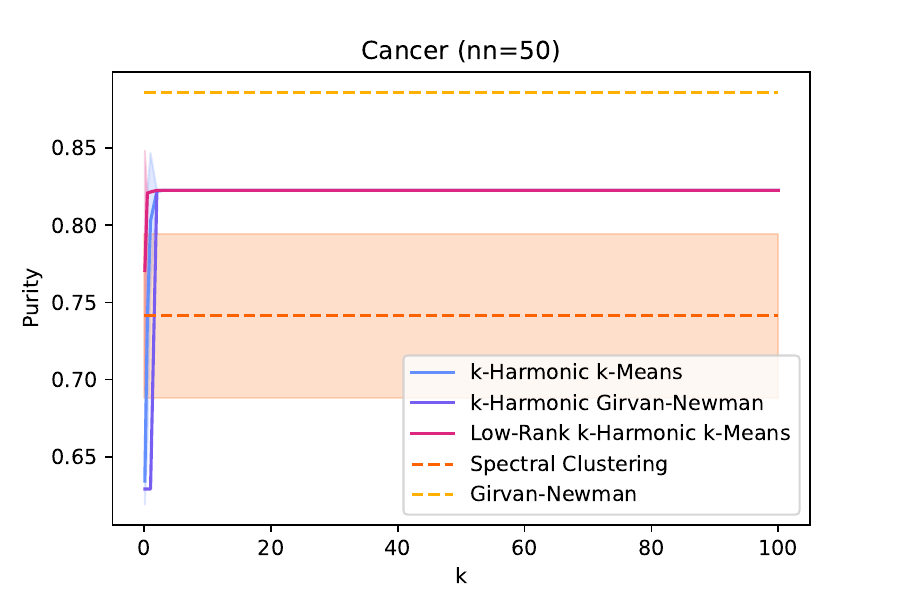}
    \includegraphics[width=0.24\linewidth]{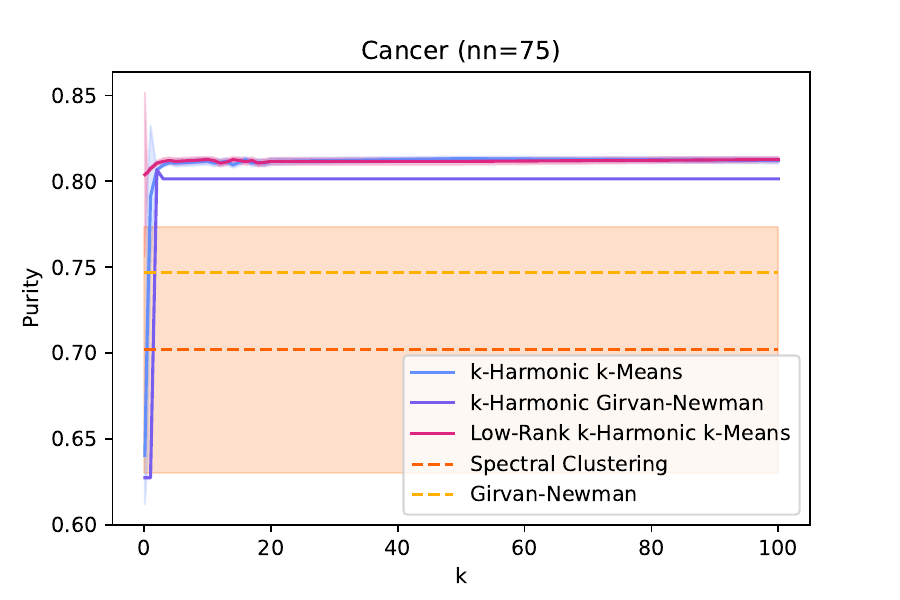}
    \includegraphics[width=0.24\linewidth]{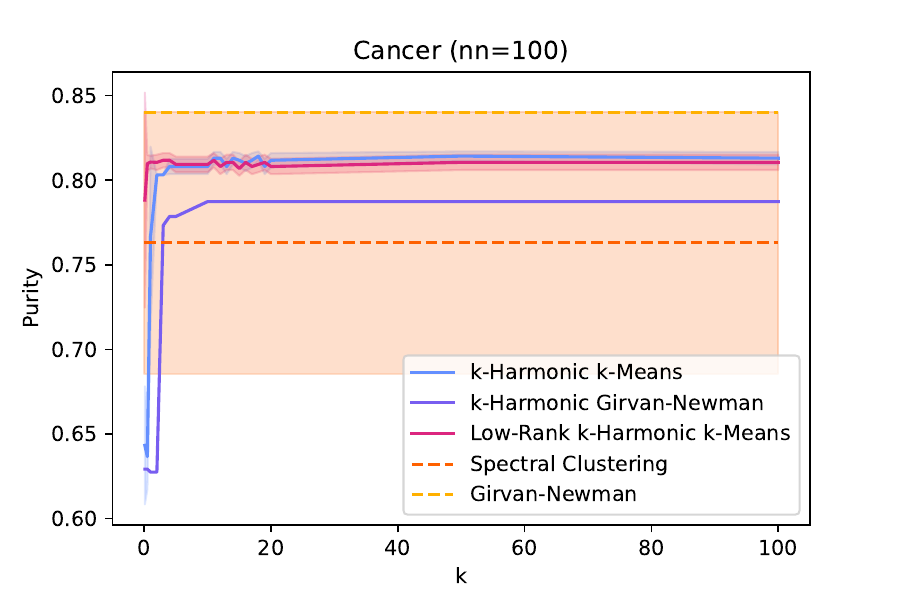}
    \caption{Plots of Purity vs.~$k$ for $k$-harmonic clustering algorithms. Different plots correspond to nearest neighbor graphs of different datasets. Algorithms not parameterized by $k$ (i.e.~Spectral Clustering and Girvan-Newman) are denoted with a dashed line. $k$-harmonic $k$-means, Low Rank $k$-harmonic $k$-means, and Spectral Clustering are averaged over 10 runs of $k$-means with different seeds.}  
    \label{fig:diagrams_in_body}
\end{figure*}

This section compares the performance of our proposed $k$-harmonic $k$-means and $k$-harmonic Girvan-Newman algorithm against classical clustering methods. We also compare $k$-means clustering using the low-rank approximation of $k$-harmonic distance from~\Cref{sec:low_rank}. We compare our algorithms against the Spectral Clustering~\cite{SheMalik2000SpecClust,ng2001spectral} and classical Girvan-Newman~\cite{girvan2002community} (using edge-betweenness centrality) clustering algorithms. As our computations of biharmonic and $k$-harmonic distances use the unnormalized Laplacian, we also use the unnormalized Laplacian for spectral clustering. 
\par 
We performed experiments of our clustering algorithms on synthetic and real world data. For the synthetic datasets, we generate connected block stochastic graphs with simulated clusters by generating edges between vertices within clusters with a probability $p$ and edges between vertices in different clusters with a probability $q$. For real world data, we use the Breast Cancer~\citep{misc_breast_cancer_14}, Wine~\citep{misc_wine_109}, and Iris~\citep{misc_iris_53} classification datasets from the UCI Machine Learning repository~\citep{ucirepository}. We create unweighted, undirected graphs from these datasets by building a $k$-nearest-neighbors graph for values of $k\in\{25, 50, 75, 100\}$. For all algorithms using the $k$-harmonic distance, we used the values of $k\in\{0.5,1,2,3,...,20,50,100\}$. For low-rank $k$-harmonic $k$-means, we set the rank $r$ to be number of clusters we are returning. We evaluate all algorithms using cluster purity. For all algorithms using k-means, we use random initialization of the centroids and give an average over ten trials with a 95\% confidence interval. Partial results of our experiments are given in~\Cref{fig:diagrams_in_body}. The results of our experiments for all datasets and algorithm are given in \Cref{apx:clustering_experiments}. 
\par 
We observe from \Cref{fig:diagrams_in_body,table:iris_experimental_results,table:cancer_experimental_results,table:wine_experimental_results,table:synthetic_experimental_results} in \Cref{apx:clustering_experiments} that algorithms using $k$-harmonic distances may be useful for clustering and can achieve results that outperform spectral clustering. Moreover, biharmonic distance almost always outperforms effective resistance, and in almost all cases, $k$-harmonic clustering and $k$-harmonic Girvan-Newman algorithms achieve optimal results at a wide range of values of $k$ larger than 2. A potential explanation for the high performance for large $k$ is given in~\Cref{apx:interpretation_spectral_clustering}. 

\textbf{Code} Code can be found at: 
{\small\url{https://github.com/ll220/biharmonic-kharmonic-clustering}}.

\section{Conclusion}

We studied the biharmonic and $k$-harmonic distance on graphs. We proved several useful theoretical properties of the biharmonic distance. We also proposed and tested graph clustering algorithms inspired by these properties. We found that our algorithms performed better with biharmonic distance than effective resistance, as expected, and the best results were achieved for $k$-harmonic distance with $k>2$. As well, our algorithms often performed comparably or better than existing algorithms.
\par 
While our theoretical results were mostly for the biharmonic distance, our experimental results suggests that the $k$-harmonic distance is also useful for edge centrality and graph clustering. It is an open question whether the $k$-harmonic distance has similar theoretical properties to the biharmonic distance.
\par 
Likewise, most of our theoretical results are for the biharmonic distance of \textit{edges}. It is an open question if the biharmonic distance between general pairs of vertices also have useful theoretical properties. 
\par 
While our experimental results on clustering appear to perform best for $k>2$, it is an open question of how to pick the best $k$ for a particular graph. Our experiments suggest that different values of $k$ are required to achieve optimal results for different graph inputs. It would be interesting to observe the behavior for $k$-harmonic clustering algorithms on a wider range of graphs to determine the graph parameters that correlate with the best $k$ value. 
\par 
Finally, we only compared biharmonic and $k$-harmonic distance to effective resistance for edge centrality and graph clustering. However, effective resistance has been used in many other areas of machine learning. (See~\Cref{sec:introduction} for examples.) Future work could explore using biharmonic and $k$-harmonic distance for these other applications.

\newpage

\section*{Acknowledgments}

Mitchell Black and Amir Nayyeri were supported by NSF Grants CCF-1816442, CCF-1941086, and CCF-2311180.

\section*{Impact Statement}

This paper presents work whose goal is to advance the field of Machine Learning. There are many potential societal consequences of our work, none which we feel must be specifically highlighted here.

\bibliography{main}

\begin{thebibliography}{50}
\providecommand{\natexlab}[1]{#1}
\providecommand{\url}[1]{\texttt{#1}}
\expandafter\ifx\csname urlstyle\endcsname\relax
  \providecommand{\doi}[1]{doi: #1}\else
  \providecommand{\doi}{doi: \begingroup \urlstyle{rm}\Url}\fi

\bibitem[Aeberhard \& Forina(1991)Aeberhard and Forina]{misc_wine_109}
Aeberhard, S. and Forina, M.
\newblock {Wine}.
\newblock UCI Machine Learning Repository, 1991.
\newblock {DOI}: https://doi.org/10.24432/C5PC7J.

\bibitem[Alev et~al.(2018)Alev, Anari, Lau, and
  Oveis~Gharan]{AlevAnaryOveis2017GraphClustEffRes}
Alev, V.~L., Anari, N., Lau, L.~C., and Oveis~Gharan, S.
\newblock Graph clustering using effective resistance.
\newblock In \emph{9th Innovations in Theoretical Computer Science Conference
  (ITCS 2018)}. Schloss Dagstuhl-Leibniz-Zentrum fuer Informatik, 2018.

\bibitem[Black et~al.(2023)Black, Wan, Nayyeri, and
  Wang]{black2023understanding}
Black, M., Wan, Z., Nayyeri, A., and Wang, Y.
\newblock Understanding oversquashing in gnns through the lens of effective
  resistance.
\newblock In \emph{International Conference on Machine Learning}, pp.\
  2528--2547. PMLR, 2023.

\bibitem[Boyd \& Vandenberghe(2004)Boyd and Vandenberghe]{boyd2004convex}
Boyd, S. and Vandenberghe, L.
\newblock \emph{Convex Optimization}.
\newblock Cambridge University Press, 2004.

\bibitem[Brandes \& Fleischer(2005)Brandes and
  Fleischer]{brandes2005centrality}
Brandes, U. and Fleischer, D.
\newblock Centrality measures based on current flow.
\newblock In \emph{Annual symposium on theoretical aspects of computer
  science}, pp.\  533--544. Springer, 2005.

\bibitem[Chakrabarti \& Faloutsos(2006)Chakrabarti and
  Faloutsos]{Chakrabarti2006survey}
Chakrabarti, D. and Faloutsos, C.
\newblock Graph mining: Laws, generators, and algorithms.
\newblock \emph{ACM Comput. Surv.}, 38\penalty0 (1):\penalty0 2–es, jun 2006.
\newblock ISSN 0360-0300.

\bibitem[Chandra et~al.(1996)Chandra, Raghavan, Ruzzo, Smolensky, and
  Tiwari]{chandra1996resistance}
Chandra, A.~K., Raghavan, P., Ruzzo, W.~L., Smolensky, R., and Tiwari, P.
\newblock The electrical resistance of a graph captures its commute and cover
  times.
\newblock \emph{computational complexity}, 6\penalty0 (4):\penalty0 312--340,
  Dec 1996.
\newblock ISSN 1420-8954.
\newblock \doi{10.1007/BF01270385}.
\newblock URL \url{https://doi.org/10.1007/BF01270385}.

\bibitem[Cheeger(1969)]{cheeger1969}
Cheeger, J.
\newblock A lower bound for the smallest eigenvalue of the laplacian.
\newblock In \emph{Proceedings of the Princeton conference in honor of
  Professor S. Bochner}, pp.\  195--199, 1969.

\bibitem[Chung(1997)]{chung1997spectral}
Chung, F.~R.
\newblock \emph{Spectral graph theory}, volume~92.
\newblock American Mathematical Soc., 1997.

\bibitem[Coifman \& Lafon(2006)Coifman and Lafon]{coifman2006diffusion}
Coifman, R.~R. and Lafon, S.
\newblock Diffusion maps.
\newblock \emph{Applied and computational harmonic analysis}, 21\penalty0
  (1):\penalty0 5--30, 2006.

\bibitem[Ellens et~al.(2011)Ellens, Spieksma, Van~Mieghem, Jamakovic, and
  Kooij]{ellens2011effective}
Ellens, W., Spieksma, F.~M., Van~Mieghem, P., Jamakovic, A., and Kooij, R.~E.
\newblock Effective graph resistance.
\newblock \emph{Linear algebra and its applications}, 435\penalty0
  (10):\penalty0 2491--2506, 2011.

\bibitem[Fisher(1988)]{misc_iris_53}
Fisher, R.~A.
\newblock {Iris}.
\newblock UCI Machine Learning Repository, 1988.
\newblock {DOI}: https://doi.org/10.24432/C56C76.

\bibitem[Foster(1949)]{foster1949average}
Foster, R.~M.
\newblock The average impedance of an electrical network.
\newblock \emph{Contributions to Applied Mechanics (Reissner Anniversary
  Volume)}, 333, 1949.

\bibitem[Ghosh et~al.(2008)Ghosh, Boyd, and Saberi]{GhoshBoydSaber08MinEffRes}
Ghosh, A., Boyd, S.~P., and Saberi, A.
\newblock Minimizing effective resistance of a graph.
\newblock \emph{{SIAM} Rev.}, 50\penalty0 (1):\penalty0 37--66, 2008.
\newblock \doi{10.1137/050645452}.
\newblock URL \url{https://doi.org/10.1137/050645452}.

\bibitem[Girvan \& Newman(2002)Girvan and Newman]{girvan2002community}
Girvan, M. and Newman, M.~E.
\newblock Community structure in social and biological networks.
\newblock \emph{Proceedings of the national academy of sciences}, 99\penalty0
  (12):\penalty0 7821--7826, 2002.

\bibitem[Hagberg et~al.(2008)Hagberg, Schult, and Swart]{hagberg2008networkx}
Hagberg, A.~A., Schult, D.~A., and Swart, P.~J.
\newblock Exploring network structure, dynamics, and function using networkx.
\newblock In Varoquaux, G., Vaught, T., and Millman, J. (eds.),
  \emph{Proceedings of the 7th Python in Science Conference}, pp.\  11 -- 15,
  Pasadena, CA USA, 2008.

\bibitem[Herbster \& Lever(2009)Herbster and Lever]{herbster2009predicting}
Herbster, M. and Lever, G.
\newblock Predicting the labelling of a graph via minimum \$p\$-seminorm
  interpolation.
\newblock In \emph{Annual Conference Computational Learning Theory}, 2009.
\newblock URL \url{https://api.semanticscholar.org/CorpusID:7413410}.

\bibitem[Horak \& Jost(2013)Horak and Jost]{horak2013spectra}
Horak, D. and Jost, J.
\newblock Spectra of combinatorial laplace operators on simplicial complexes.
\newblock \emph{Advances in Mathematics}, 244:\penalty0 303--336, 2013.

\bibitem[Kannan et~al.(2004)Kannan, Vempala, and
  Vetta]{KannaVempalaVetta2004ClustGoodBadSpec}
Kannan, R., Vempala, S., and Vetta, A.
\newblock On clusterings: Good, bad and spectral.
\newblock \emph{J. ACM}, 51\penalty0 (3):\penalty0 497–515, may 2004.
\newblock ISSN 0004-5411.
\newblock \doi{10.1145/990308.990313}.
\newblock URL \url{https://doi.org/10.1145/990308.990313}.

\bibitem[Kelly et~al.()Kelly, Longjohn, and Nottingham]{ucirepository}
Kelly, M., Longjohn, R., and Nottingham, K.~N.
\newblock The {UCI} machine learning repository.
\newblock URL \url{https://archive.ics.uci.edu}.

\bibitem[Khoa \& Chawla(2012)Khoa and Chawla]{khoa2012large}
Khoa, N. L.~D. and Chawla, S.
\newblock Large scale spectral clustering using resistance distance and
  spielman-teng solvers.
\newblock In \emph{Discovery Science: 15th International Conference, DS 2012,
  Lyon, France, October 29-31, 2012. Proceedings 15}, pp.\  7--21. Springer,
  2012.

\bibitem[Kirchhoff(1847)]{kirchhoff1847resistance}
Kirchhoff, G.
\newblock Ueber die aufl{\"o}sung der gleichungen, auf welche man bei der
  untersuchung der linearen vertheilung galvanischer str{\"o}me gef{\"u}hrt
  wird.
\newblock \emph{Annalen der Physik}, 148:\penalty0 497--508, 1847.

\bibitem[Klein \& Randi{\'{c}}(1993)Klein and
  Randi{\'{c}}]{Klein1993resistance}
Klein, D.~J. and Randi{\'{c}}, M.
\newblock Resistance distance.
\newblock \emph{Journal of Mathematical Chemistry}, 12\penalty0 (1):\penalty0
  81--95, Dec 1993.
\newblock ISSN 1572-8897.
\newblock \doi{10.1007/BF01164627}.
\newblock URL \url{https://doi.org/10.1007/BF01164627}.

\bibitem[Lee et~al.(2014)Lee, Gharan, and
  Trevisan]{LeeGharanTrevisan2014MultiwaySpectral}
Lee, J.~R., Gharan, S.~O., and Trevisan, L.
\newblock Multiway spectral partitioning and higher-order cheeger inequalities.
\newblock \emph{J. ACM}, 61\penalty0 (6), dec 2014.
\newblock ISSN 0004-5411.
\newblock \doi{10.1145/2665063}.
\newblock URL \url{https://doi.org/10.1145/2665063}.

\bibitem[Li \& Zhang(2018)Li and Zhang]{li2018kirchhoff}
Li, H. and Zhang, Z.
\newblock Kirchhoff index as a measure of edge centrality in weighted networks:
  Nearly linear time algorithms.
\newblock In \emph{Proceedings of the Twenty-Ninth Annual ACM-SIAM Symposium on
  Discrete Algorithms}, pp.\  2377--2396. SIAM, 2018.

\bibitem[Lin(2022)]{lin2022biharmonic}
Lin, Z.
\newblock The biharmonic index of connected graphs.
\newblock \emph{AIMS Mathematics}, 7\penalty0 (4):\penalty0 6050--6065, 2022.
\newblock ISSN 2473-6988.
\newblock \doi{10.3934/math.2022337}.
\newblock URL \url{https://www.aimspress.com/article/doi/10.3934/math.2022337}.

\bibitem[Lipman et~al.(2010)Lipman, Rustamov, and
  Funkhouser]{lipman2010biharmonic}
Lipman, Y., Rustamov, R.~M., and Funkhouser, T.~A.
\newblock Biharmonic distance.
\newblock \emph{ACM Transactions on Graphics (TOG)}, 29\penalty0 (3):\penalty0
  1--11, 2010.

\bibitem[Louis et~al.(2012)Louis, Raghavendra, Tetali, and
  Vempala]{louis2012many}
Louis, A., Raghavendra, P., Tetali, P., and Vempala, S.
\newblock Many sparse cuts via higher eigenvalues.
\newblock In \emph{Proceedings of the forty-fourth annual ACM symposium on
  Theory of computing}, pp.\  1131--1140, 2012.

\bibitem[L{\"u} \& Zhou(2011)L{\"u} and Zhou]{lu2011link}
L{\"u}, L. and Zhou, T.
\newblock Link prediction in complex networks: A survey.
\newblock \emph{Physica A: statistical mechanics and its applications},
  390\penalty0 (6):\penalty0 1150--1170, 2011.

\bibitem[Mavroforakis et~al.(2015)Mavroforakis, Garcia-Lebron, Koutis, and
  Terzi]{mavroforakis2015spanning}
Mavroforakis, C., Garcia-Lebron, R., Koutis, I., and Terzi, E.
\newblock Spanning edge centrality: Large-scale computation and applications.
\newblock In \emph{Proceedings of the 24th international conference on world
  wide web}, pp.\  732--742, 2015.

\bibitem[Newman(2005)]{newman2005measure}
Newman, M.~E.
\newblock A measure of betweenness centrality based on random walks.
\newblock \emph{Social networks}, 27\penalty0 (1):\penalty0 39--54, 2005.

\bibitem[Ng et~al.(2001)Ng, Jordan, and Weiss]{ng2001spectral}
Ng, A., Jordan, M., and Weiss, Y.
\newblock On spectral clustering: Analysis and an algorithm.
\newblock \emph{Advances in neural information processing systems}, 14, 2001.

\bibitem[Shi \& Malik(2000)Shi and Malik]{SheMalik2000SpecClust}
Shi, J. and Malik, J.
\newblock Normalized cuts and image segmentation.
\newblock \emph{IEEE Trans. Pattern Anal. Mach. Intell.}, 22\penalty0
  (8):\penalty0 888–905, aug 2000.
\newblock ISSN 0162-8828.
\newblock \doi{10.1109/34.868688}.
\newblock URL \url{https://doi.org/10.1109/34.868688}.

\bibitem[Spearman(1904)]{spearman1904proof}
Spearman, C.
\newblock The proof and measurement of association between two things.
\newblock \emph{The American Journal of Psychology}, 15\penalty0 (1), 1904.

\bibitem[Spielman(2019)]{spielman2019sagt}
Spielman, D.
\newblock Spectral and algebraic graph theory.
\newblock Available at http://cs-www.cs.yale.edu/homes/spielman/sagt/sagt.pdf
  (2021/12/01), 2019.

\bibitem[Spielman \& Srivastava(2011)Spielman and
  Srivastava]{SpielmanSrivastava2011SpecSpars}
Spielman, D.~A. and Srivastava, N.
\newblock Graph sparsification by effective resistances.
\newblock \emph{SIAM Journal on Computing}, 40\penalty0 (6):\penalty0
  1913--1926, 2011.

\bibitem[Spielman \& Teng(2004)Spielman and Teng]{spielman2004nearly}
Spielman, D.~A. and Teng, S.-H.
\newblock Nearly-linear time algorithms for graph partitioning, graph
  sparsification, and solving linear systems.
\newblock In \emph{Proceedings of the thirty-sixth annual ACM symposium on
  Theory of computing}, pp.\  81--90, 2004.

\bibitem[Stephenson \& Zelen(1989)Stephenson and
  Zelen]{stephenson1989rethinking}
Stephenson, K. and Zelen, M.
\newblock Rethinking centrality: Methods and examples.
\newblock \emph{Social Networks}, 11\penalty0 (1):\penalty0 1--37, 1989.
\newblock ISSN 0378-8733.
\newblock \doi{https://doi.org/10.1016/0378-8733(89)90016-6}.
\newblock URL
  \url{https://www.sciencedirect.com/science/article/pii/0378873389900166}.

\bibitem[Summers et~al.(2015)Summers, Shames, Lygeros, and
  D{\"o}rfler]{summers2015topology}
Summers, T., Shames, I., Lygeros, J., and D{\"o}rfler, F.
\newblock Topology design for optimal network coherence.
\newblock In \emph{2015 European Control Conference (ECC)}, pp.\  575--580.
  IEEE, 2015.

\bibitem[Teixeira et~al.(2013)Teixeira, Monteiro, Carri{\c{c}}o, Ramirez, and
  Francisco]{teixeira2013spanning}
Teixeira, A.~S., Monteiro, P.~T., Carri{\c{c}}o, J.~A., Ramirez, M., and
  Francisco, A.~P.
\newblock Spanning edge betweenness.
\newblock In \emph{Workshop on mining and learning with graphs}, volume~24,
  pp.\  27--31. Citeseer, 2013.

\bibitem[Velingker et~al.(2024)Velingker, Sinop, Ktena, Veli{\v{c}}kovi{\'c},
  and Gollapudi]{velingker2024affinity}
Velingker, A., Sinop, A., Ktena, I., Veli{\v{c}}kovi{\'c}, P., and Gollapudi,
  S.
\newblock Affinity-aware graph networks.
\newblock \emph{Advances in Neural Information Processing Systems}, 36, 2024.

\bibitem[Von~Luxburg(2007)]{vonluxburg2007tutorial}
Von~Luxburg, U.
\newblock A tutorial on spectral clustering.
\newblock \emph{Statistics and computing}, 17:\penalty0 395--416, 2007.

\bibitem[Wei et~al.(2021)Wei, Li, and Yang]{wei2021biharmonic}
Wei, Y., Li, R.-h., and Yang, W.
\newblock Biharmonic distance of graphs.
\newblock \emph{arXiv preprint arXiv:2110.02656}, 2021.

\bibitem[Yen et~al.(2005)Yen, Vanvyve, Wouters, Fouss, Verleysen, Saerens,
  et~al.]{yen2005clustering}
Yen, L., Vanvyve, D., Wouters, F., Fouss, F., Verleysen, M., Saerens, M.,
  et~al.
\newblock clustering using a random walk based distance measure.
\newblock In \emph{ESANN}, pp.\  317--324, 2005.

\bibitem[Yi et~al.(2018{\natexlab{a}})Yi, Shan, Li, and
  Zhang]{yi2018biharmonic}
Yi, Y., Shan, L., Li, H., and Zhang, Z.
\newblock Biharmonic distance related centrality for edges in weighted
  networks.
\newblock In \emph{IJCAI}, pp.\  3620--3626, 2018{\natexlab{a}}.

\bibitem[Yi et~al.(2018{\natexlab{b}})Yi, Yang, Zhang, and
  Patterson]{yi2018consensus}
Yi, Y., Yang, B., Zhang, Z., and Patterson, S.
\newblock Biharmonic distance and performance of second-order consensus
  networks with stochastic disturbances.
\newblock In \emph{2018 Annual American Control Conference (ACC)}, pp.\
  4943--4950. IEEE, 2018{\natexlab{b}}.

\bibitem[Yi et~al.(2021)Yi, Yang, Zhang, Zhang, and Patterson]{yi2021consensus}
Yi, Y., Yang, B., Zhang, Z., Zhang, Z., and Patterson, S.
\newblock Biharmonic distance-based performance metric for second-order noisy
  consensus networks.
\newblock \emph{IEEE Transactions on Information Theory}, 68\penalty0
  (2):\penalty0 1220--1236, 2021.

\bibitem[Zhang et~al.(2023)Zhang, Luo, Wang, and He]{zhang2023rethinking}
Zhang, B., Luo, S., Wang, L., and He, D.
\newblock Rethinking the expressive power of gnns via graph biconnectivity.
\newblock \emph{arXiv preprint arXiv:2301.09505}, 2023.

\bibitem[Zhu et~al.(2003)Zhu, Ghahramani, and
  Lafferty]{ZhuEtal2003SemiSupHarmonic}
Zhu, X., Ghahramani, Z., and Lafferty, J.
\newblock Semi-supervised learning using gaussian fields and harmonic
  functions.
\newblock In \emph{Proceedings of the Twentieth International Conference on
  International Conference on Machine Learning}, ICML'03, pp.\  912–919. AAAI
  Press, 2003.
\newblock ISBN 1577351894.

\bibitem[Zwitter \& Soklic(1988)Zwitter and Soklic]{misc_breast_cancer_14}
Zwitter, M. and Soklic, M.
\newblock {Breast Cancer}.
\newblock UCI Machine Learning Repository, 1988.
\newblock {DOI}: https://doi.org/10.24432/C51P4M.

\end{thebibliography}
\bibliographystyle{icml2024}

%%%%%%%%%%%%%%%%%%%%%%%%%%%%%%%%%%%%%%%%%%%%%%%%%%%%%%%%%%%%%%%%%%%%%%%%%%%%%%%
%%%%%%%%%%%%%%%%%%%%%%%%%%%%%%%%%%%%%%%%%%%%%%%%%%%%%%%%%%%%%%%%%%%%%%%%%%%%%%%
% APPENDIX
%%%%%%%%%%%%%%%%%%%%%%%%%%%%%%%%%%%%%%%%%%%%%%%%%%%%%%%%%%%%%%%%%%%%%%%%%%%%%%%
%%%%%%%%%%%%%%%%%%%%%%%%%%%%%%%%%%%%%%%%%%%%%%%%%%%%%%%%%%%%%%%%%%%%%%%%%%%%%%%
\newpage
\appendix
\onecolumn

\section{Bounds on Biharmonic Distance}
\label{sec:bounds}

In this section, we prove bounds on the biharmonic distance in unweighted graphs.

\begin{theorem}[Lower Bound]
\label{thm:lower_bound_biharmonic}
    Let $G=(V,E)$ be an unweighted connected graph with $n$ vertices. Let $s,t\in V$. Then 
    $$
        \frac{2}{n^{2}}\leq B_{st}^{2}
    $$
\end{theorem}
\begin{proof}
    This follows from the \textit{\textbf{Courant-Fischer theorem}}. While the Courant-Fischer theorem is more general, in this context, we only need the following corollary of the main theorem: for a symmetric matrix $A$ with minimal non-zero eigenvalue $\lambda_{\min}$, then for any non-zero vector $x\perp\ker A$, $\lambda_{\min}\leq x^{T}Ax/x^{T}x$.
    \par 
    The eigenvalues of the Laplacian of $G$ are bound above by $n$, so the non-zero eigenvalues of $L^{+}$ are bound below by $\frac{1}{n}$ and the non-zero eigenvalues of $L^{2+}$ are bound below by $\frac{1}{n^2}$. Moreover, the vector $(1_s-1_t)\perp \ker L$ as $\ker L$ is spanned by the all-ones vector. Therefore the squared biharmonic distance between any pair of vertices $s$ and $t$ is $B_{st}^{2} = (1_s-1_t)^{T}L^{2+}(1_s-1_t)\geq \frac{2}{n^2}$, with the factor of 2 coming from $(1_s-1_t)^{T}(1_s-1_t)$. 
\end{proof}

As the eigenvalues of $L^{2+}$ are upper-bounded by $O(n^{4})$, we can use the Courant-Fischer theorem to similarly upper-bound the biharmonic distance by $B_{st}^{2}\in O(n^{4})$; however, we can use a different argument to show that $B_{st}^{2}\in O(n^{3})$.

\begin{theorem}[Upper Bound]
\label{thm:upper_bound_biharmonic}
    Let $G=(V,E)$ be an unweighted graph with $n$ vertices. Let $s,t\in V$. Then 
    $$
        B_{st}^{2}\leq n^{3}.
    $$
\end{theorem}
\begin{proof}
    By \Cref{lem:properties_of_potentials} Item 1, we know that $R_{st} = p_{st}(s) - p_{st}(t)$. Moreover, by combining \Cref{lem:properties_of_potentials} Items 2 and 3, we conclude that $p_{st}(s)$ is positive and $p_{st}(t)$ is negative, which implies that $|p_{st}(s)|, |p_{st}(t)|\leq R_{st}$. Further, \Cref{lem:properties_of_potentials} Item 2 implies that $|p_{st}(v)|\leq R_{st}$ for all $v\in V$. Therefore,
    $$
        B_{st}^{2}  = \| p_{st} \|_2^2\leq \sum_{v\in V} R_{st}^{2} = nR_{st}^{2}.
    $$
    As $R_{st} \leq n-1$, then $B_{st}^{2}\leq n^{3}$
\end{proof}

\begin{theorem}[Upper Bound on Edges]
\label{thm:upper_bound_biharmonic_edge}
    Let $G=(V,E)$ be an unweighted graph with $n$ vertices. Let $e\in E$. Then
    $$
        B_{e}^2\leq n.
    $$
\end{theorem}
\begin{proof}
    The proof of this theorem follows the same steps as \Cref{thm:upper_bound_biharmonic}, except that we can derive a tighter upper bound as $R_{e}\leq 1$ for an edge $e$.
\end{proof}

\subsection{Tight Examples}

We now show that the bounds above are all tight up to a constant.
\begin{itemize}
    \item \textbf{Lower Bound} Consider the complete graph $K_n$ on $n$ vertices. The Laplacian $L$ of $K_n$ has $n-1$ eigenvalues of value $n$, and $1$ eigenvalue of value 0; moreover, the eigenvector associated with the 0 eigenvalue is the all-ones vector 1~\citep[Lemma 6.1.1]{spielman2019sagt}. The eigenvectors of $L$ of different eigenvalues are orthogonal as $L$ is symmetric, so because $(1_s-1_t)^{T}1=0$ for any vertices $s$ and $t$, then $(1_s-1_t)$ is an eigenvector of $L$ with eigenvalue $n$. Therefore, the squared biharmonic distance between any two vertices $s$ and $t$ is $B_{st}^{2}=(1_s-1_t)^{T}L^{2+}(1_s-1_t) = \frac{2}{n^2}$.  

    \item \textbf{Upper Bound} Consider the path graph $P_{n+1}$ on the vertices $\{v_0,\ldots,v_n\}$ for $n+1$ odd, and let $s=v_0$ and $t=v_n$ be the endpoints. The $st$-potential is the function $p_{st}(v_i) = \frac{n}{2}-i$. Therefore, the squared biharmonic distance is $B^{2}_{st} = p_{st}^{T}p_{st}\in\Omega(n^{3})$.

    \item \textbf{Upper Bound on Edges} Consider the path graph $P_{n}$ on the vertices $\{v_1,\ldots,v_n\}$ for $n$ even. Consider the center edge $e=\{v_{n/2},v_{n/2+1}\}$. This edge is a cut edge with $v_{n/2}\in S$ and $v_{n/2+1}\notin S$. Moreover, $|S|=\frac{n}{2}$ and $|V\setminus S|=\frac{n}{2}$. Therefore, by~\Cref{thm:biharmonic_of_cutedge_and_sparsity}, $B^{2}_{e} = \frac{|S||V\setminus S|}{n} = \frac{n^{2}}{4n} = \Omega(n)$.
\end{itemize}

\subsection{Bounds on k-Harmonic Distance}

Given that we are able to prove bounds on the biharmonic distance, it is a natural question if these same bounds can be generalized to the $k$-harmonic distance. The answer is only partially. By this, we mean that the technique of using the Courant-Fischer theorem (ala the proof of~\Cref{thm:lower_bound_biharmonic}) generalizes to $k$-harmonic distance, but the more advanced techniques in~\Cref{thm:upper_bound_biharmonic,thm:upper_bound_biharmonic_edge} do not generalize to $k$-harmonic distances for $k>2$. 

\begin{theorem}
\label{thm:bounds_karmonic}
    Let $G=(V,E)$ be an unweighted graph. Let $s,t\in V$. Let $k\in R$. Then
    $$
        \frac{2}{n^{k}} \leq (\kharmonic{st}{k})^{2} \leq n^{2k}.
    $$
\end{theorem}

\section{Proofs from \Cref{sec:background}}
\label{apx:background}

\propertiesofpotentials*

\begin{proof} 
    \begin{enumerate}
        \item This follows from the definition of biharmonic distance and $st$-potentials.
        
        \item This follows from the definition of effective resistance and $st$-potentials.

        \item Recall for any vertex $v\in V$, $Lp_{st}(u) = \sum_{\{u,v\}\in E} w_{\{u,v\}}(p_{st}(u) - p_{st}(v))$. The $st$-potentials satisfy $Lp_{st} = 1_s-1_t$. As all edge weights $w_{\{u,v\}}$ are positive, this means that for any vertex $u\in V\setminus\{s\}$, as $Lp_{st}(u)\leq 0$, then $u$ must have a neighbor that is strictly larger than it on $p_{st}$; thus, $u\neq\arg\max_{v\in V}p_{st}$. The only remaining possibility is that $s = \arg\max_{v\in V} p_{st}(v)$. We can prove $t = \arg\max_{v\in V} p_{st}(v)$ using a similar argument.

        \item This follows as $\im L^{+}\perp 1$, the all-ones vector. As $p_{st}\in\im L^{+}$, then $p_{st}^{T}1 = \sum_{v\in V} p_{st}(v)=0$. \qedhere
    \end{enumerate}
\end{proof}

\propertiesofelectricalflows*

\begin{proof}
    \begin{enumerate}
        \item This follows from the following derivation: 
        \begin{align*} 
        f_{st}^{T}W^{-1}f_{st} =& (1_s-1_t)^{T}L^{+}\partial W W^{-1}W\partial^{T}L^{+}(1_s-1_t) \\
        =& (1_s-1_t)^{T}L^{+}\partial W \partial^{T}L^{+}(1_s-1_t) \\
        =& (1_s-1_t)^{T}L^{+}LL^{+}(1_s-1_t) \\
        =& (1_s-1_t)^{T}L^{+}(1_s-1_t) = R_{st}
        \end{align*}
        \item This is an example of least-norm regression and can be solved in the standard way, e.g., Lagrange multipliers (see~\citep{boyd2004convex}).
    \end{enumerate}
\end{proof}

\section{Proofs from \Cref{sec:new_formula}}
\label{apx:new_formula}

\subsection{Proof of~\Cref{thm:biharmonic_down_laplacian_diagonal}}

\biharmonicdownlaplaciandiagonal*

\begin{proof}
    We can express the left hand side as
    \begin{align*}
        w_e B_{e}^{2} =& w_e(\boundary 1_e)^{T} L^{2+} \boundary 1_e \\
        =& (\boundary W^{1/2} 1_e)^{T} L^{2+} \boundary W^{1/2}1_e \\
        =& (\wpartial 1_e)^{T} L^{2+} \wpartial 1_e \\
        =& (1_e^T \wpartial^T) \left((\wpartial^+)^T\wpartial^{+}(\wpartial^+)^T\wpartial^{+}\right) (\wpartial 1_e),
    \end{align*}
    where we use the fact that $L^{+} = (\wpartial^+)^T\wpartial^{+}$. Further, the operators $\wpartial^{T} (\wpartial^{T})^{+} = \wpartial^{+} \wpartial = \proj_{\im(\wpartial^{T})}$, where $\proj_{\im(\wpartial^{T})}$ is the projection onto $\im(\wpartial^{T})$. As $\im(\wpartial)^{+} = \im(\wpartial)^{T}$, then $\wpartial^{T} (\wpartial^{T})^{+} \wpartial^{+} = \wpartial^{+}$ and $(\wpartial^{T})^{+} \wpartial^{+} \wpartial = (\wpartial^{T})^{+}$. Therefore, we can simplify the expression for $w_eB_e^{2}$ to 
    \begin{equation*}
        B_{e}^{2} =  1_e^{T} \wpartial^{+} (\wpartial^{T})^{+}  1_{e} 
        = 1_e^{T} (\downlap)^{+}  1_e = (\downlap)^{+}_{ee}. \qedhere
    \end{equation*}
\end{proof}

\section{Proofs from~\Cref{sec:biharmonic_and_electrical_flows}}
\label{apx:biharmonic_and_electrical_flows}

\biharmonicelectricalflow*

\begin{proof}
    It will be convenient to express $f_{st}(e)^{2}/w_e$ using vector notation.
    $$
        \frac{f_{st}(e)^{2}}{w_e}
         = \frac{1}{w_e}(1_e^{T} W\boundary^{T}L^{+} (1_s-1_t))^{2}
         = (1_e^{T} W^{-1/2} W\boundary^{T}L^{+} (1_s-1_t))^{2}
        = (1_e^{T} \wpartial^{T} L^{+} (1_s-1_t))^{2}. 
    $$
    Therefore,
    \begin{align*}
        \frac{f_{st}(e)^{2}}{w_e} =& \left(1_e^{T} \wpartial^{T}L^+ (1_s-1_t)\right)^{2} \\
        =& \left(1_e^{T} \wpartial^{T}L^+ (1_s-1_t)\right)\left(1_e^{T} \wpartial^{T}L^+ (1_s-1_t)\right)^T \\
        =& 1_e^{T} \wpartial^{T}L^+ (1_s-1_t)   (1_s-1_t)^TL^+ \wpartial 1_e.
    \end{align*}
    Now, adding the sides for all pairs $s,t\in V$, we obtain
    \begin{align*}
        & \sum_{s,t\in V}{\frac{f_{st}(e)^{2}}{w_e}} \\
        =&
        \sum_{s,t\in V}{\left(1_e^{T} \wpartial^{T}L^+ (1_s-1_t)   (1_s-1_t)^TL^+ \wpartial 1_e\right)} \\
        =& 1_e^{T} \wpartial^{T}L^+ \left(\sum_{s,t\in V}{(1_s-1_t)(1_s-1_t)^T}\right)L^+ \wpartial 1_e.
    \end{align*}
The term $(1_s-1_t)(1_s-1_t)^{T}$ is the Laplacian of the graph with the single edge $\{s,t\}$. Hence, the sum $\sum_{s,t\in V}{(1_s-1_t)(1_s-1_t)^{T}}$ is the Laplacian of the complete (unweighted) graph $K_n$. 
$$
    L_{K_n} = \sum_{s,t\in V} (1_s-1_t)(1_s-1_t)^{T}.
$$
The eigenvalues of $L_{K_n}$ are all $n$, with the exception of the all-ones vector $1$, which is the eigenvector with eigenvalue $0$~\citep[Lemma 6.1.1]{spielman2019sagt}. However, as $1\perp\im L^{+}$, then $L^{+}\wpartial 1_e$ is an eigenvector of $L_{K_n}$ with eigenvalue $n$. Therefore, 
\begin{align*}
\sum_{s,t\in V}{\frac{f_{st}(e)^{2}}{w_e}} =&
1_e^{T} \wpartial^{T}L^+ L_{K_n}L^+ \wpartial 1_e \\ =&
n\cdot 1_e^{T} \wpartial^{T}L^+ L^+ \wpartial 1_e \\ =&
n\cdot 1_e^{T} \partial^{T}W^{-1/2}L^+ L^+ \partial W^{-1/2}1_e \\ =&
w_e\cdot n  \cdot 1_e^{T} \partial^{T}L^+ L^+ \partial 1_e \\ =&
w_e\cdot n \cdot B_e. \qedhere
\end{align*}
\end{proof}

\section{Proofs from~\Cref{sec:high_biharm_and_sparse_cuts}}
\label{apx:high_biharm_and_sparse_cuts}

\Cref{lem:flow_on_edges_crossing_cut} is a useful property of cuts and flows that will be important for proving the theorems in this section.

\begin{lemma}
\label{lem:flow_on_edges_crossing_cut}
    Let $G=(V,E)$ be an unweighted graph. Let $s,t\in V$. Let $f:E\to\R$ such that $\boundary f=1_s-1_t$. Let $S\subset V$ such that $s\in S$ and $t\notin S$. Then 
    $$
        \sum_{e\in E(S,V\setminus S)} |f(e)| \geq 1
    $$
    Moreover, if $|E(S,V\setminus S)|=1$, then $\sum_{e\in E(S,V\setminus S)} |f(e)|=1$. 
\end{lemma}
\begin{proof}
    Consider the sum $\sum_{v\in S}\boundary f(v)$. As $\boundary f(v)=0$ for any $v\in S\setminus\{s\}$, then $\sum_{v\in S}\boundary f(v) = \boundary f(s) =1$. Moreover, recall that the boundary map $\boundary$ is obtained by choosing an (arbitrary) order $(u,v)$ for each edge $\{u,v\}\in E$. Therefore, we can write $\boundary f(v)$ using this ordering on the edges as $\boundary f(v)=\sum_{(u,v)\in E} f(\{u,v\}) - \sum_{(v,u)\in E} f(\{v,u\})$. Therefore, if we write $\sum_{v\in S}\boundary f(v)$ in terms of the value of $f$ on the edges, then for any edge $(u,v)$ with $u,v\in V$, the terms $f(\{u,v\})$ and $-f(\{v,u\})$ cancel, and we are left with
    \begin{align*}
        1 =& \sum_{v\in S}\boundary f(v) \\
        =& \sum_{v\in S} \left(\sum_{(u,v)\in E} f(\{u,v\}) - \sum_{(v,u)\in E}f(\{v,u\})\right)\\
        =& \sum_{e\in E(S,V\setminus S)} (\pm 1)f(e)
    \end{align*}
    where the $\pm 1$ term depends on the order of the edge. Therefore, the absolute value $1\leq \sum_{e\in E(S,V\setminus S)} |f(e)|$. Moreover, if there is a single edge in $E(S,V\setminus S)$, we find that $1= \sum_{e\in E(S,V\setminus S)} |f(e)|$.
\end{proof}

\biharmoniccutedgesparsity*

\begin{proof}
    By~\Cref{thm:biharmonic_distance_is_squared_electrical_flow}, we know that $n\cdot B_{e}^{2}=\sum_{s,t\in V} f_{st}(e)^{2}$. We can decompose the right-hand side into three sums: $\sum_{s\in S}\sum_{t\in T} f_{st}(e)^{2}$, $\sum_{s, t\in S} f_{st}(e)^{2}$, and $\sum_{s,t\in T} f_{st}(e)^{2}$.
    \par 
    For the first sum $\sum_{s\in V}\sum_{t\in T} f_{st}(e)^{2}$, we know by~\Cref{lem:flow_on_edges_crossing_cut} that $|f_{st}(e)|=1$ as $|E(S,T)|=1$. Therefore, $\sum_{s\in V}\sum_{t\in T} f_{st}(e)^{2} = |S||T|$.
    \par 
    For the second and third sums, we will reuse some ideas from the proof of~\Cref{lem:flow_on_edges_crossing_cut}. First, observe that if $s,t\in S$ or $s,t\notin S$, then $\sum_{v\in S} \boundary f(s) =0$. Therefore, $\sum_{v\in S}\boundary f(v) = \sum_{e\in E(S,T)} (\pm 1)f(e)=0$. As $|E(S,T)|=1$, then we conclude that $f(e)=0$ for $e\in |E(S,T)|$. Therefore, $\sum_{s, t\in S} f_{st}(e)^{2}=\sum_{s,t\in T} f_{st}(e)^{2}=0$.
    \par
    Therefore, $n\cdot B_{e}^{2}=\sum_{s,t\in V} f_{st}(e)^{2}=|S||T|$ and the theorem follows.
\end{proof}

\sparsecutimplieslargebiharmonic*

\begin{proof}
    Fix a pair of vertices $s\in S$ and $t\in V\setminus S$. Consider the electrical flow $f_{st}$. By \Cref{lem:flow_on_edges_crossing_cut}, we know that $f_{st}$ must send at least 1 unit of flow along the cut edges $E(S, V\setminus S)$, i.e. 
    $$\sum_{e\in E(S, V\setminus S)} | f_{st}(e) | \geq 1.$$
    The left-hand side is the 1-norm on a vector, so we can apply the Cauchy-Schwarz inequality to bound 
    \begin{align*}
    & \sqrt{\sum_{e\in E(S, V\setminus S)} f^{2}_{st}(e)} \\
    \geq & \frac{1}{\sqrt{|E(S, V\setminus S)|}}\sum_{e\in E(S, V\setminus S)} |f_{st}(e)|  \\
    \geq & \frac{1}{\sqrt{|E(S, V\setminus S)|}},
    \end{align*}
    or equivalently,
    $$
    \sum_{e\in E(S, V\setminus S)} f^{2}_{st}(e) \geq \frac{1}{|E(S, V\setminus S)|}.
    $$
    We can combine this with Theorem \ref{thm:biharmonic_distance_is_squared_electrical_flow} to get the following inequality.
    \begin{align*}
         \sum_{e\in E(S, V\setminus S)} B_{e}^{2} = & \frac{1}{n}\sum_{e\in E(S, V\setminus S)}\sum_{s,t\in V} f_{st}(e)^{2} \\
         \geq & \frac{1}{n}\sum_{e\in E(S, V\setminus S)}\sum_{s\in S}\sum_{t\in V\setminus S} f_{st}(e)^{2} \\
         \geq & \frac{1}{n}\sum_{s\in S}\sum_{t\in V\setminus S} \frac{1}{|E(S, V\setminus S)|} \\
         \geq & \frac{|S||V\setminus S|}{n|E(S, V\setminus S)|} = \Theta(S)^{-1} \qedhere
    \end{align*}
\end{proof}

To prove the converse, we use Cheeger's Inequality (\Cref{lem:cheeger}), that, for any vector $x$, proves the existence of a cut whose sparsity is upper-bounded by the Rayleigh quotient of $x$ with the Laplacian. 

\begin{lemma}[Cheeger's Inequality, e.g.~{\cite{chung1997spectral}}]
\label{lem:cheeger}
    Let $G=(V,E)$ be a graph. Let $x:V\to\R$. Let $\dmax$ be the maximum degree of a vertex in $G$. Then there is a subset of vertices $S\subset V$ such that 
    $$
        \frac{x^{T}Lx}{x^{T}x} \in \Omega(\Theta(S)^{2}\dmax^{-1}).
    $$ 
    Moreover, $S=\{v\in V: x(v)\geq t\}$ for some $t\in\R$ 
\end{lemma}

\begin{restatable}{theorem}{largebiharmonicimpliessparsecut}
\label{thm:large_biharmonic_implies_sparse_cut}
    Let $G=(V,E)$ be an unweighted graph. Let $s,t\in V$. Then there is a subset $S\subset V$ such that $s\in S$, $t\in V\setminus S$, and 
    $$
        B_{st}^{2}/R_{st}\in O(d_{\max}\Theta(S)^{-2}).
    $$
\end{restatable}

\begin{proof}
    Consider the $st$-potentials $p_{st} = L^{+}(1_s-1_t)$. By \Cref{lem:cheeger}, we can find a cut $S$ with sparsity $\frac{p_{st}^{T}Lp_{st}}{p_{st}^{T}p_{st}} \in \Omega(\Theta(S)^{2}/\dmax)$. The denominator $p_{st}^{T}p_{st} = B_{st}^{2}$. To bound the numerator, we recall the identity $L^{+}=L^{+}LL^{+}$; then, the numerator is just the effective resistance:
    \begin{align*}
        p_{st}Lp_{st} =& (1_s-1_t)^{T} L^{+}LL^{+} (1_s-1_t) \\
        =& (1_s-1_t)^{T}L^{+}(1_s-1_t) \\
        =& R_{st}
    \end{align*}
    This implies that $\frac{R_{st}}{B_{st}^{2}}\in \Omega(\Theta(S)^{2}/\dmax)$.
    \par
    Finally,~\Cref{lem:properties_of_potentials} shows that the maximum and minimum values of $p_{st}$ are achieved at $s$ and $t$. As $S$ is a superlevel set of $p_{st}$, then $s\in S$ and $t\in V\setminus S$.
\end{proof}

\largebiharmonicedgeimpliessparsecut*

\begin{proof}
    This follows from \Cref{thm:large_biharmonic_implies_sparse_cut} by observing that $R_{st}\leq 1$ for an edge $\{s,t\}$.
\end{proof}

\section{Proofs from~\Cref{sec:karmonic}}
\label{apx:karmonic}

\generalizedbiharmonicstpairs*

\begin{proof}
    For any $s,t\in V$, we have
    \begin{align*}
        ((\wpartial& 1_{e})^T (L^+)^k (1_s - 1_t))^2 \\
        =& (\wpartial 1_{e})^T (L^+)^k (1_s - 1_t)(1_s - 1_t)^T (L^+)^k (\wpartial& 1_{e}).
    \end{align*}
    Since the sum $\sum_{s,t\in V}{(1_s-1_t)(1_s-1_t)^{T}}$ is the Laplacian of the complete (unweighted) graph $K_n$,
    $$
        L_{K_n} = \sum_{s,t\in V} (1_s-1_t)(1_s-1_t)^{T},
    $$
    we obtain,
    \begin{align*}
        \sum_{s.v\in V}&{\left((\wpartial 1_{e})^T (L^+)^k (1_s - 1_t))\right)^2} \\
        =& (\wpartial 1_{e})^T (L^+)^{k} L_{K_n} (L^+)^{k} (\wpartial 1_{e}) \\
        =& n\cdot (\wpartial 1_{e})^T (L^+)^{2k} (\wpartial 1_{e})^T.
    \end{align*}
    The last inequality holds because the eigenvalues of $L_{K_n}$ are all $n$, with the exception of the all-$1$s vector denoted $1$, which is the eigenvector with eigenvalue $0$. However, as $1\perp\im (L^{+})^k$, then $(L^{+})^k(1_u - 1_v)$ is an eigenvector of $L_{K_n}$ if $k>0$  with eigenvalue $n$.  The same fact is true for the case that $k=0$, as $1\perp (1_u - 1_v)$.

    Finally, since $\wpartial 1_e = \sqrt{w_e}\cdot \partial 1_e$, we obtain
    \[
    n\cdot (\wpartial 1_e)^T (L^+)^{2k} \wpartial 1_e = 
    n\cdot w_e \cdot (\kharmonic{e}{2k})^2 \qedhere
    \]
\end{proof}

\generalizedbiharmonicedges*

\begin{proof}
    For any $e\in E$, we have
    \begin{align*}
         & (1_{e}^T \wpartial^T (L^+)^k (1_s - 1_t))^2 \\
        =& (1_{e}^T \wpartial^T (L^+)^k (1_s - 1_t))^T(1_{e}^T \cdot \wpartial^T (L^+)^k (1_s - 1_t)) \\
        =& (1_s - 1_t)^T (L^+)^k \wpartial 1_{e} 1_{e}^T \wpartial^T (L^+)^k (1_s - 1_t).
    \end{align*}
    Since the sum $\sum_{e\in E}{\wpartial 1_{e} 1_{e}^T \wpartial^T}=L$, the Laplacian of $G$, then
    \begin{align*}
        \sum_{e\in E}&{\left((\wpartial 1_{e})^T (L^+)^k (1_s - 1_t)\right)^2} \\
        =& (1_s - 1_t)^T (L^+)^k L (L^+)^k (1_s - 1_t) \\
        =& (1_s - 1_t)^T (L^+)^{k-1} L^+ (L^+)^{k-1} (1_s - 1_t) \\
        =& (1_s - 1_t)^T (L^+)^{2k-1} (1_s - 1_t) \\
        =& (\kharmonic{st}{2k-1})^2 \qedhere
    \end{align*}
\end{proof}

\generalfoster*

\begin{proof}
    Consider
    \[
    S := \sum_{e\in E}\sum_{s,t\in V}{\left(1_{e}^T\wpartial^T (L^+)^k (1_s - 1_t)\right)^2}
    \]
    By \Cref{thm:generalized_biharmonic_st_pairs},
    \[
    S = n\cdot \sum_{e\in E} w_{e}\cdot (\kharmonic{e}{2k})^2 .
    \]
    By \Cref{thm:generalized_biharmonic_edges}, when flipping the sums in $S$,
    \begin{align*}
    S &= \sum_{s,t\in V}\sum_{e\in E}{\left(1_{e}^T\wpartial^T (L^+)^k (1_s - 1_t)\right)^2} \\
    &= \sum_{s,t\in V} (\kharmonic{st}{2k-1})^2 . 
    \end{align*}
    Therefore,
    \[
    \sum_{s,t\in V} (\kharmonic{st}{2k-1})^2 = S = n\cdot \sum_{e\in E} w_{e}\cdot (\kharmonic{e}{2k})^2 \qedhere
    \]
\end{proof}

\section{An Alternative Interpretation of the \textit{k}-Harmonic \textit{k}-Means Algorithm}
\label{apx:interpretation_spectral_clustering}

In addition to the connection between biharmonic distance and sparse cuts (\Cref{thm:large_biharmonic_edge_implies_sparse_cut,thm:sparse_cut_implies_large_biharmonic}), another potential explanation for the success of the $k$-harmonic distance for clustering is its connection to the spectral clustering algorithm~\citep{SheMalik2000SpecClust,ng2001spectral} and the higher-order Cheeger inequality~\citep{louis2012many,LeeGharanTrevisan2014MultiwaySpectral}. As mentioned in \Cref{sec:low_rank}, the spectral clustering and $k$-harmonic $k$-means algorithms have some similarities. The spectral clustering algorithms performs $k$-means clustering on the embedding of the vertices using the first $k$ eigenvectors of the Laplacian. In contrast, our $k$-harmonic $k$-means performs $k$-means clustering on the embedding of the vertices using \textit{all} eigenvectors of the Laplacian, except with the eigenvectors scaled by the inverse polynomial of their eigenvalue $\lambda_i^{-k/2}$. Therefore, the eigenvectors that most contribute to the distribution of the vertices are the ones with the smallest eigenvalues (as these have the largest coefficient $\lambda_i^{-k/2}$.) Moreover, the larger the value of $k$, the more the embedding is weighted towards the eigenvectors with the smallest eigenvalues.
\par 
Both algorithms embed the vertices using the eigenvectors of the Laplacian and give particular preference to the small eigenvectors in different ways, either by dropping the larger eigenvectors (in the case of spectral clustering) or by weighting the smaller eigenvector more heavily (in the case of $k$-harmonic $k$-means.) This idea of using the smallest eigenvectors to cluster the graph is also supported theoretically. The higher-order Cheeger inequality \citep{louis2012many,LeeGharanTrevisan2014MultiwaySpectral} proves that the ability to partition a graph into $k$ subsets of small isoperimetric ratio is upper bounded by the $k$th smallest eigenvector. Moreover, the algorithm to find this partition uses a projection of the $k$ smallest eigenvectors. This reinforce the idea of clustering based on the smallest eigenvectors.
\par 
The connection between spectral clustering and the effective resistance has also been posed as an explanation for the success of spectral clustering; see \citep[Section 6.2]{vonluxburg2007tutorial}.  

\section{Empirical Analysis of Time to Compute Biharmonic Distance}
\label{apx:time_comparision}

\begin{figure}[h]
    \centering
    \includegraphics[width=2in]{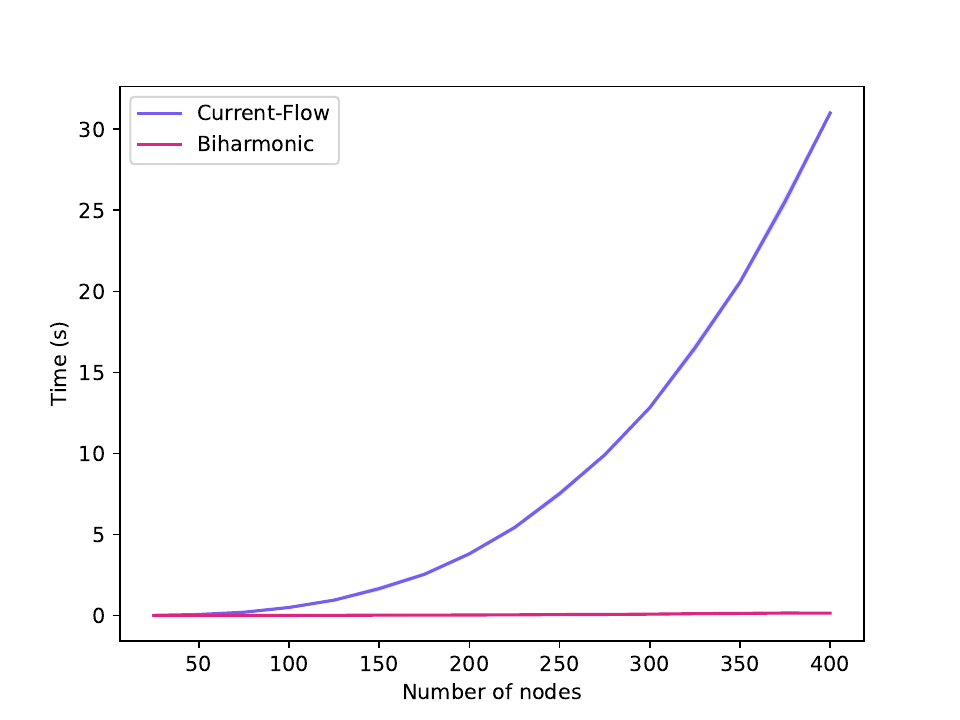}
    \caption{Time needed to compute biharmonic distance or current-flow centrality for all edges vs.~number of vertices in the graph.}
    \label{fig:enter-label}
\end{figure}

We compare the time needed to compute the biharmonic distance for all edges in a graph to the time needed to compute the current-flow centrality for all edges. We compare our implementation of the algorithm to compute biharmonic distance against the \small{\texttt{edge\_current\_flow\_betweenness\_centrality}} method from the NetworkX package~\cite{hagberg2008networkx}. We use the naive $O(n^{3})$ time algorithm for biharmonic distance; see~\Cref{sec:centrality}. We tested our method on Erdos-Renyi graphs with $n=25,50,...,400$ vertices and edge probability $p=0.5$. Results are averaged over 5 trials. Experiments were run on a 2018 MacBook Pro with a 2.3 GHz Quad-Core Intel Core i5 processor.
\par 
We found that our implementation was significantly faster than the NetworkX algorithm for current-flow centrality. While the difference in running times could be attributed to one of many differences in the implementation of the two algorithms, we feel reasonably confident this is a fair comparison as our implementation is a naive implementation with no significant optimizations. 

\newpage

\section{Figures}
\label{apx:examples}

\begin{figure}[ht]
    \centering
    \includegraphics[width=0.24\linewidth]{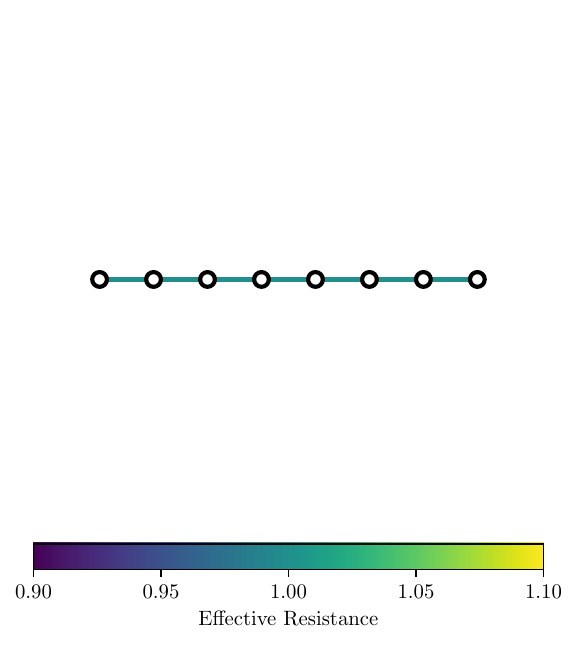}
    \includegraphics[width=0.24\linewidth]{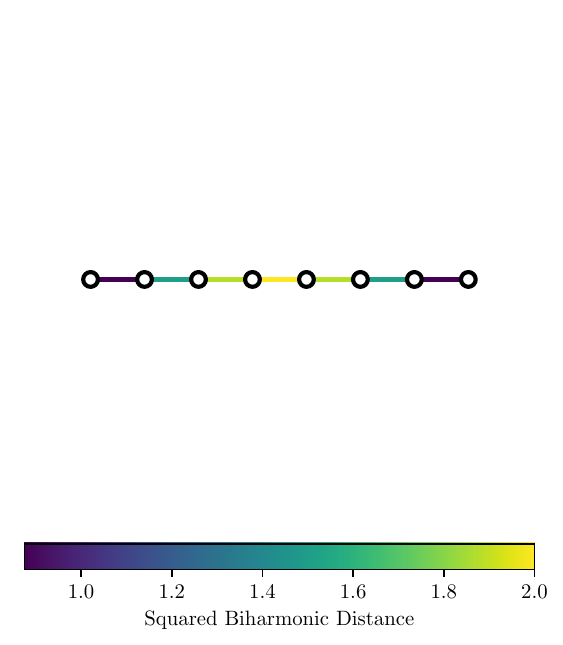}
    \includegraphics[width=0.24\linewidth]{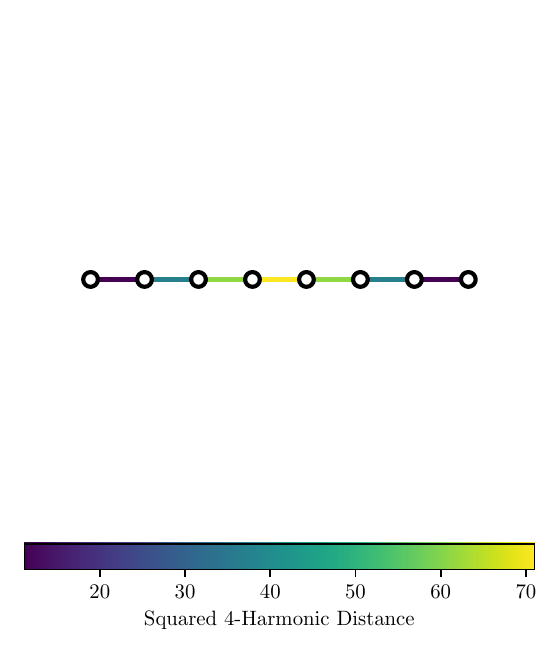}
    \includegraphics[width=0.24\linewidth]{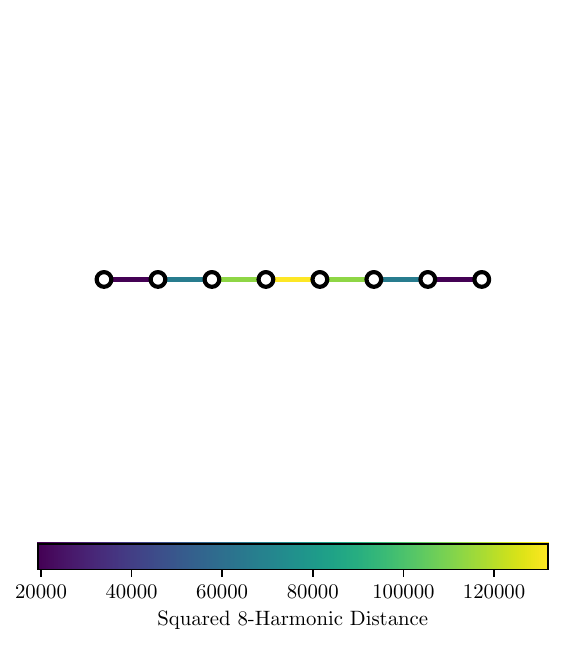}
                
    \includegraphics[width=0.24\linewidth]{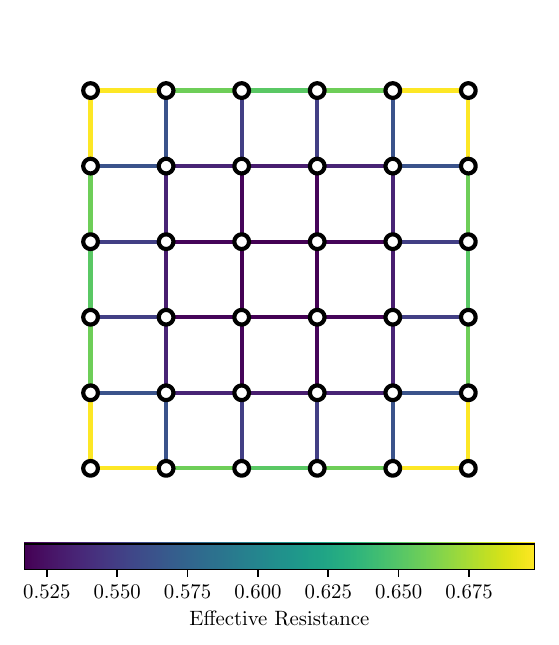}
    \includegraphics[width=0.24\linewidth]{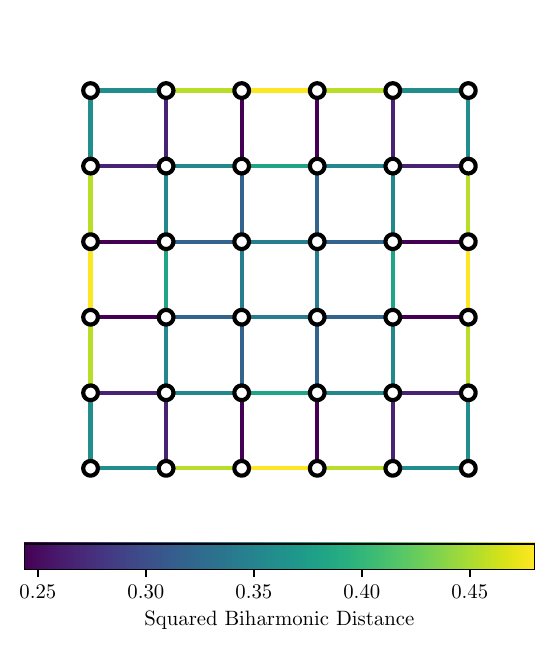}
    \includegraphics[width=0.24\linewidth]{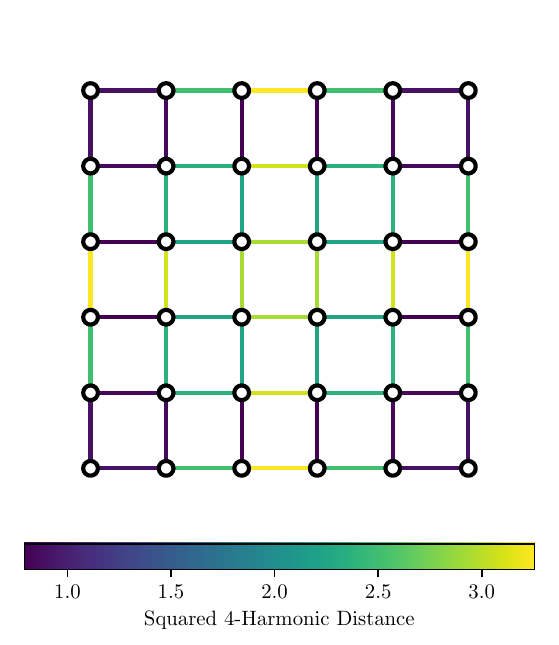}
    \includegraphics[width=0.24\linewidth]{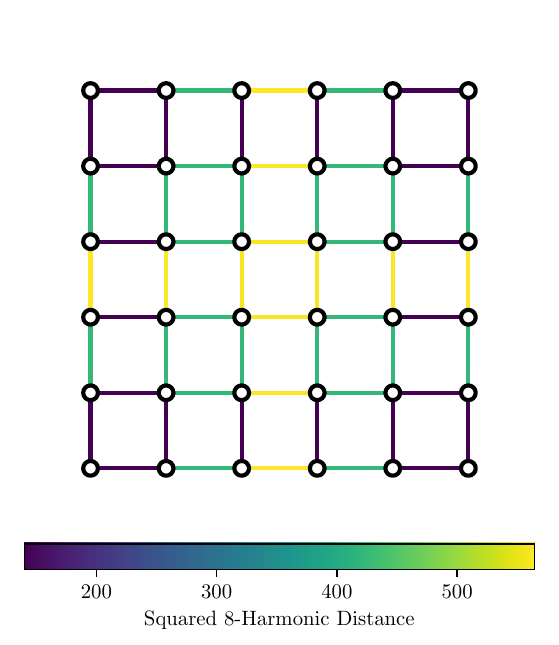}

    \includegraphics[width=0.24\linewidth]{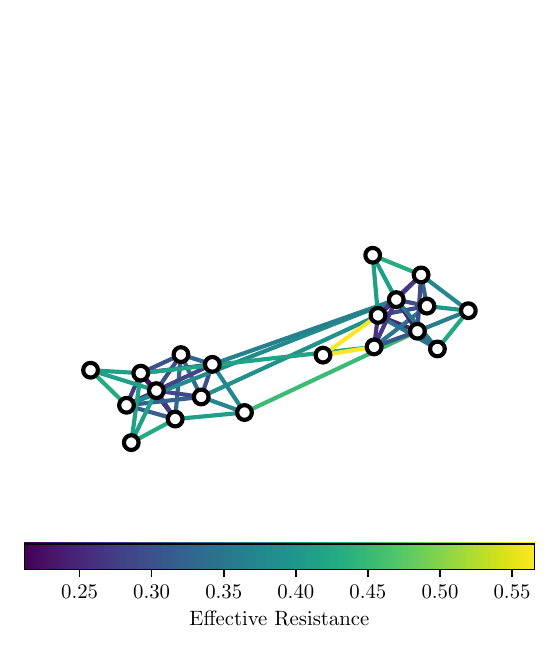}
    \includegraphics[width=0.24\linewidth]{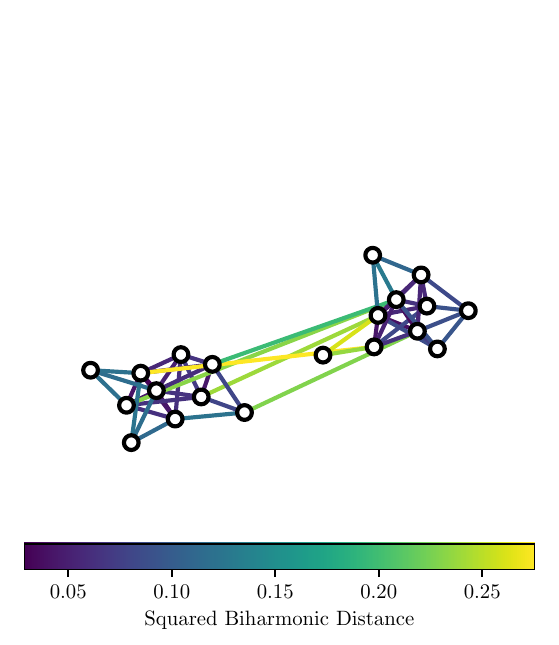}
    \includegraphics[width=0.24\linewidth]{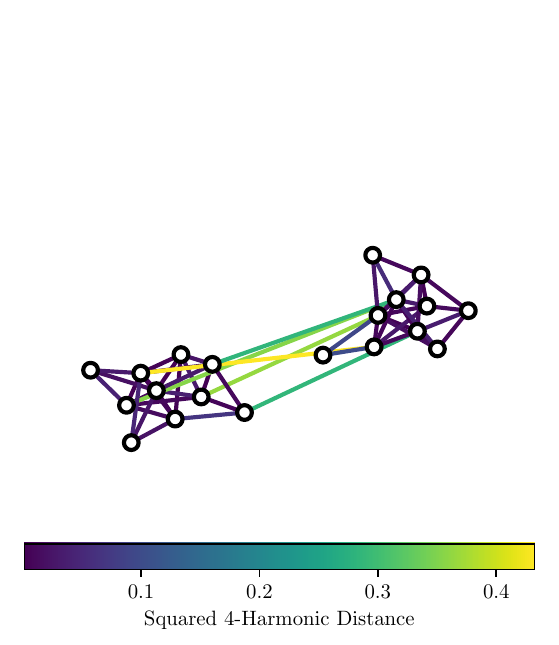}
    \includegraphics[width=0.24\linewidth]{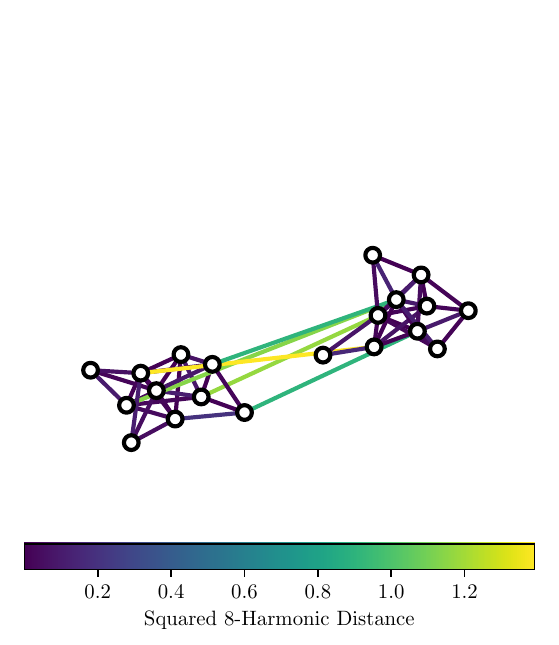}
    
    \caption{Examples of the squared $k$-harmonic distance for different graphs.}
    \label{fig:examples_karmonic}
\end{figure}

\newpage

\section{Results from Centrality Experiments}
\label{sec:centrality_experiments}

\newcommand{\mywidth}{1.5in}
\begin{figure}[H]
    \centering
    \begin{subfigure}
        \centering
        \includegraphics[width=\mywidth]{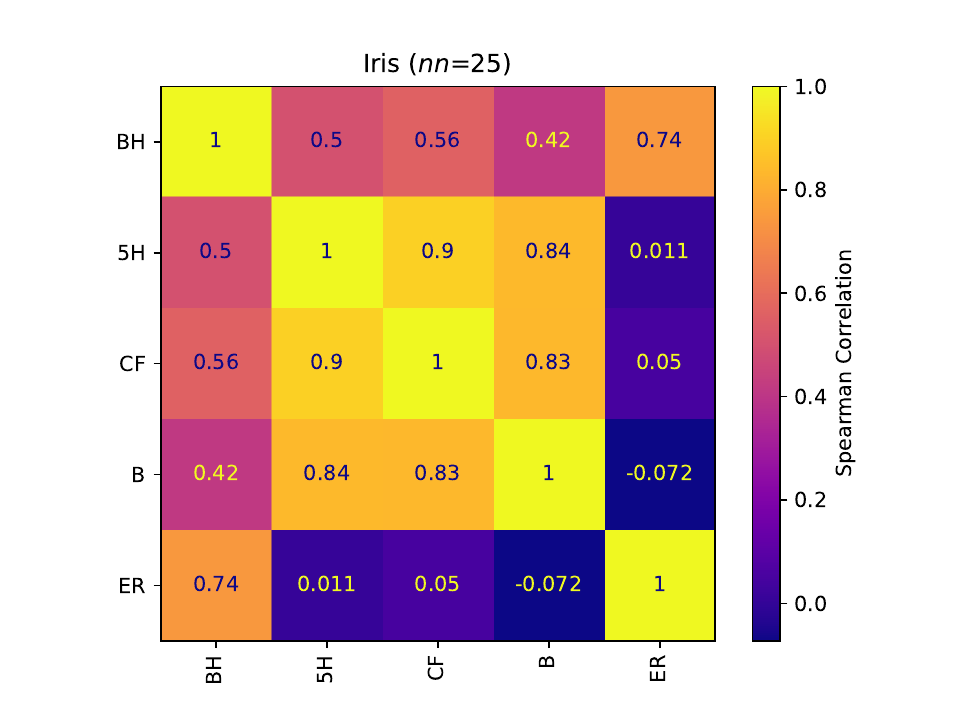}
    \end{subfigure}
    \begin{subfigure}
        \centering
        \includegraphics[width=\mywidth]{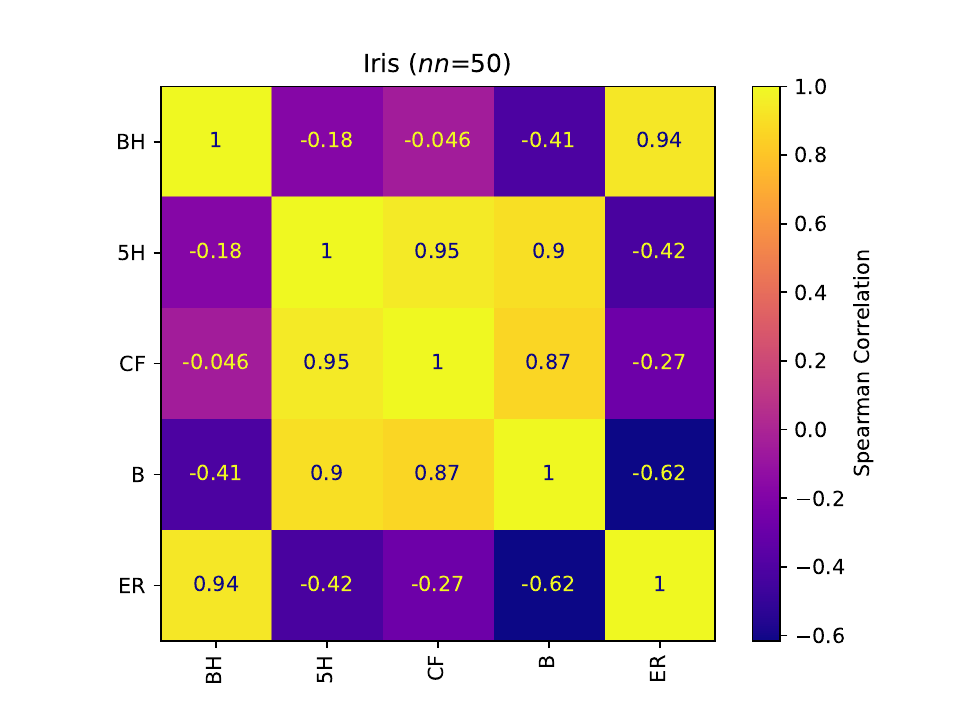}
    \end{subfigure}
    \begin{subfigure}
        \centering
        \includegraphics[width=\mywidth]{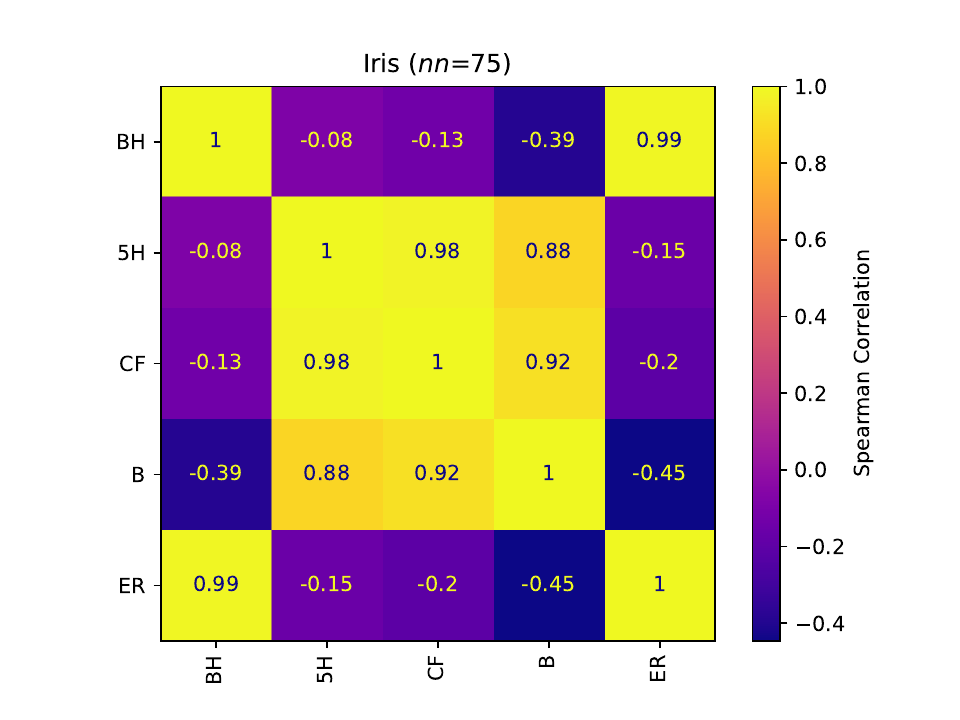}
    \end{subfigure}
    \begin{subfigure}
        \centering
        \includegraphics[width=\mywidth]{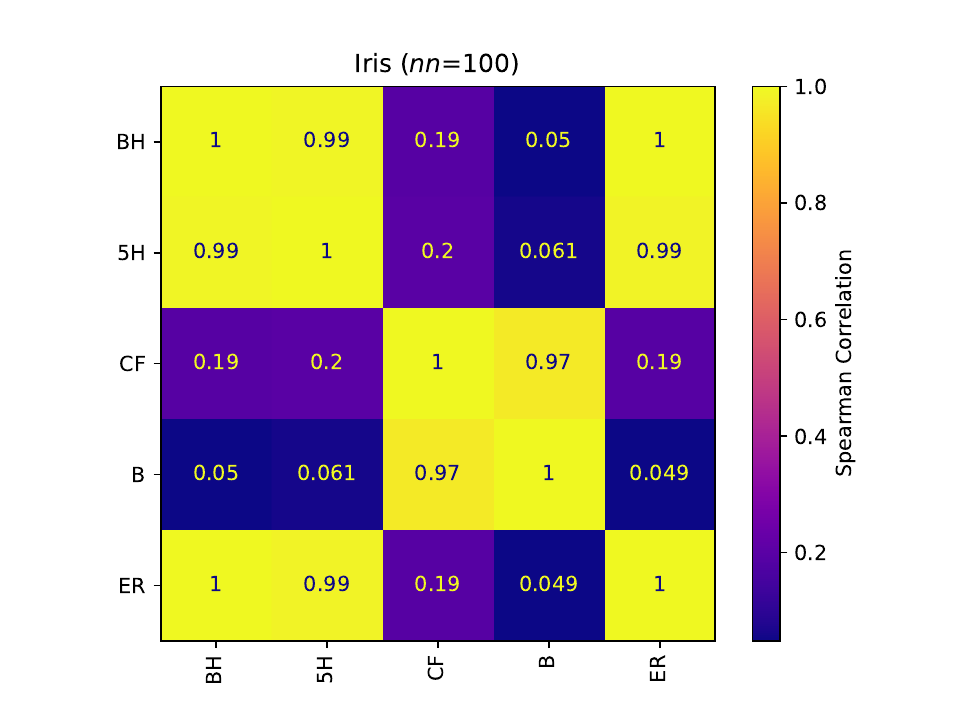}
    \end{subfigure}

    \begin{subfigure}
        \centering
        \includegraphics[width=\mywidth]{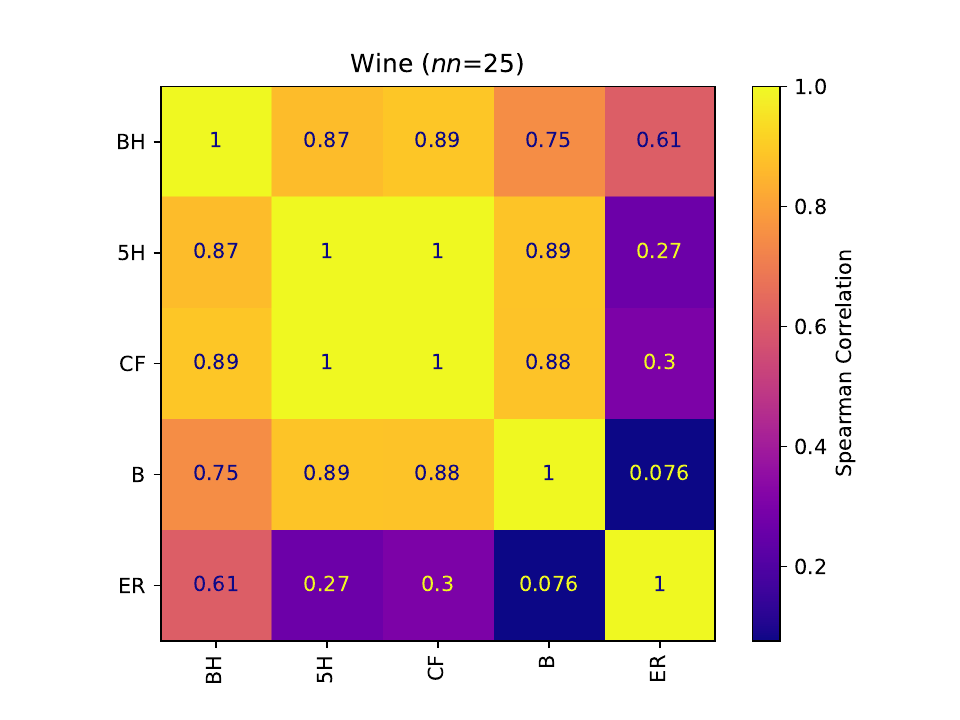}
    \end{subfigure}
    \begin{subfigure}
        \centering
        \includegraphics[width=\mywidth]{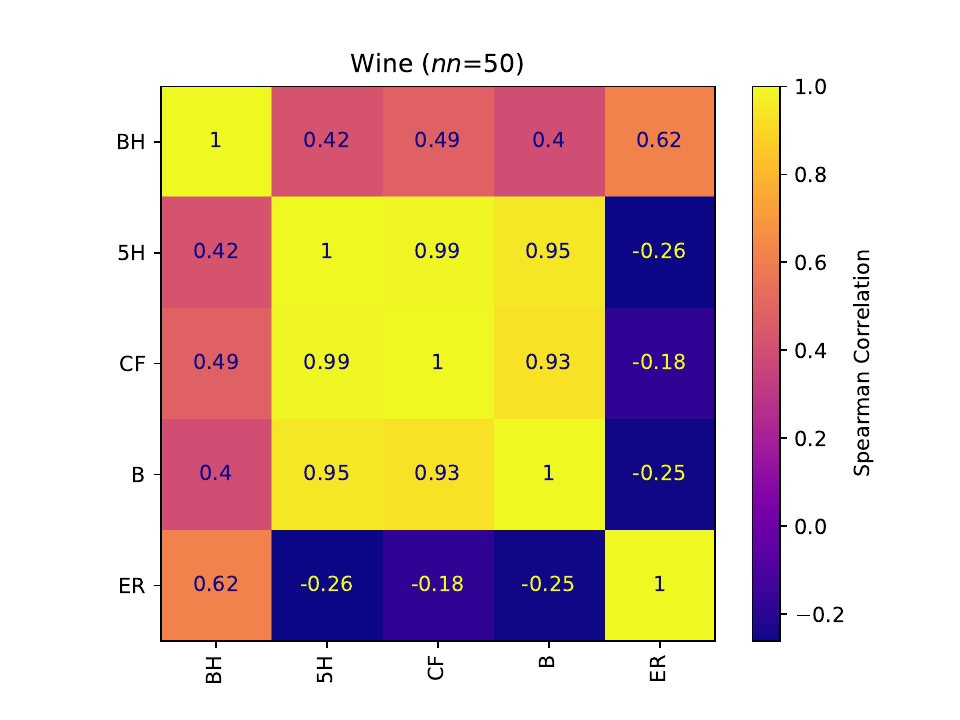}
    \end{subfigure}
    \begin{subfigure}
        \centering
        \includegraphics[width=\mywidth]{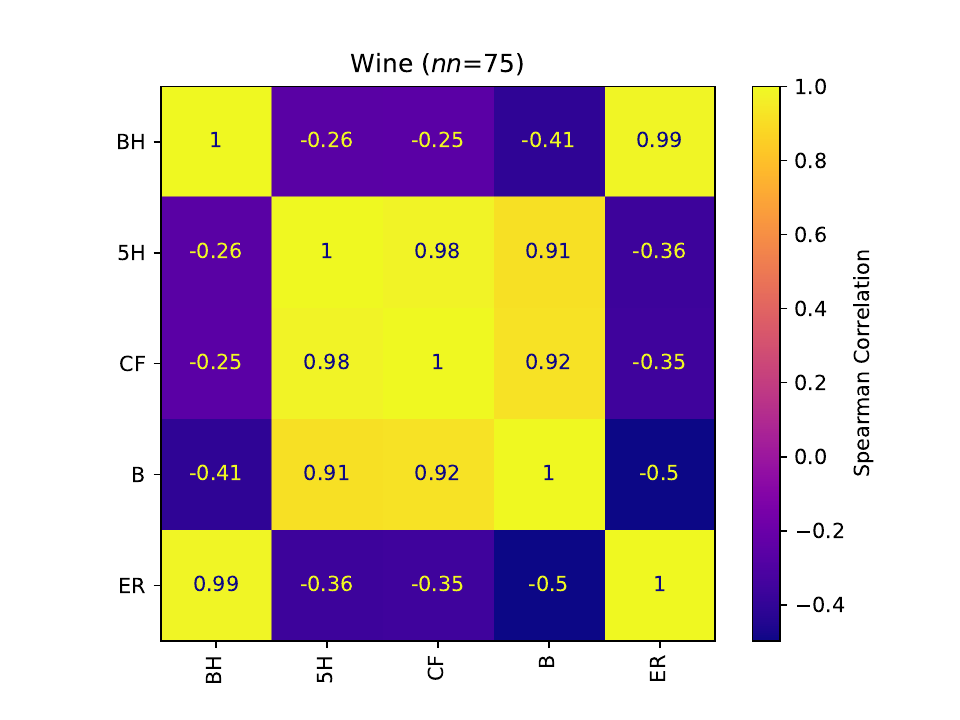}
    \end{subfigure}
    \begin{subfigure}
        \centering
        \includegraphics[width=\mywidth]{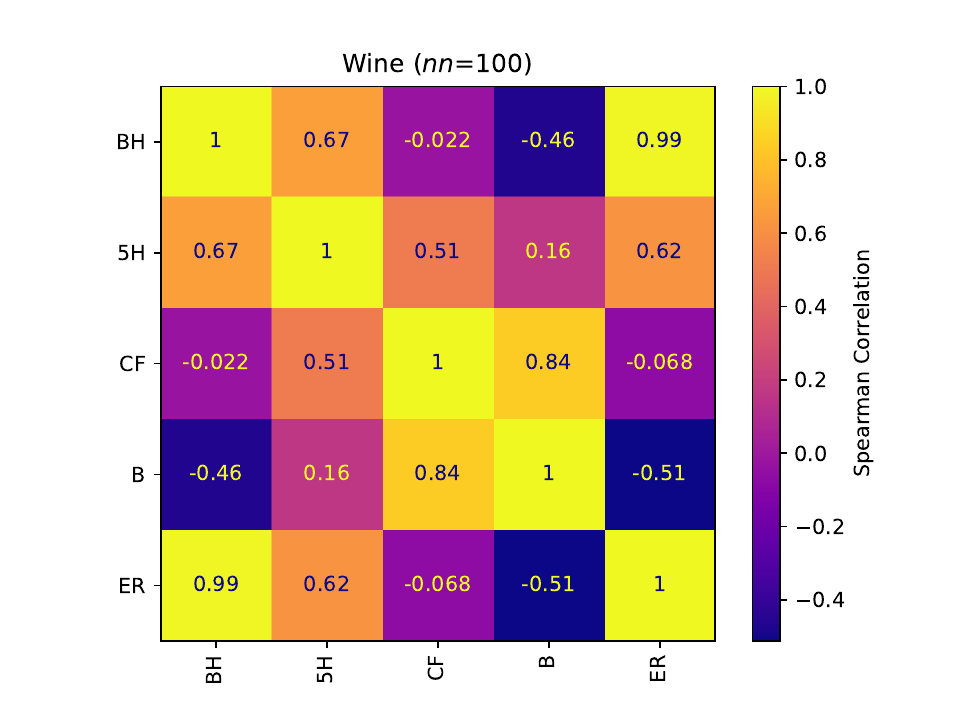}
    \end{subfigure}

    \begin{subfigure}
        \centering
        \includegraphics[width=\mywidth]{Figures/Cancer_nn25_confusion_matrix.pdf}
    \end{subfigure}
    \begin{subfigure}
        \centering
        \includegraphics[width=\mywidth]{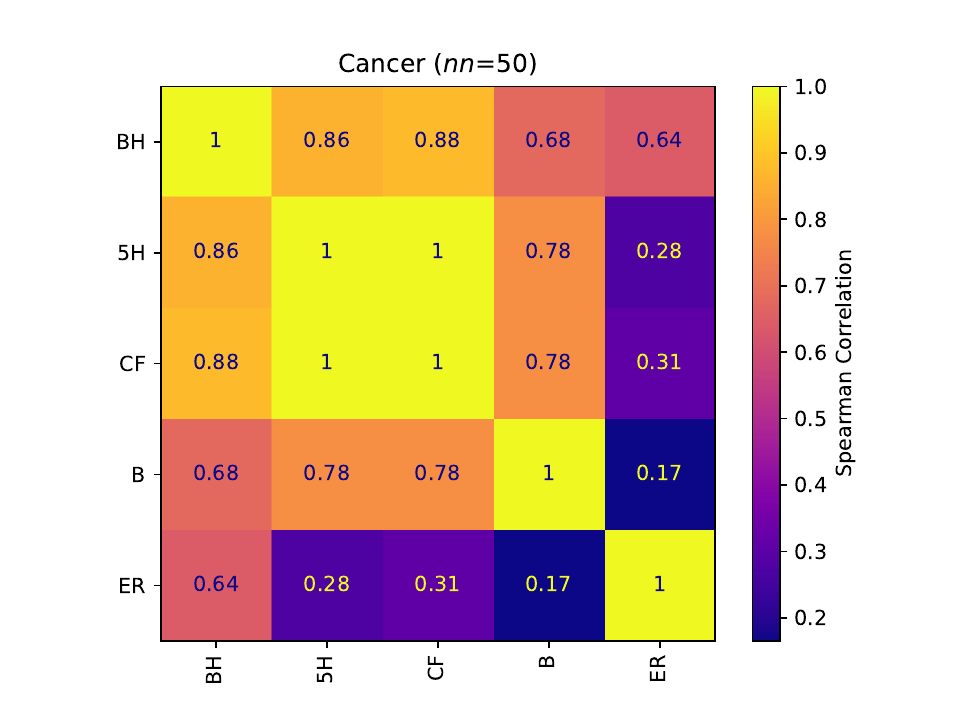}
    \end{subfigure}
    \begin{subfigure}
        \centering
        \includegraphics[width=\mywidth]{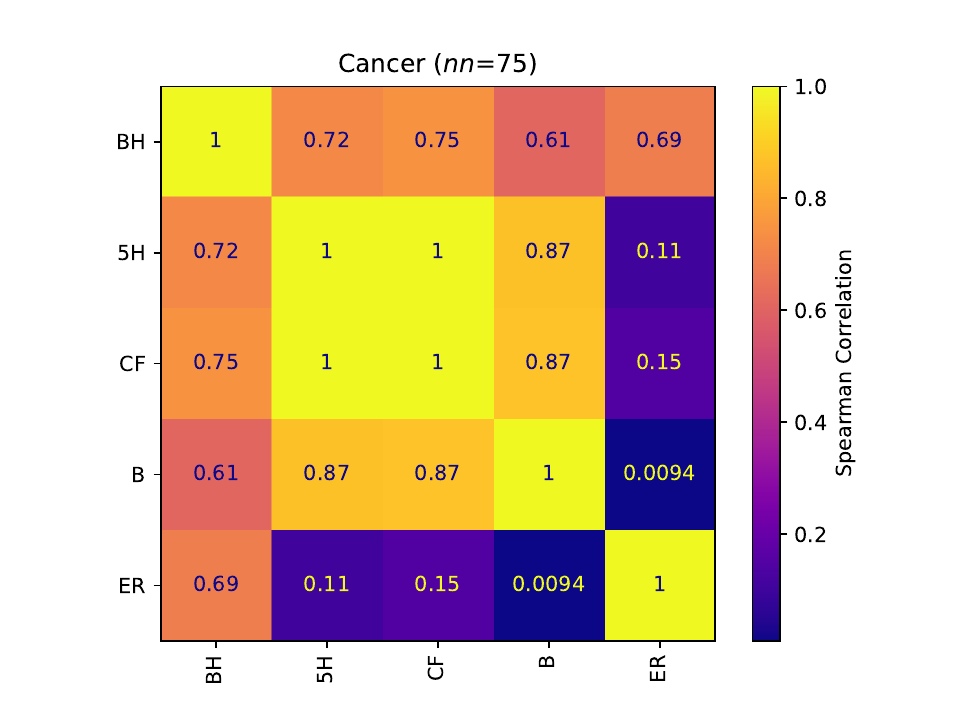}
    \end{subfigure}
    \begin{subfigure}
        \centering
        \includegraphics[width=\mywidth]{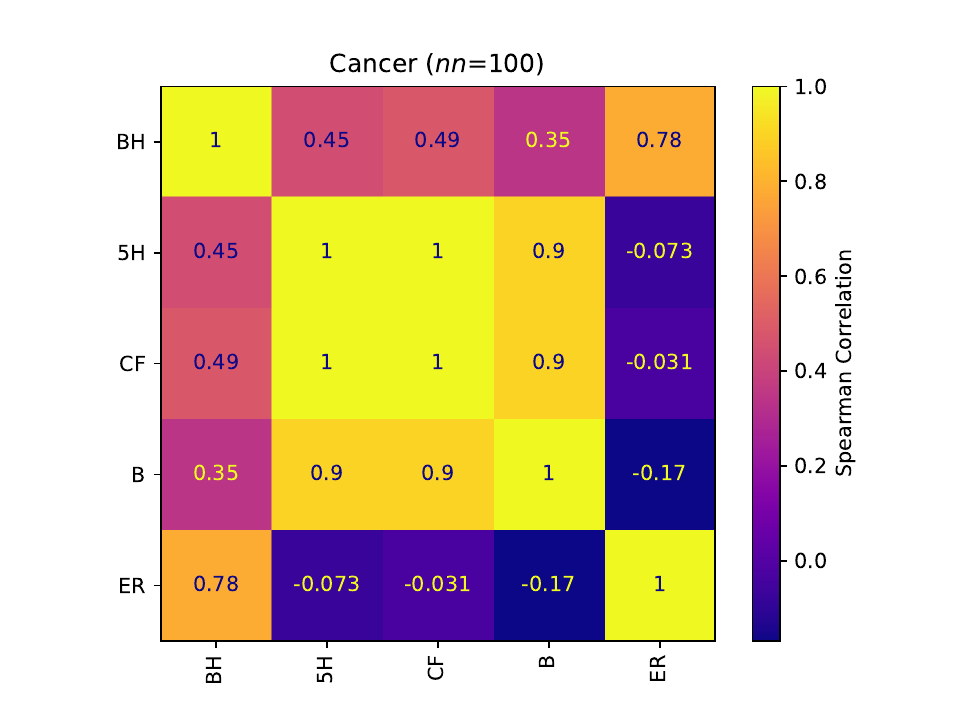}
    \end{subfigure}
    \caption{The Spearman Rank Correlation Coefficient between different centrality measures on the nearest neighbor graphs of various datasets. Here BH is biharmonic distance, CF is current-flow centrality, B is betweenness centrality, ER is effective resistance, and 5H is 5-harmonic distance.}
\end{figure}

\begin{figure}[H]
    \centering
    \begin{subfigure}
        \centering
        \includegraphics[width=\mywidth]{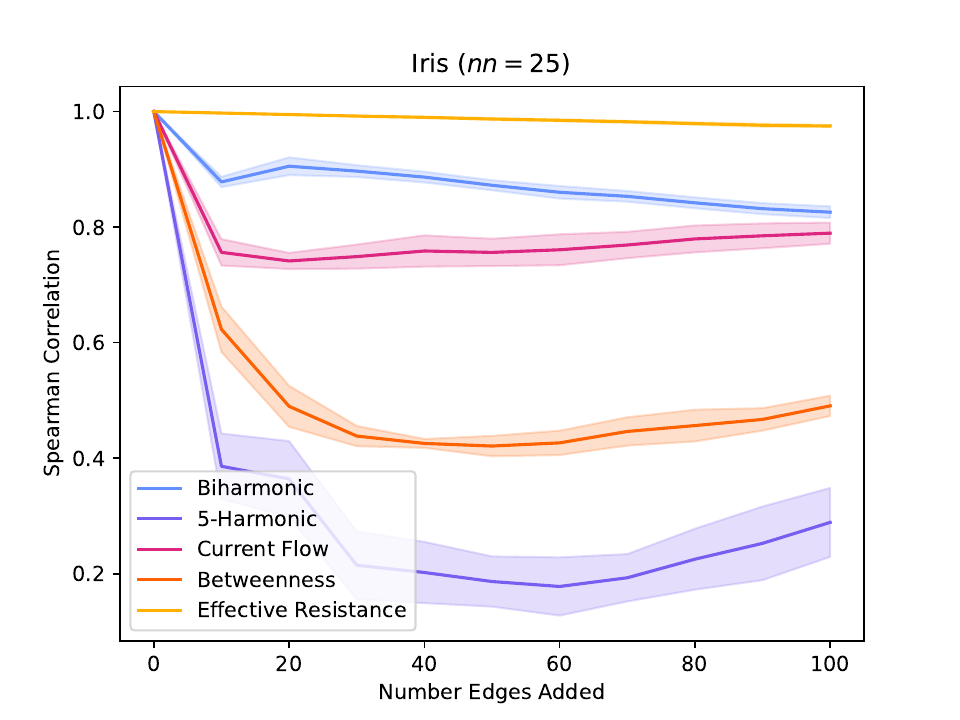}
    \end{subfigure}
    \begin{subfigure}
        \centering
        \includegraphics[width=\mywidth]{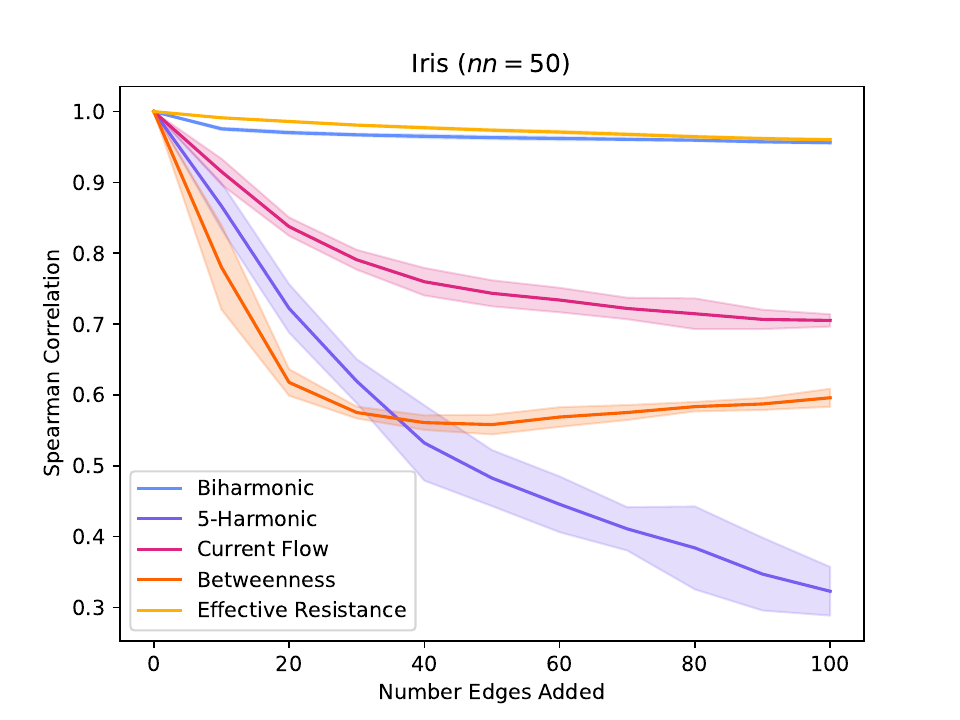}
    \end{subfigure}
    \begin{subfigure}
        \centering
        \includegraphics[width=\mywidth]{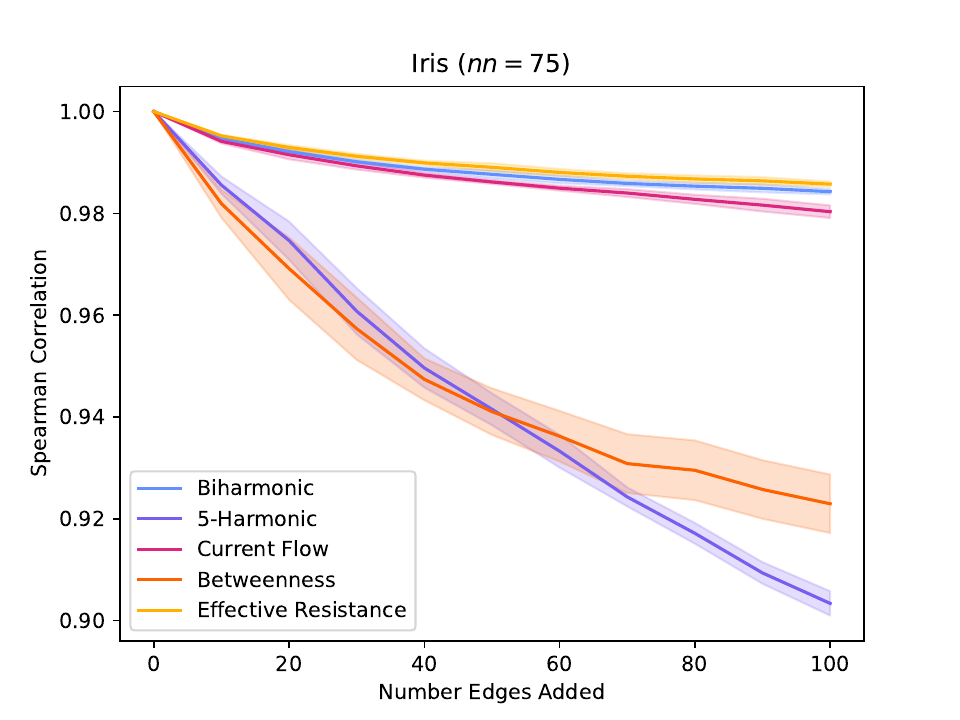}
    \end{subfigure}
    \begin{subfigure}
        \centering
        \includegraphics[width=\mywidth]{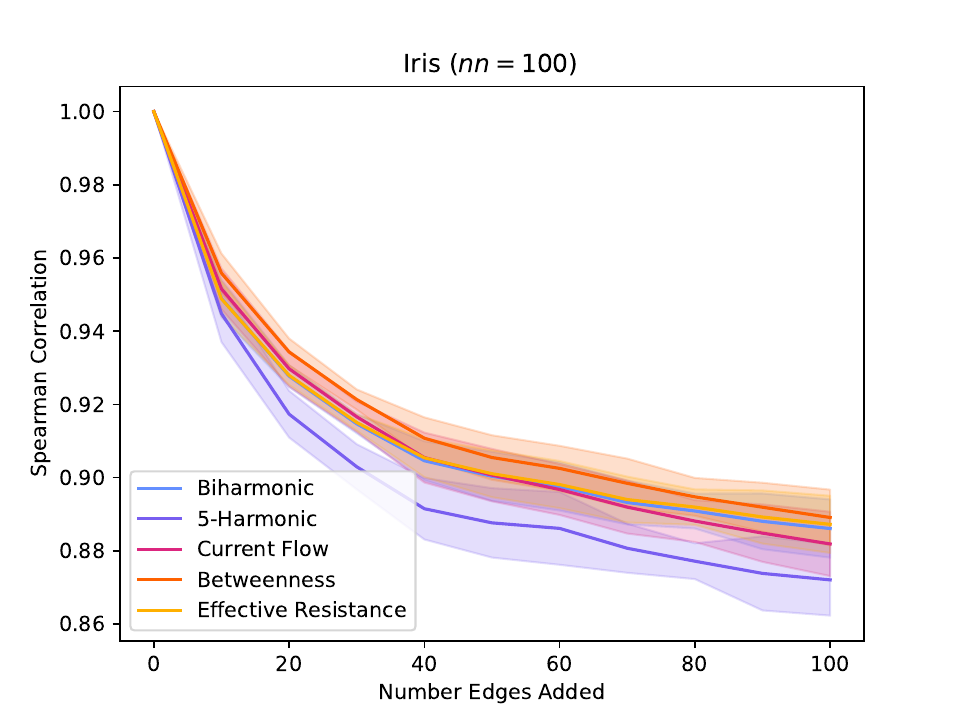}
    \end{subfigure}

    \begin{subfigure}
        \centering
        \includegraphics[width=\mywidth]{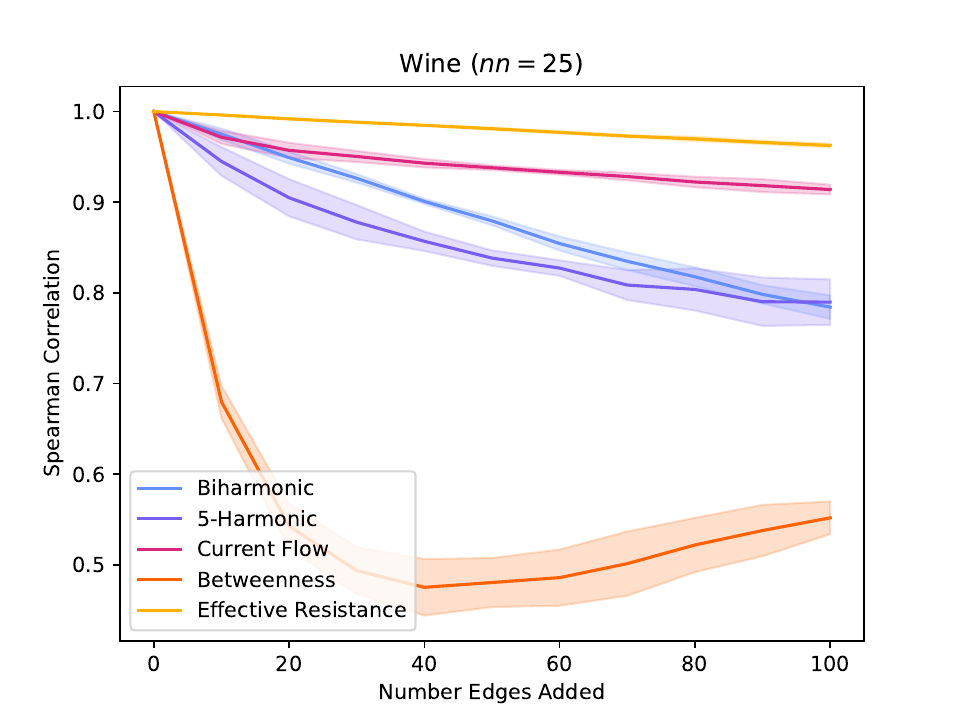}
    \end{subfigure}
    \begin{subfigure}
        \centering
        \includegraphics[width=\mywidth]{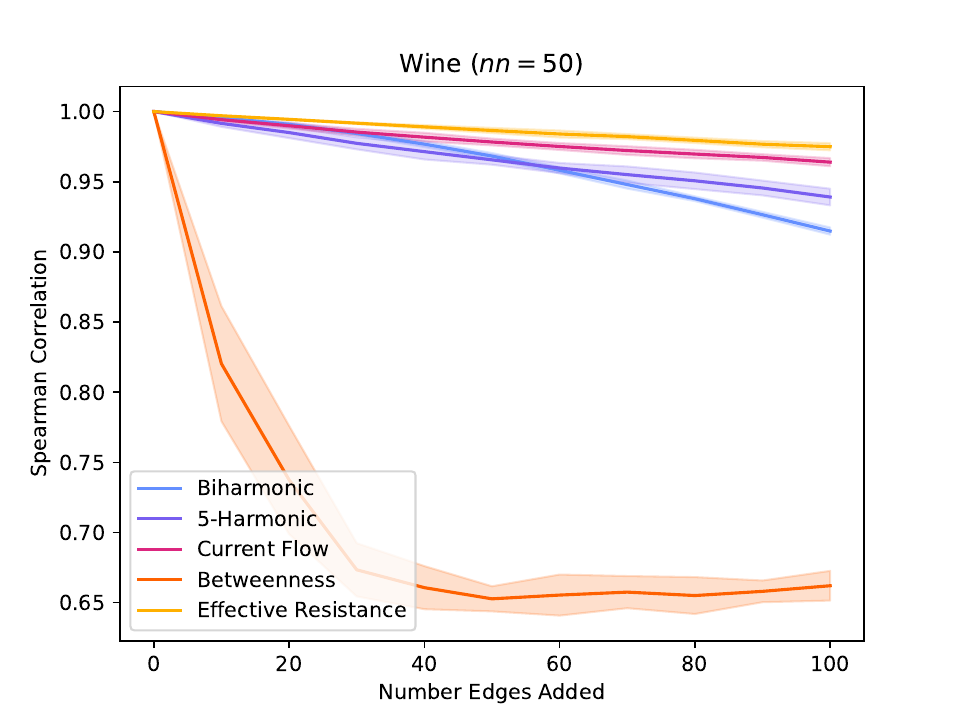}
    \end{subfigure}
    \begin{subfigure}
        \centering
        \includegraphics[width=\mywidth]{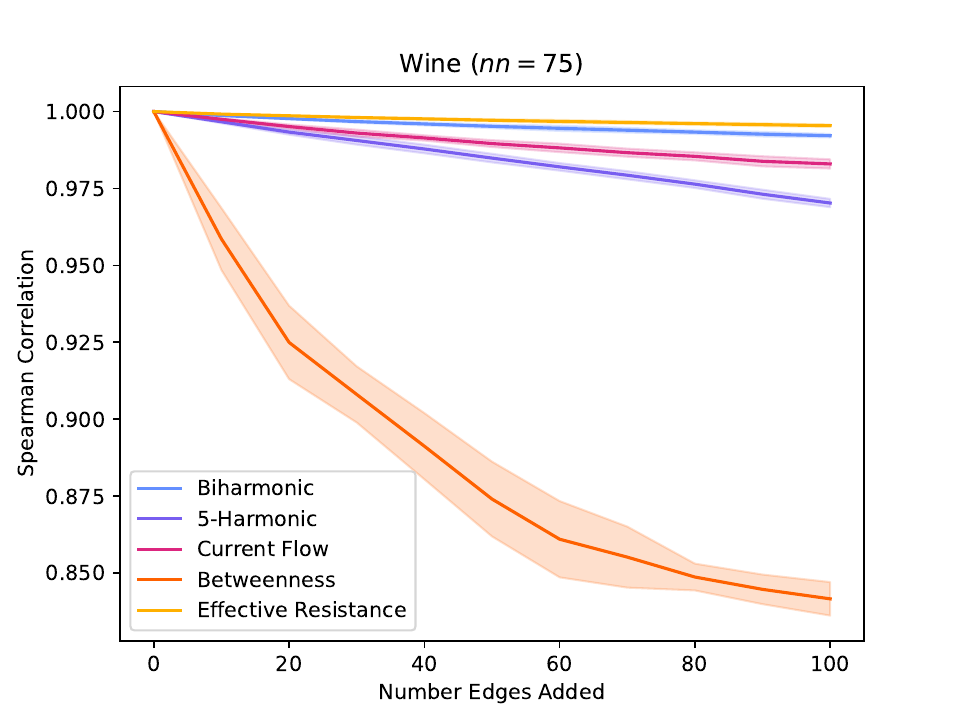}
    \end{subfigure}
    \begin{subfigure}
        \centering
        \includegraphics[width=\mywidth]{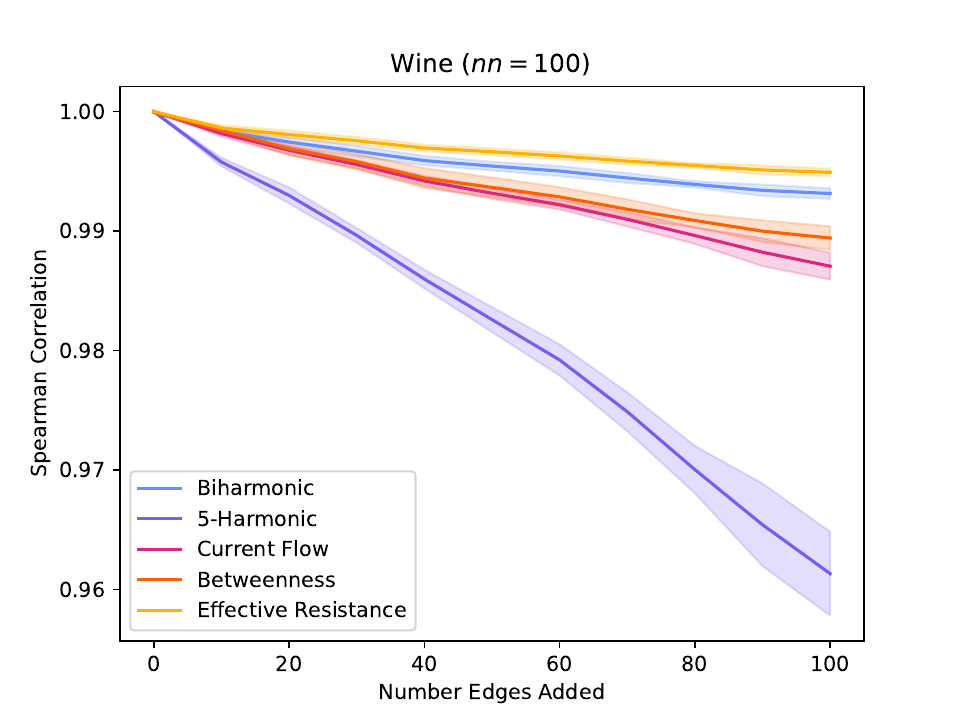}
    \end{subfigure}

    \begin{subfigure}
        \centering
        \includegraphics[width=\mywidth]{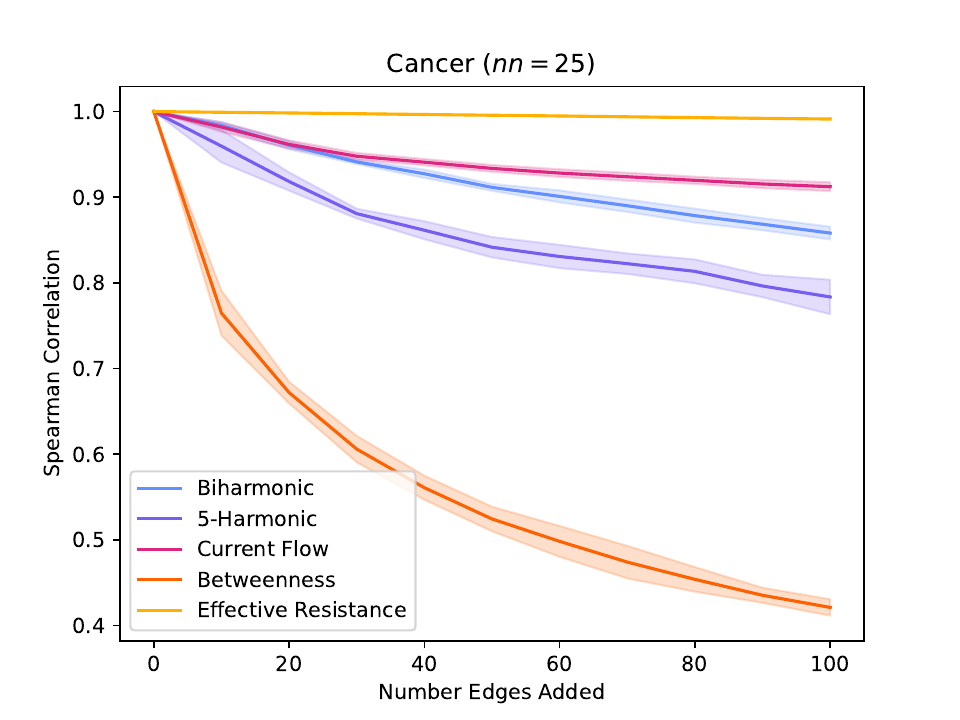}
    \end{subfigure}
    \begin{subfigure}
        \centering
        \includegraphics[width=\mywidth]{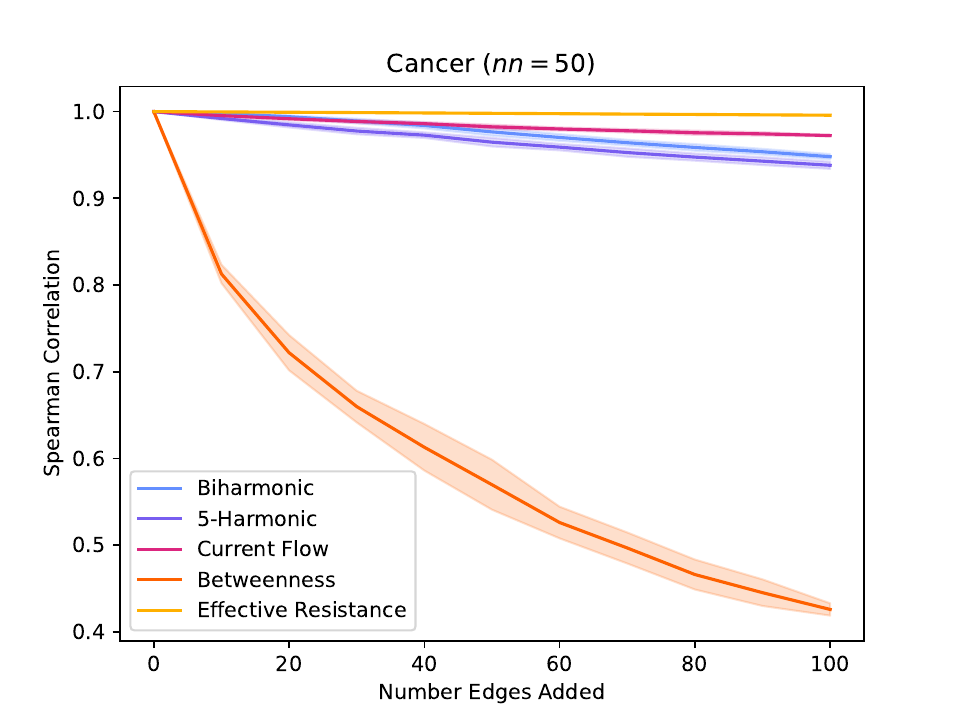}
    \end{subfigure}
    \begin{subfigure}
        \centering
        \includegraphics[width=\mywidth]{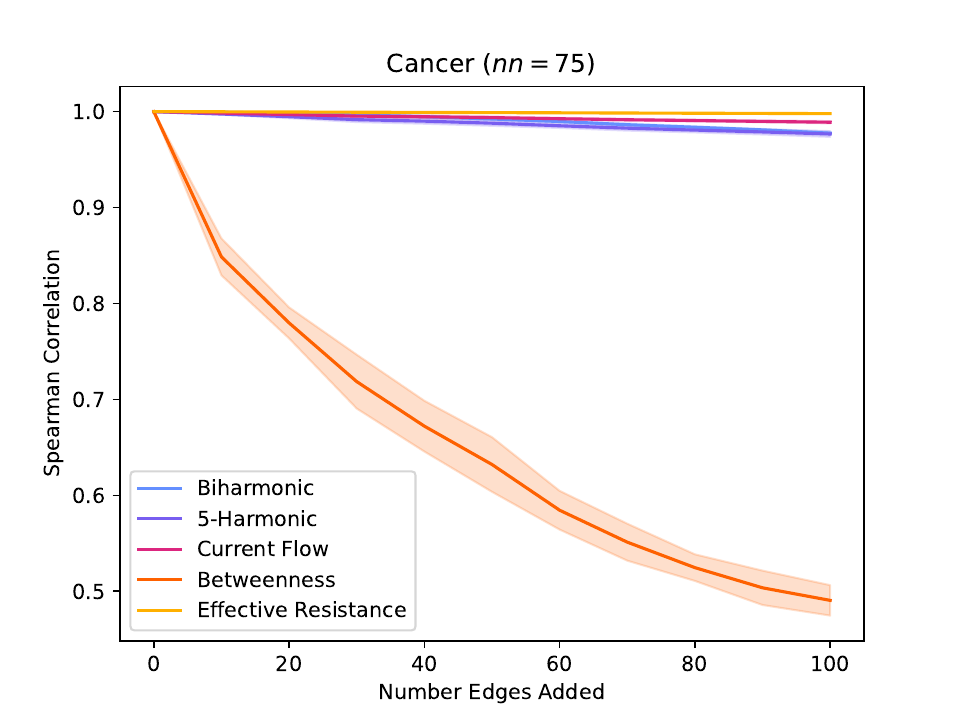}
    \end{subfigure}
    \begin{subfigure}
        \centering
        \includegraphics[width=\mywidth]{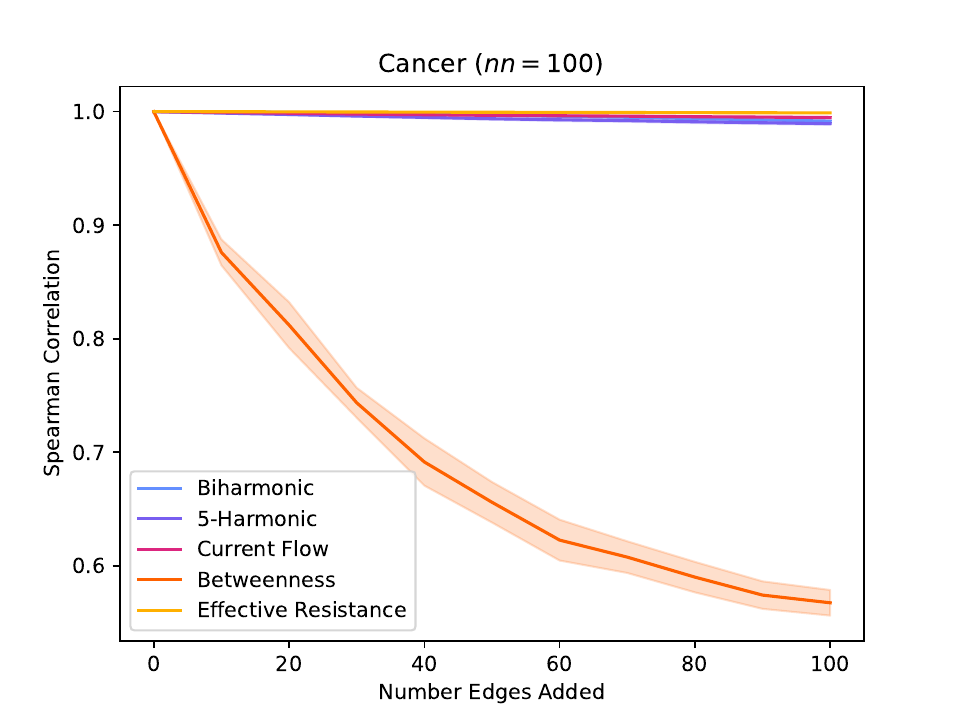}
    \end{subfigure}
    \caption{The Spearman Rank Correlation Correlation between an edge centrality measure and the same edge centrality measure after a number of random edges were added. Experiments were repeated 5 times for each centrality measure.}    
\end{figure}

\section{Results from Clustering Experiments}
\label{apx:clustering_experiments}

\vspace{-0.75cm}

\begin{table*}[ht]
    \centering
    \begin{tabular}{|p{0.1\linewidth}|p{0.15\linewidth}|p{0.15\linewidth}|p{0.15\linewidth}|p{0.15\linewidth}|}
        \hline
        \multicolumn{5}{|c|}{\textbf{Vertex Classification: Purity $(\uparrow)$}} \\
        \hline
        & \textbf{Iris} (nn=25) & \textbf{Iris} (nn=50)  & \textbf{Iris} (nn=75) & \textbf{Iris} (nn=100) \\
        \hline
        \multicolumn{5}{|c|}{\textbf{$k$-harmonic Girvan-Newman}} \\
        \hline
        $k=0.5$ & 0.673 & 0.347 & 0.347 & 0.347 \\
        $k=1$ & 0.673 & 0.673 & 0.347 & 0.347 \\
        $k=2$ & 0.673 & 0.673 & 0.347 & 0.347 \\
        $k=5$ & \first{0.907} & 0.673 & 0.347 & 0.347 \\
        $k=10$ & \first{0.907} & \second{0.907} & 0.667 & 0.347 \\
        $k=20$ & \first{0.907} & \first{0.913} & 0.667 & 0.347 \\
        \hline
        \multicolumn{5}{|c|}{\textbf{$k$-harmonic $k$-means}} \\
        \hline
        $k=0.5$ & 0.723 $\pm$ 0.119 & 0.541 $\pm$ 0.117 & 0.378 $\pm$ 0.072 & 0.417 $\pm$ 0.122 \\
        $k=1$ & 0.765 $\pm$ 0.087 & 0.572 $\pm$ 0.146 & 0.379 $\pm$ 0.072 & 0.421 $\pm$ 0.086 \\
        $k=2$ & 0.825 $\pm$ 0.098  & 0.740 $\pm$ 0.080 & 0.379 $\pm$ 0.072 & 0.477 $\pm$ 0.152 \\
        $k=5$ & 0.758 $\pm$ 0.069 & 0.833 $\pm$ 0.060 & \first{0.833 $\pm$ 0.000} & 0.437 $\pm$ 0.098 \\
        $k=10$ & 0.816 $\pm$ 0.056 & 0.793 $\pm$ 0.073 & 0.737 $\pm$ 0.059 & \third{0.873 $\pm$ 0.090} \\
        $k=20$ & 0.815 $\pm$ 0.056 & 0.793 $\pm$ 0.074 & \second{0.785 $\pm$ 0.055} & \first{0.933 $\pm$ 0.000} \\
        \hline
        \multicolumn{5}{|c|}{\textbf{Low-Rank $k$-harmonic $k$-means}} \\
        \hline
        $k=0.5$ & 0.831 $\pm$ 0.081 & 0.884 $\pm$ 0.055 & 0.713 $\pm$ 0.046 & 0.700 $\pm$ 0.058 \\
        $k=1$ & 0.791 $\pm$ 0.094 & 0.845 $\pm$ 0.065 & 0.715 $\pm$ 0.055 & 0.722 $\pm$ 0.034 \\ 
        $k=2$ & \third{0.838 $\pm$ 0.007} & 0.813 $\pm$ 0.069 & \second{0.785 $\pm$ 0.055} & \third{0.876 $\pm$ 0.091} \\
        $k=5$ & 0.820 $\pm$ 0.077 & 0.813 $\pm$ 0.069 & \third{0.769 $\pm$ 0.059} & \second{0.906 $\pm$ 0.061} \\
        $k=10$ & 0.806 $\pm$ 0.070  & 0.793 $\pm$ 0.073 & 0.753 $\pm$ 0.060 & \first{0.933 $\pm$ 0.000} \\
        $k=20$ & \second{0.859 $\pm$ 0.051} & 0.812 $\pm$ 0.068 & 0.721 $\pm$ 0.055 & \first{0.933 $\pm$ 0.000} \\
        \hline
        \multicolumn{5}{|c|}{\textbf{Girvan-Newman}} \\
        \hline
        - & \first{0.907} & \third{0.873} & 0.667 & 0.346  \\
        \hline
        \multicolumn{5}{|c|}{\textbf{Spectral Clustering}} \\
        \hline
        - & 0.704 $\pm$ 0.051 & 0.835 $\pm$ 0.089 & 0.647 $\pm$ 0.028  & 0.661 $\pm$ 0.064 \\
        \hline
    \end{tabular}
    \caption{Results of clustering experiments on nearest neighbor graphs of the Iris dataset~\cite{misc_iris_53}, evaluated using cluster purity. Averages were taken over ten trials for all algorithms using k-Means, provided with a 95\% confidence interval. The Girvan-Newman algorithm is deterministic, so there is no confidence interval for these results. \first{First}, \second{second}, and \third{third} best results for each graph are colored.}
    \label{table:iris_experimental_results}
\end{table*}

\begin{table*}[ht]
    \centering
    \begin{tabular}{|p{0.1\linewidth}|p{0.15\linewidth}|p{0.15\linewidth}|p{0.15\linewidth}|p{0.15\linewidth}|}
        \hline
        \multicolumn{5}{|c|}{\textbf{Vertex Classification: Purity $(\uparrow)$}} \\
        \hline
        & \textbf{Cancer} (nn=25) & \textbf{Cancer} (nn=50)  & \textbf{Cancer} (nn=75) & \textbf{Cancer} (nn=100)  \\
        \hline
        \multicolumn{5}{|c|}{\textbf{$k$-harmonic Girvan-Newman}} \\
        \hline
        $k=0.5$ & 0.629 & 0.629 & 0.627 & 0.629 \\
        $k=1$  & 0.627 & 0.629 & 0.627 & 0.627 \\
        $k=2$  & \second{0.822} & \second{0.822} & 0.807  & 0.627 \\
        $k=5$  & \second{0.822} & \second{0.822} & 0.801  & 0.779 \\
        $k=10$  & \second{0.822} & \second{0.822} & 0.801  & 0.787 \\
        $k=20$  & \second{0.822}  & \second{0.822} & 0.801 & 0.787 \\
        \hline
        \multicolumn{5}{|c|}{\textbf{$k$-harmonic $k$-means}} \\
        \hline
        $k=0.5$ & \third{0.774 $\pm$ 0.071}  & 0.739 $\pm$ 0.082 & 0.716 $\pm$ 0.082  & 0.637 $\pm$ 0.019 \\
        $k=1$ & \second{0.822 $\pm$ 0.000}  & 0.803 $\pm$ 0.044 & 0.791 $\pm$ 0.041  & 0.767 $\pm$ 0.053 \\
        $k=2$  & \second{0.822 $\pm$ 0.000} &  \second{0.822 $\pm$ 0.001} & 0.807 $\pm$ 0.000 &  0.803 $\pm$ 0.000 \\
        $k=5$  & \second{0.822 $\pm$ 0.000} &  \second{0.822 $\pm$ 0.000} & \second{0.812 $\pm$ 0.002} & 0.808 $\pm$ 0.004 \\
        $k=10$ & \second{0.822 $\pm$ 0.000} & \second{0.822 $\pm$ 0.000} & \second{0.812 $\pm$ 0.002} & 0.808 $\pm$ 0.004 \\
        $k=20$ & \second{0.822 $\pm$ 0.000} & \second{0.822 $\pm$ 0.000} & \second{0.812 $\pm$ 0.002} & \second{0.811 $\pm$ 0.004} \\
        \hline
        \multicolumn{5}{|c|}{\textbf{Low-Rank $k$-harmonic $k$-means}} \\
        \hline
        $k=0.5$ & \second{0.822 $\pm$ 0.000} & \third{0.821 $\pm$ 0.000} & 0.805 $\pm$ 0.002  & 0.809 $\pm$ 0.005 \\
        $k=1$  & \second{0.822 $\pm$ 0.000} & \third{0.821 $\pm$ 0.001} & 0.807 $\pm$ 0.001  & \third{0.810 $\pm$ 0.004} \\
        $k=2$  & \second{0.822 $\pm$ 0.000} & \second{0.822 $\pm$ 0.000} & \third{0.811 $\pm$ 0.001} & \third{0.810 $\pm$ 0.004} \\
        $k=5$  & \second{0.822 $\pm$ 0.000} & \second{0.822 $\pm$ 0.000} & \second{0.812 $\pm$ 0.002}  & 0.809 $\pm$ 0.005 \\
        $k=10$ & \second{0.822 $\pm$ 0.000} & \second{0.822 $\pm$ 0.000} & \first{0.813 $\pm$ 0.002} & 0.809 $\pm$ 0.005 \\
        $k=20$ & \second{0.822 $\pm$ 0.000} & \second{0.822 $\pm$ 0.000} & \second{0.812 $\pm$ 0.002} & 0.808 $\pm$ 0.005 \\
        \hline
        \multicolumn{5}{|c|}{\textbf{Girvan-Newman}} \\
        \hline
        - & \first{0.824} & \first{0.886} & 0.747  & \first{0.840} \\ 
        \hline
        \multicolumn{5}{|c|}{\textbf{Spectral Clustering}} \\
        \hline
        - &  0.715 $\pm$ 0.077 &  0.741 $\pm$ 0.014 & 0.702 $\pm$ 0.133 & 0.763 $\pm$ 0.078 \\
        \hline
    \end{tabular}
    \caption{Results of clustering experiments on nearest neighbor graphs of the Cancer dataset~\cite{misc_breast_cancer_14}, evaluated using cluster purity. Averages were taken over ten trials for all algorithms using k-Means, provided with a 95\% confidence interval. The Girvan-Newman algorithm is deterministic, so there is no confidence interval for these results. \first{First}, \second{second}, and \third{third} best results for each graph are colored.}
    \label{table:cancer_experimental_results}
\end{table*}

\begin{table*}[ht]
    \centering
    \begin{tabular}{|p{0.1\linewidth}|p{0.15\linewidth}|p{0.15\linewidth}|p{0.15        \linewidth}|p{0.15\linewidth}|}
        \hline
        \multicolumn{5}{|c|}{\textbf{Vertex Classification: Purity $(\uparrow)$}} \\
        \hline
        & \textbf{Wine} (nn=25) & \textbf{Wine} (nn=50)  & \textbf{Wine} (nn=75) & \textbf{Wine} (nn=100) \\
        \hline
        \multicolumn{5}{|c|}{\textbf{$k$-harmonic Girvan-Newman}} \\
        \hline
        $k=0.5$ & 0.404 & 0.399 & 0.404 & 0.410   \\
        $k=1$ & 0.410 & 0.399 & 0.404 & 0.410 \\
        $k=2$ & \second{0.708} & 0.607 & 0.404  & 0.410  \\
        $k=5$ & \second{0.708} & \second{0.730} & 0.680 & 0.680 \\
        $k=10$ & \second{0.708} & \second{0.730} & 0.680 & 0.680  \\
        $k=20$ & \second{0.708} &  \second{0.730} & 0.680 & 0.680 \\
        \hline
        \multicolumn{5}{|c|}{\textbf{$k$-harmonic $k$-means}} \\
        \hline
        $k=0.5$ & 0.656 $\pm$ 0.067 & 0.537 $\pm$ 0.085 & 0.423 $\pm$ 0.042  & 0.443 $\pm$ 0.060  \\ 
        $k=1$ & \third{0.688 $\pm$ 0.034} & 0.626 $\pm$ 0.088 & 0.463 $\pm$ 0.082  & 0.466 $\pm$ 0.074  \\
        $k=2$ & \first{0.719 $\pm$ 0.000} & 0.699 $\pm$ 0.040 & 0.537 $\pm$ 0.104  & 0.430 $\pm$ 0.063  \\ 
        $k=5$ & \first{0.719 $\pm$ 0.000} & \third{0.725 $\pm$ 0.000} & \second{0.701 $\pm$ 0.007} & \first{0.760 $\pm$ 0.002}  \\
        $k=10$ & \first{0.719 $\pm$ 0.000} & \third{0.725 $\pm$ 0.000} & \third{0.699 $\pm$ 0.006} & \third{0.715 $\pm$ 0.002}  \\ 
        $k=20$ & \first{0.719 $\pm$ 0.000}  & \third{0.725 $\pm$ 0.000} & 0.697 $\pm$ 0.006  & \third{0.715 $\pm$ 0.002} \\
        \hline
        \multicolumn{5}{|c|}{\textbf{Low-Rank $k$-harmonic $k$-means}} \\
        \hline
        $k=0.5$ & \first{0.719 $\pm$ 0.000} & 0.724 $\pm$ 0.004 & 0.696 $\pm$ 0.015  & 0.670 $\pm$ 0.021  \\
        $k=1$ & \first{0.719 $\pm$ 0.000} & \first{0.733 $\pm$ 0.002} & \third{0.699 $\pm$ 0.012} & 0.694 $\pm$ 0.021 \\
        $k=2$ & \first{0.719 $\pm$ 0.000} & \third{0.725 $\pm$ 0.000} & 0.698 $\pm$ 0.005  & \third{0.715 $\pm$ 0.002}  \\ 
        $k=5$ & \first{0.719 $\pm$ 0.000} & \third{0.725 $\pm$ 0.000} & \first{0.703 $\pm$ 0.004}  & \second{0.716 $\pm$ 0.002} \\ 
        $k=10$ & \first{0.719 $\pm$ 0.000} & \third{0.725 $\pm$ 0.000} & \second{0.701 $\pm$ 0.007} & \third{0.715 $\pm$ 0.002}  \\
        $k=20$ & \first{0.719 $\pm$ 0.000} & \third{0.725 $\pm$ 0.000} & 0.697 $\pm$ 0.005 & \third{0.715 $\pm$ 0.002} \\
        \hline
        \multicolumn{5}{|c|}{\textbf{Girvan-Newman}} \\
        \hline
        - & \second{0.708} & 0.719 & 0.685  & 0.685  \\ 
        \hline
        \multicolumn{5}{|c|}{\textbf{Spectral Clustering}} \\
        \hline
        - & 0.642 $\pm$ 0.043 & 0.647 $\pm$ 0.028  & 0.661 $\pm$ 0.064 & 0.625 $\pm$ 0.030 \\
        \hline
    \end{tabular}
    \caption{Results of clustering experiments on nearest neighbor graphs of the Wine dataset~\cite{misc_wine_109}, evaluated using cluster purity. Averages were taken over ten trials for all algorithms using k-Means, provided with a 95\% confidence interval. The Girvan-Newman algorithm is deterministic, so there is no confidence interval for these results. \first{First}, \second{second}, and \third{third} best results for each graph are colored.}
    \label{table:wine_experimental_results}
\end{table*}

\begin{table*}[th]
    \centering
    \begin{tabular}{|p{0.1\linewidth}|p{0.20\linewidth}|p{0.20\linewidth}|}
        \hline
        \multicolumn{3}{|c|}{\textbf{Vertex Classification: Purity $(\uparrow)$}} \\
        \hline
        & \textbf{Synthetic} (c=3,p=0.6,q=0.2) & \textbf{Synthetic} (c=5,p=0.7,q=0.15) \\
        \hline
        \multicolumn{3}{|c|}{\textbf{Girvan-Newman}} \\
        \hline
        - & 0.347 & 0.800 \\ 
        \hline
        \multicolumn{3}{|c|}{\textbf{$k$-harmonic Girvan-Newman}} \\
        \hline
        $k=0.5$ & 0.347 & 0.216  \\
        $k=1$ & 0.347 & 0.216 \\
        $k=2$ & 0.347 & 0.216 \\
        $k=5$ & 0.347 & 0.216 \\
        $k=10$ & 0.347 & 0.216 \\
        $k=20$ & 0.347 & 0.216 \\
        \hline
        \multicolumn{3}{|c|}{\textbf{$k$-harmonic $k$-means}} \\
        \hline
        $k=0.5$ & 0.426 $\pm$ 0.074 & 0.282 $\pm$ 0.043  \\ 
        $k=1$ & 0.471 $\pm$ 0.107 & 0.304 $\pm$ 0.059 \\
        $k=2$ & 0.471 $\pm$ 0.107 & 0.327 $\pm$ 0.097 \\ 
        $k=5$ & 0.669 $\pm$ 0.156 & 0.567 $\pm$ 0.129 \\
        $k=10$ & 0.900 $\pm$ 0.114 & \second{0.940 $\pm$ 0.069} \\ 
        $k=20$ & \third{0.926 $\pm$ 0.098} & 0.880 $\pm$ 0.073 \\
        \hline
        \multicolumn{3}{|c|}{\textbf{Low Rank $k$-harmonic $k$-means}} \\
        \hline
        $k=0.5$ & 0.903 $\pm$ 0.112 & 0.880 $\pm$ 0.100 \\
        $k=1$ &  0.867 $\pm$ 0.123 & \first{0.960 $\pm$ 0.060} \\
        $k=2$ &  0.967 $\pm$ 0.075 & 0.860 $\pm$ 0.100 \\ 
        $k=5$ &  0.867 $\pm$ 0.123 & 0.860 $\pm$ 0.118 \\ 
        $k=10$ & \first{0.988 $\pm$ 0.002} & \second{0.940 $\pm$ 0.070} \\
        $k=20$ & \second{0.960 $\pm$ 0.000} & \first{0.960 $\pm$ 0.060} \\
        \hline
        \multicolumn{3}{|c|}{\textbf{Spectral Clustering}} \\
        \hline
        -& 0.893 $\pm$ 0.109 & \third{0.899 $\pm$ 0.102} \\
        \hline
    \end{tabular}
    \caption{Results of experiments on two synthetically generated stochastic block graphs with 50 vertices per cluster. The 3-cluster graph (left) was generated with edges within clusters occurring with probability $p$=0.6 and between clusters with probability $q$=0.2. The 5-cluster graph (right) is similarly generated, but with $p$=0.7 and $q$=0.15. Results were evaluated with cluster purity, and averages were taken over ten trials for all algorithms using $k$-means, provided with a 95\% confidence interval. \first{First}, \second{second}, and \third{third} best results for each graph are colored.}
    \label{table:synthetic_experimental_results}
\end{table*}

\end{document}